\documentclass[journal]{IEEEtran}
\IEEEoverridecommandlockouts
\usepackage{amsfonts}
\usepackage{amssymb}
\usepackage[cmex10]{amsmath}
\usepackage{graphicx}
\usepackage{subfigure}
\usepackage{verbatim}
\usepackage{cite}
\usepackage{multicol}
\usepackage{diagbox}
\usepackage{tabularx,booktabs}
\usepackage{amsthm}
\usepackage{booktabs}

\usepackage{multirow}
\usepackage{xypic}
 \usepackage{enumerate}
 \usepackage{color}
 \usepackage{soul}
\usepackage[ruled,lined]{algorithm2e}
\usepackage{setspace}
\usepackage{fancyhdr}
\usepackage{algorithmic}

\begin{document}
\newcommand{\be}{\begin{equation}}
\newcommand{\ee}{\end{equation}}
\newcommand{\br}{{\mbox{\boldmath{$r$}}}}
\newcommand{\bp}{{\mbox{\boldmath{$p$}}}}
\newcommand*{\QEDA}{\hfill\ensuremath{\square}}  %×Ô¶šÒå£¬ÊµÐÄ
\newcommand{\bn}{{\mbox{\boldmath{$n$}}}}
\newcommand{\balfa}{{\mbox{\boldmath{$\alpha$}}}}
\newcommand{\ba}{\mbox{\boldmath{$a $}}}
\newcommand{\bta}{\mbox{\boldmath{$\beta $}}}
\newcommand{\bg}{\mbox{\boldmath{$g $}}}
\newcommand{\bPsi}{\mbox{\boldmath{$\Psi $}}}
\newcommand{\bsigma}{\mbox{\boldmath{ $\Sigma $}}}
\newcommand{\bpi}{\mbox{\boldmath{$\pi $}}}
\newcommand{\bPi}{\mbox{\boldmath{$\Pi $}}}
\newcommand{\bGamma}{{\bf \Gamma }}
\newcommand{\bA}{{\bf A }}
\newcommand{\bII}{{\bf I }}
\newcommand{\bP}{{\bf P }}
\newcommand{\bX}{{\bf X }}
\newcommand{\bI}{{\bf I }}
\newcommand{\bR}{{\bf R }}
\newcommand{\bff}{{\mathbf{f}}}
\newcommand{\bZ}{{\bf Z }}
\newcommand{\bz}{{\bf z }}
\newcommand{\bx}{{\mathbf{x}}}
\newcommand{\bM}{{\bf M}}
\newcommand{\bW}{{\bf W}}
\newcommand{\bU}{{\bf U}}
\newcommand{\bD}{{\bf D}}
\newcommand{\bJ}{{\bf J}}
\newcommand{\bH}{{\bf H}}
\newcommand{\bK}{{\bf K}}
\newcommand{\bm}{{\bf m}}
\newcommand{\bN}{{\bf N}}
\newcommand{\bC}{{\bf C}}
\newcommand{\bL}{{\bf L}}
\newcommand{\bF}{{\bf F}}
\newcommand{\bv}{{\bf v}}
\newcommand{\bSigma}{{\bf \Sigma}}
\newcommand{\bS}{{\bf S}}
\newcommand{\bs}{{\bf s}}
\newcommand{\bO}{{\bf O}}
\newcommand{\bQ}{{\bf Q}}
\newcommand{\bY}{{\bf Y}}
\newcommand{\by}{{\bf y}}
\newcommand{\btr}{{\mbox{\boldmath{$tr$}}}}
\newcommand{\bNSCM}{{\bf NSCM}}
\newcommand{\bpsi}{\mbox{\boldmath{$\psi $}}}
\newcommand{\barg}{{\bf arg}}
\newcommand{\bmax}{{\bf max}}
\newcommand{\test}{\mbox{$
\begin{array}{c}
\stackrel{ \stackrel{\textstyle H_1}{\textstyle >} } { \stackrel{\textstyle <}{\textstyle H_0} }
\end{array}
$}}
\newcommand{\tabincell}[2]{\begin{tabular}{@{}#1@{}}#2\end{tabular}}
\newtheorem{Def}{Definition}
\newtheorem{Pro}{Proposition}
\newtheorem{Rem}{Remark}
\newtheorem{Lam}{Lemma}
\theoremstyle{remark}\newtheorem{Exa}{Example}
\theoremstyle{plain}\newtheorem{Cor}{Corollary}
\newcommand{\proproof}[1]{\noindent\textbf{Proof. } See Appendix #1.}

\title{Robust Distributed Fusion with Labeled Random Finite Sets}% 
\author{Suqi Li, Wei Yi*, Reza Hoseinnezhad,
Giorgio Battistelli, Bailu Wang, Lingjiang Kong
\thanks{This work was supported
by the Chang Jiang Scholars Program, the National Natural Science Foundation of China under Grants 61501505 and 61771110, the Fundamental Research Funds of Central Universities under Grants ZYGX2016J031, the Chinese Postdoctoral Science Foundation under Grant 2014M550465 and Special Grant 2016T90845,  and the Australian Research Council (ARC) through the Linkage Project Grant LP160101081. (\textit{Corresponding author: Wei Yi.})

S. Li, B. Wang, W. Yi,  and L. Kong  are with the School of  Electronic Engineering, University of Electronic Science and Technology of China, Chengdu 611731, China (Email: qi\_qi\_zhu1210@163.com;  kussoyi@gmail.com; w\_b\_l3020@163.com; lingjiang.kong@gmail.com).

R. Hoseinnezhad is with the School of Aerospace, Mechanical and Manufacturing Engineering, RMIT University, Victoria 3083, Australia (Email: reza.hoseinnezhad@rmit.edu.au). 

G. Battistelli is  with the Dipartimento di Ingegneria dell' Informazione (DINFO), Universit$\grave{\mbox{a}}$ degli Studi di Firenze, Via Santa Marta 3, 50139, Firenze, Italy (Email: giorgio.battistelli@unifi.it).  

%This paper has supplementary downloadable material available at http://ieeexplore.ieee.org., provided by the author. The material includes mathematical proofs of Propositions 2, 3 and 4. Contact kussoyi@gmail.com for further questions about this work.
}
}

\maketitle

%(\emph{Wei Yi is the corresponding author, Email:  kussoyi@gmail.com}).

\begin{abstract}
This paper considers the problem of the distributed fusion of multi-object posteriors in the labeled random finite set filtering framework, using Generalized Covariance Intersection (GCI) method. Our analysis shows that  GCI fusion with  labeled multi-object densities strongly relies on label consistencies between local multi-object posteriors at different sensor nodes, and hence suffers from a severe performance  degradation when perfect label consistencies  are violated. Moreover,  we  mathematically analyze this phenomenon  from the perspective of Principle of Minimum Discrimination Information and the so called yes-object probability.   Inspired by the analysis, we propose a novel and general  solution for the distributed fusion with labeled multi-object densities that is robust to label inconsistencies between sensors. Specifically, the labeled multi-object posteriors are firstly marginalized  to their unlabeled posteriors  which  are then fused using  GCI method. We also introduce a principled method to construct the labeled fused density and produce tracks 
formally. Based on the developed theoretical framework, we present tractable algorithms for the family of generalized labeled multi-Bernoulli (GLMB) filters including $\delta$-GLMB, marginalized $\delta$-GLMB and labeled multi-Bernoulli  filters.
%The proposed solution makes it possible to derive robust distributed fusion algorithms for labeled multi-object densities belonging to the generalized labeled multi-Bernoulli (GLMB) density family, and hence including
%$\delta$-GLMB, marginalized $\delta$-GLMB, labeled multi-Bernoulli (LMB) filters. 
%is performed fully label-wise
The robustness and efficiency of the proposed distributed fusion algorithm are demonstrated  in  challenging tracking scenarios via numerical experiments.
\end{abstract}

\section{Introduction}
\IEEEPARstart{D}{istributed} multi-sensor multi-object tracking (DMMT) solutions are generally designed to benefit from lower communication cost and higher fault tolerance than centralized methods. Devising  DMMT solutions becomes particularly challenging when the correlations between the estimates from different sensors are not known.  The optimal solution to this problem was developed in \cite{CY-Chong}, but the computational cost of calculating the common information can make the solution intractable in many real-world applications.  A number of suboptimal solutions with demonstrated tractability have been formulated based on  the Generalized Covariance Intersection (GCI)~\cite{Mahler-1,Hurley} or  the exponential mixture densities (EMD)~\cite{EMD-Julier,Clark,Uney-2} or the Kullback-Leibler average (KLA)~\cite{Battistelli,double-counting}. 

The GCI fusion rule is an extension of the Covariance Intersection method  which only utilizes the mean and covariance and is limited to Gaussian posteriors~\cite{Uhlmann}. The GCI fusion rule relaxes the Gaussian constraint, and can be used to fuse multi-object distributions with completely unknown correlations, since it intrinsically avoids any double counting of common information. Furthermore, the GCI can be computed in a distributed way by means of suitable consensus algorithms \cite{double-counting, Battistelli}. Finally, from an information-theoretic point of view, GCI fusion rule admits a meaningful interpretation that the fused density is   the centroid of the local posteriors  with Kullback-Leibler divergence considered as the distance.
 
%
%as it minimizes a combined distance (defined in terms of Kullback-Leibler divergences) from the fused density to the centroid of the local posteriors that are fused. 

Based on the GCI fusion rule and  its variants, several  DMMT algorithms have been proposed in the literature. Specifically, the distributed  fusion algorithms for probability hypothesis density~(PHD)~\cite{PHD-Vo,Ristic_PHD}, cardinalized PHD (CPHD)~\cite{Franken, Vo-CPHD} and multi-Bernoulli~(MB)~\cite{MeMBer_Vo2,MeMber_Vo3,Reza_ground_target,Reza_visual_tracking} filters  have  been explored in~\cite{Mahler-1,Clark,Uney-2, double-counting, Battistelli, Mehmet, GCI-MB,Li-tian-cheng, gmb-second}. 
The aforementioned methods are multi-object filters not \textit{trackers} in the sense that object states are estimated without labels. 
%In addition, all of those filters assume particular multi-object distributions (e.g., Poisson, i.i.d. clusters and multi-Bernoulli RFSs) that are not conjugate priors with the standard multi-object likelihood function~\cite{book_mahler}. Therefore, the first or second order moment approximation or other approximations are used in devising the Bayesian update  of the aforementioned multi-object filters.
% (they are indistinguishable)

Recently, in a series of works, the notion of  labeled random finite set (RFS) was introduced to address object trajectories and their uniqueness~\cite{GLMB,LMB_Vo2, delta_GLMB, LRFS_Bear_Vo, Fantacci-BN, GLMB_Papi_Vo,LMO-GOM}. Vo \textit{et al.}~\cite{GLMB,LMB_Vo2} proposed a particular class of labeled multi-object densities called generalized labeled multi-Bernoulli (GLMB) densities.\footnote{GLMB distribution was also named as Vo-Vo distribution by Mahler in his book~\cite{refr:tracking-2}.}  The class of  GLMB densities are  conjugate priors  with respect to the standard multi-object likelihood  and also closed under the Chapman-Kolmogorov equation in Bayesian inference. A variant of the GLMB filter called  the $\delta$-GLMB filter  can be used to multi-object tracking directly, and  not only produces 
trajectories formally but also outperforms the aforementioned filters \cite{GLMB}. Two computationally-efficient approximations of the $\delta$-GLMB filter, i.e.,  the labeled multi-Bernoulli (LMB) filter~\cite{delta_GLMB} and the marginalized $\delta$-GLMB (M$\delta$-GLMB) filter~\cite{Fantacci-BN},  have  also been  
developed.

The enhanced accuracy and superior tracking capability that are inherent in the labeled random finite set filters, have motivated the development of \textit{distributed fusion} methods that work in tandem with these multi-object tracking algorithms. Fantacci~\textit{et~al.}~\cite{Fantacci-BT} were among the first who investigated the distributed fusion of labeled multi-object densities and derived closed-form solutions for GCI fusion with M$\delta$-GLMB and  LMB densities, based on the assumption that different sensors share the same label space for the birth process. However, their work does not explore the  profound meaning of ``sharing the same label space'' and does not address the conditions implied by the assumption in practice.  

In this paper, we further investigate the distributed fusion for labeled random finite set filters. Our major contributions are as follow:
%\begin{enumerate}[i)]
%\item

\textit{i) We analyse the drawback of GCI fusion with labeled multi-object densities  by showing that the fusion performance is highly sensitive to label inconsistencies between sensor nodes. The analysis is carried out  in a principled theoretical framework by virtue of the following proposed notions:}

 -- GCI divergence,  which is  a new measure of discrepancy compatible with GCI fusion rule to quantify the degree of similarity between multiple densities. 
 
-- Conditional multi-label distribution,   which facilitates a new decomposition of the labeled multi-object density. 

-- Label inconsistency indicator, which  quantifies the inconsistency of label information embedded in multiple labeled multi-object densities.

%\item

 \textit{ii) Motivated by the aforementioned performance analysis, we propose a novel and general  solution to the distributed fusion with labeled set filters that is immune to the effect of label inconsistencies between sensor nodes.   Based on the developed theoretical framework, we also  present tractable distributed fusion algorithms for the family of GLMB filters including the $\delta$-GLMB, M$\delta$-GLMB and LMB filters}.

Extensive numerical experiments verify the robustness and effectiveness of the proposed fusion algorithm with Gaussian mixture implementation in challenging tracking scenarios. 

Preliminary results have been announced in the conference
paper \cite{GMB-fusion}. This paper presents a more complete theoretical and
numerical study. The layout of this paper is as follows. The background and notations are presented in Section \ref{chp:2}, followed by drawbacks of GCI fusion with labeled densities  discussed in Section \ref{chp:3}. Section \ref{chp:4} provides a mathematical analysis of the performance degradation of GCI fusion with labeled densities. Section \ref{chp:5} proposes a robust solution to distributed 
fusion with labeled set filters and presents tractable algorithms for the family of GLMB filters.  
Section \ref{chp:6} demonstrates the performance of the proposed algorithms via numerical examples. Conclusions are presented in Section \ref{chp:7}.\section{Notations and Background}\label{chp:2}
\subsection{Notations}
We adopt the convention that single-object states are denoted by lowercase letter ``x'', e.g.  $x,\bx$, while the multi-object states are denoted by capital letter ``X'', e.g.  $X,\bX$.  To distinguish labeled states and distributions from the unlabeled ones, bold face letters are adopted for the labeled ones, e.g.  $\bx$, $\bX$, $\bpi$. Observations generated by single-object states are denoted by $z$, and the multi-object observations are denoted by $Z$. Moreover, blackboard bold letters represent spaces, e.g.  the state space is represented by $\mathbb{X}$, the label space by $\mathbb{L}$, and the observation space  by $\mathbb{Z}$. The collection of all finite subsets of $\mathbb{X}$ is denoted by $\mathcal{F}(\mathbb{X})$. The number of elements in a set is called its cardinality, and denoted by $| \cdot |$. The set of all finite subsets of $\mathbb{X}$ with cardinality $n$ is denoted by $\mathcal{F}_n(\mathbb{X})$. 

We use the multi-object exponential notation
\begin{equation}\label{multi-object exponential notation }
  h^{X}\triangleq{\prod}_{x\in X}h(x)
\end{equation}
for any set $X$ and real-valued function $h$, with $h^\emptyset=1$ by convention. 
%To admit arbitrary arguments like sets, vectors and integers, the generalized Kronecker delta function is used, given by
%\begin{equation}\label{delta}
%  \delta_{Y}(X)\triangleq\left\{\begin{array}{l}
%\!\!1,\,\,\,\, \mbox{if $X =  Y$} \\
%\!\!0,\,\,\,\, \mbox{otherwise}
%\end{array}\right.
%\end{equation}
The inclusion function is given by
\begin{equation}\label{inclusion function}
  1_{Y}(X)\triangleq\left\{\begin{array}{l}
\!\!1, \,\,\,\,\mbox{if $X \subseteq Y$}\\
\!\!0 \,\,\,\,\mbox{otherwise}.
\end{array}\right.
\end{equation}
If $X$ is a singleton, i.e., $X=\{x\}$, the notation $1_Y(x)$ is used instead of $1_Y(\{x\})$.
%\textcolor{red}{If $\bX$ is a singleton, i.e.  $\bX=\{\bx\}$, the notation $1_\bY(x)$ is used instead of $1_\bY(\{\bx\})$.}
\subsection{Labeled Random Finite Set Distributions and Filters}
%\begin{Def}
%Given a  multi-object density $\pi$ on $\mathbb{X}$, and any positive integer n, we define the cardinality probability mass function that $X$ has $n\geqslant 0$ elements by 
%\begin{equation}
%p(n)=\int\pi(\{x_1, \cdots, x_n\})d(x_1, \cdots, x_n)
%\end{equation}
%and   the joint probability density   over $\mathbb{X}^n$ conditional on $X$ have $n$ elements by
%\begin{equation}\label{unlabeled-RFS}
%f(x_1,\cdots,x_n|n)=\frac{\pi(\{x_1,\cdots,x_n\})}{n!p(n)}
%\end{equation}
%Consequently, the multi-object density can be expressed as
%\begin{equation}\label{unlabeled-RFS}
%\begin{split}
%&\pi(\{x_1,\cdots, x_n\})=n!p(n)f(\{(x_1,\cdots,x_n\}|n)
%\end{split}
%\end{equation}
%\end{Def}
Let $\mathcal{L}:\mathbb{X}\mathcal{%
\times }\mathbb{L}\rightarrow \mathbb{L}$ be the projection defined by $%
\mathcal{L}((x,\ell ))=\ell $, then $\mathcal{L}(\mathbf{x})$ is called the
label of the point $\mathbf{x}\in \mathbb{X}\mathcal{\times }\mathbb{L}$. A finite subset $\mathbf{X}$ of $\mathbb{X}\mathcal{\times }\mathbb{L}$
is said to have \emph{distinct labels} if and only if $\mathbf{X}$ and its
labels $\mathcal{L}(\mathbf{X})\triangleq\{\mathcal{L}(\mathbf{x}):\mathbf{x}\in
\mathbf{X}\}$ have the same cardinality. We define the \emph{distinct label
indicator }of $\mathbf{X}$ as $\Delta (\mathbf{X})= \delta _{|%
\mathbf{X}|}(|\mathcal{L(}\mathbf{X})|)$.
\begin{Def}
Given a labeled multi-object density $\bpi$ on $\mathbb{X}\times\mathbb{L}$, and any positive integer $n$, we define the joint existence probability of the label set $\{\ell_1,\cdots,\ell_n\} \subseteq\mathbb{L}$ by
\begin{equation}\notag
w(\{\ell_1,\cdots,\ell_n\})\triangleq\int\bpi(\{(x_1, \ell_1), \cdots, (x_n,\ell_n)\})d(x_1, \cdots, x_n)
\end{equation}
and the  joint probability density on $\mathbb{X}^n$ conditional on their corresponding labels $\ell_1,\cdots,\ell_n$ by
\begin{equation}\notag
p(\{(x_1,\ell_1),\cdots,(x_n,\ell_n)\})\triangleq\frac{ \bpi(\{(x_1, \ell_1), \cdots, (x_n,\ell_n)\})}{w(\{\ell_1,\cdots,\ell_n\})}
\end{equation}
For $n=0$, we define $w(\emptyset)\triangleq\bpi(\emptyset)$ and $p(\emptyset)\triangleq1$. It is implicit that $p(\bX)$ is defined to be zero whenever $w(\mathcal{L}(\bX))$ is zero. 
\end{Def}
Definition 1 is first provided in \cite{refr:label_6}, and using Definition 1, the labeled multi-object density  can be decomposed as 
\begin{align}\label{labeled-RFS}
\bpi(\bX)=w(\mathcal{L}(\bX))p(\bX).
\end{align}

In this paper, we focus on  two most commonly used labeled multi-object distributions, namely the GLMB and LMB distributions. A GLMB labeled RFS is distributed according to: ~\cite{GLMB,refr:tracking-2}
\begin{align}\label{GLMB}
\begin{split}
\bpi(\bX)=\Delta(\bX){\sum}_{c\in\mathbb{C}}w^{(c)}(\mathcal{L}(\bX))[p^{(c)}]^\bX
\end{split}
\end{align}
where $\mathbb{C}$ is a discrete index set, and $w^{(c)}(L)$ and $p^{(c)}$ satisfy
\begin{align}
\begin{split}
{\sum}_{L\subseteq\mathbb{L}}{\sum}_{c\in\mathbb{C}}w^{(c)}(L)&=1,\,\,\mbox{and}
\int p^{(c)}(x,\ell)dx=1.
\end{split}
\end{align}
A GLMB RFS is completely characterized by the set of parameters $\{(\omega^{(c)}(I),p^{(c)}(\cdot)):(I,c)\in\mathcal{F}(\mathbb{L})\times\mathbb{C}\}$.

A labeled multi-Bernoulli (LMB) RFS [24] with state space $\mathbb{X}$, label space $\mathbb{L}$ and (finite) parameter set $\{(r^{(\ell)},p^{(\ell)}(x)):\ell\in\mathbb{L}\}$, is distributed according to
\begin{align}\label{LMB}
\begin{split}
\bpi(\mathbf{X})=\Delta(\bX)w(\mathcal{L}(\bX))p^\bX
\end{split}
\end{align}
where
\begin{align}
w(I)&={\prod}_{i\in\mathbb{L}}(1-r^{(i)}){\prod}_{\ell\in I }1_{\mathbb{L}}(\ell)\frac{r^{(\ell)}}{1-r^{(\ell)}}\\
p(x,\ell)&=p^{(\ell)}(x).
\end{align}

The centerpiece of the RFS based multi-object filtering  is the Bayes multi-object filter~\cite{book_mahler}, which recursively propagates the multi-object posterior density forward in time through a prediction then an update step. When the objects are modelled by labeled RFSs, the Bayesian filter also becomes a multi-object tracker, as object identities (along with other parameters of labeled multi-object distributions) are propagated through the prediction and update steps of the filter.
Of particular interest in this paper is the $\delta$-GLMB filter proposed by Vo~\textit{et~al.}~\cite{GLMB,LMB_Vo2}. This filter is devised based on assuming a special type of GLMB distribution (called $\delta$-GLMB distribution), and is more intuitive on label and data association hypotheses and can be directly implemented. 

Due to presence of explicit data associations in the standard multi-object likelihood, the $\delta$-GLMB suffers from exponential growth of number of components of the posterior with time. To resolve this problem, approximations of $\delta$-GLMB filter that allow tractability have been proposed. An approximation that preserves both the first-order moment and cardinality distribution  of the posterior is the M$\delta$-GLMB filter~\cite{Fantacci-BN}. A faster yet less accurate approximation (that only preserves the first-order moment of the posterior) is LMB filter~\cite{delta_GLMB}. The LMB filter can be implemented substantially faster than the M$\delta$-GLMB filter.

\begin{Rem}
In this paper, an RFS $X$ defined on space $\mathcal{F}(\mathbb{X})$ is referred to as an unlabeled RFS, while an RFS $\bX$ defined on space $\mathcal{F}(\mathbb{X}\times\mathbb{L})$ with each realization having distinct labels is referred to as a labeled RFS. Both unlabeled and labeled RFSs belong to the family of simple finite point processes \cite{Stochastic-Geometry}. 
\end{Rem}
%Denoting the label space at time $k$ by $\mathbb{L}_{[k]}$ (note that due to birth process, the size of label space linearly increases with time), the number of components of the M$\delta$-GLMB posterior is proportional to $|\mathcal{F}(\mathbb{L}_{[k]})|$, whereas the number of components  for an LMB posterior is proportional to $|\mathbb{L}_{[k]}|$.

\subsection{Generalized Covariance Intersection (GCI)}
Using the concept of GCI for distributed multi-sensor fusion was first proposed by Mahler~\cite{Mahler-1} who later developed the GCI fusion  to extend the theory of finite set statistics (FISST) to the distributed environment. Consider a set of sensor nodes $\mathcal{N}=\{1,2,\!\cdots\!,N_s\}$ in a sensor network. Suppose that in  each node $s\in\mathcal{N}$, an RFS-based multi-object filter returns a local multi-object posterior $\pi_s(X)$ defined on a space $\chi$ and computed on the basis of the local information embedded in measurements acquired at node $s$. The GCI fusion rule combines all the local posteriors, returning their geometric mean in the form of an exponential mixture of the local densities,
%$\Pi=\{\pi_s\}_{s\in\mathcal{N}}$ 
\begin{align}\label{G-CI}
\begin{split}
\pi_{\omega}(X)=\frac{\prod_{s\in\mathcal{N}}[\pi_{s}(X)]^{\omega_s} }                                             {\int\prod_{s\in\mathcal{N}}[\pi_{s}(X)]^{\omega_s} \delta X}
\end{split}
\end{align}
where the integral is a set integral as defined in FISST~(see~\cite{book_mahler}) and the weights $\omega_s$ are user-defined parameters with the constraint:
\begin{equation}\label{weight}
\omega_s\geqslant 0, \,\,\,\,{\sum}_{s\in\mathcal{N}}\omega_s=1.
\end{equation}

Note that the GCI rule~(\ref{G-CI}) can also be used to fuse labeled multi-object densities. In that case, the labeled set integral should be used, as defined in \cite{LMB_Vo2}. With labeled multi-object posteriors the space $\chi$ is $\mathbb X \times \mathbb L$, otherwise we have $\chi = \mathbb X$.

The name GCI stems from the fact that (\ref{G-CI}) is the multi-object counterpart of the analogous fusion rule for (single-object) probability densities \cite{EMD-Julier} which, in turn, is a generalization of \textit{Covariance Intersection} originally conceived  for Gaussian probability densities \cite{Uhlmann}. In~\cite{Heskes,Battistelli}, it has been shown that the GCI fusion in (\ref{G-CI}) essentially minimizes the weighted sum of the Kullback-Leibler divergence (KLD) with respect to the local densities, i.e.
\begin{equation}\label{KLA}
  \pi_\omega=\arg \min_\pi {\sum}_{s\in\mathcal{N}}\,\omega_s D_{\text{KL}}(\pi;\pi_s)\end{equation}
where $D_{\text{KL}}$ denotes the KLD, defined as:
\begin{equation}\label{KLD}
\begin{split}
 D_{\text{KL}}(f;g)\triangleq\int f(X)\log{\left({f(X)}\big{/}{g(X)}\right)}\delta X.
  \end{split}
\end{equation}

In view of (\ref{KLA}), the GCI fusion is also called  Kullback-Leibler average (KLA) fusion. In Bayesian statistics, the KLD  can be seen as the information gain achieved when moving from a prior density to a posterior density. Hence, the GCI fusion   essentially provides the density that minimizes the weighted sum of the information gains from the local densities on the basis of (\ref{KLA}).  This choice is coherent with the Principle of Minimum Discrimination Information (PMDI) according to which the probability
density which best represents the current state of knowledge is the one
which produces an information gain as small as possible.
This property is important in order to ensure immunity to double counting, thus avoiding being overconfident on the available information.

\section{Drawbacks of GCI Fusion with Labeled Multi-Object Densities}\label{chp:3}
%Consider a distributed fusion problem in which different sensors have the same field of view and observe the same objects. 
Consider a set of labeled multi-object densities and the corresponding weights $\bPi=\{(\bpi_s(\bX),\omega_s)\}_{s\in\mathcal{N}}$, with each $\bpi_s$ defined on space $\mathbb{X}\times\mathbb{L}$, in the form of (\ref{labeled-RFS}),
\begin{equation}
\begin{split}
&\bpi_s(\{(x_1,\ell_1),\cdots,(x_n,\ell_n)\})=\\
&\,\,\,\,\,\,\,\,\,\,\,\,w_s(\{\ell_1,\cdots,\ell_n\})p_s(\{(x_1,\ell_1),\cdots,(x_n,\ell_n)\}).\\
\end{split}
\end{equation}
Substituting into the GCI fusion rule~\eqref{G-CI}, leads to a fused density in the similar form
\begin{equation}
\begin{split}
&\bpi_\omega(\{(x_1,\ell_1),\cdots,(x_n,\ell_n)\})=\\
&\,\,\,\,\,\,\,\,\,\,\,\,w_\omega(\{\ell_1,\cdots,\ell_n\})p_\omega(\{(x_1,\ell_1),\cdots,(x_n,\ell_n)\}),\\
\end{split}
\end{equation}
where
\begin{align}\label{GCI-GLMB-w}
&\notag w_\omega(\{\ell_1,\!\cdots\!,\ell_n\})=\\
&\,\,\,\,\,\,\,\,\,\,\,\,\,\,\,\,\frac{\prod_{s\in\mathcal{N}}[w_s(\{\ell_1,\cdots,\ell_n\})]^{\omega_s}\eta(\{\ell_1,\cdots,\ell_n\})}{\sum_{I\in\mathcal{F}(\mathbb{L})}\prod_{s\in\mathcal{N}}[w_s(I)]^{\omega_s} \eta(I)}\\
\label{GCI-GLMB-p}
&\notag p_\omega(\{(x_1,\ell_1),\!\cdots\!,(x_n,\ell_n)\})=\\
&\,\,\,\,\,\,\,\,\,\,\,\,\,\,\,\,\frac{\prod_{s\in\mathcal{N}}[p_s(\{(x_1,\ell_1),\!\cdots\!,(x_n,\ell_n)\})]^{\omega_s}}{\eta(\{\ell_1,\cdots,\ell_n\})}
\end{align}
with 
\begin{equation}\label{eta}\notag{\small{\begin{split}
\eta(\{\ell_1,\!\cdots\!,\ell_n\})=\!\!&\int \!\!{\prod}_{s\in\mathcal{N}}[p_s(\{(x_1,\ell_1),\!\cdots\!,(x_n,\ell_n)\})]^{\omega_s}\\
 &\,\,\,\,\,d  \, (x_1\cdots x_n).
\end{split}}}\end{equation}
This indicates that GCI fusion for labeled densities essentially  is performed label-wise, and thus   inherently demands perfect label consistencies between different local sensors, i.e., the same track has the same label in all sensor nodes. When different labels are associated with the same object  in different sensors, GCI fusion does not make sense and indeed performs poorly  because the probabilities $w_s(I)$ or the conditional probability density $p_s(\bX)$  for the same label set hypothesis (except for $\emptyset$) can have large disparity at different sensor nodes, which can largely diminish the fused probability of the label set hypothesis, $w_\omega(I)$.

In the following, we present three common phenomena that lead to the label consistency assumption being violated and GCI fusion failing to produce accurate results.

%\subsection*{Using Adaptive Birth Processes}
\vspace{3mm}
\noindent\emph{Using Adaptive Birth Processes}

The standard formulation of labeled multi-object filters is based on assuming that object birth process is known  as a priori. In some practical situations where the objects can appear anywhere in the surveillance area, the object birth intensity needs to cover the entire the surveillance area and does not add any prior information to the filtering process. An extension which can distinguish between the persistent
and newly-born objects is to formulate an adaptive birth process that is tuned at each scan using the received observations. In presence of such adaptive birth processes at each sensor of a distributed multi-object tracking system, the same newly-born object may be labeled differently in different sensors.
%\subsection*{Random Uncertainties in Observations}

\vspace{3mm}
\noindent \emph{Random Uncertainties in Observations}

Suppose  different sensors share  the same prior information based birth processes \cite{LMB_Vo2}.  In this situation, even if the same label drawn from the label spaces of different sensors has the same implication,  statistical distribution of labels conditioned on the observation set could be substantially different from one sensor node to another, because of randomness of observations. As a result,  the estimated label for the same object may be different in different sensor nodes.

Note that to ensure uniqueness of labels, in labeled multi-object tracking algorithms, a label $\ell$ is comprised of two elements: the time of birth $k$ and the index $i$ that distinguishes different objects born at the same time, i.e.  
$\ell=(k,i)$. It is common that with the labeled multi-object posterior formed in each sensor node, for each object there is a label, say $(\hat k, \hat i)$, with a large weight and hence, representing the estimated label for that object. Due to false alarms (clutters), miss-detections or excessive observation noise, the estimated time of birth $\hat k$ may be different from one sensor node to another and  from the true time of birth $k_0$.  Specifically,
% \begin{enumerate}[$\bullet$] 
% \item 

$\bullet$ $\hat k< k_0$ may occur in a local sensor node due to a false measurement appearing nearby the true track before $k_0$, leading to the deduction that the object is born earlier than $k_0$.

$\bullet$ $\hat k>k_0$ may occur in a local sensor node due to the excessive observation noise or mis-detection of the object  during a few first time steps, leading to the deduction that the object is born after $k_0$.

%\end{enumerate}  
%\subsection*{Local Pruning}
\noindent \emph{Local pruning}

For the sake of numerical tractability of the labeled set filtering algorithms, and reduction of communication costs  in a distributed sensor network, pruning strategies are usually devised to keep the number of hypotheses bounded in each sensor node. This can clearly lead to an object's label to be pruned in one sensor node while remaining in the other, and the label-wise GCI fusion will lead to that the fused label set hypothesis including this label to be given a zero probability because it not supported by all sensor nodes.
\begin{Exa}
Consider a sensor network with two sensors employing an LMB filter in each sensor node.   The  surveillance region is $[-800,800]~\text{m} \times[-600,600]$~m. The standard dynamic and observation models  provided in \cite{LMB_Vo2} are used. The observation model of each single object is linear Gaussian,  the probability of detection $P_{D,k}=0.99$, and the intensity function of clutter $\kappa(\cdot)=5.2\times10^{-6}$. The transition of each single target follows the linear Gaussian model, and the probability of survival $P_{S,k}=0.99$. Both sensors have the prior knowledge that objects are born during times $k = $4\,\text{s}, 5\,\text{s} or 6\,\text{s}, and around the origin $(0,0)$~m. The birth process used at each  sensor is a labeled Bernoulli process $(r^{(k,i)},p^{(k,i)}(x))$,  
\begin{equation}\notag
p^{(k,i)}(x)=\mathcal{N}(x;(0,0,0,0);\mathrm{diag}([300\,\,300\,\,300\,\,300]))
\end{equation}
where  the state $x$   is a vector of planar position and velocity,  the index $i=1$ because only one object is born at the same time, and $k$ denotes the time of birth taking values from $\{4,5, 6\}\,\text{s}$.
The true object is born at $k=5$\,s, and the  true track is shown in Fig. 1 (a).  Two representative types of measurements from sensors 1 and 2 are shown in Fig. 1 (a). No pruning strategies are adopted.

Fig. 1 (a)  also shows object state estimates computed at each sensor node of a single run, which reflect that both local sensors can  estimate the kinematic states accurately in general.  However, the time of birth estimates computed at the two sensor nodes have a  small  but indeed existing difference. Specifically, \begin{enumerate}[Sensor 1:]
\item
due to an excessively noisy measurement at $k\!=\!5$\,s, the object is missed but detected for the first time at $\hat k_1\!=\!6$\,s which is later than the true time of birth.
\item
at time $k\!=\!4$\,s a false measurement is close  to the true track and lead to an estimated birth time of $\hat k_2\!=\!4$\,s which is earlier than the true time of birth.
\end{enumerate}
\begin{figure}[!h]
\begin{minipage}[htbp]{0.49\linewidth}
  \centering
\centerline{\includegraphics[width=4.95cm]{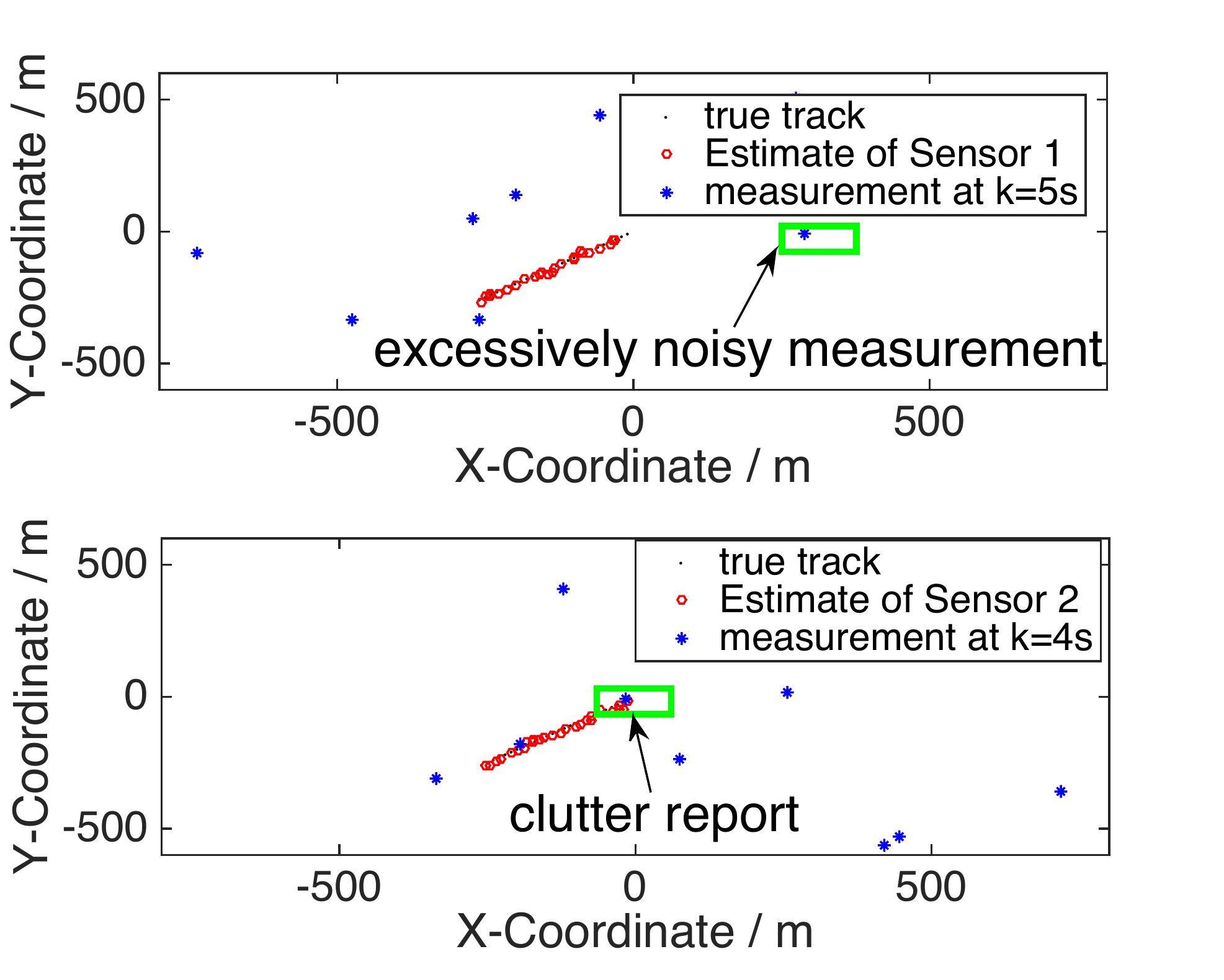}}
%  \vspace{1.5cm}
  \centerline{\small{\small{(a)}}}\medskip
  \end{minipage}
  \hfill
\begin{minipage}[htbp]{0.49\linewidth}
  \centering
  \centerline{\includegraphics[width=4.95cm]{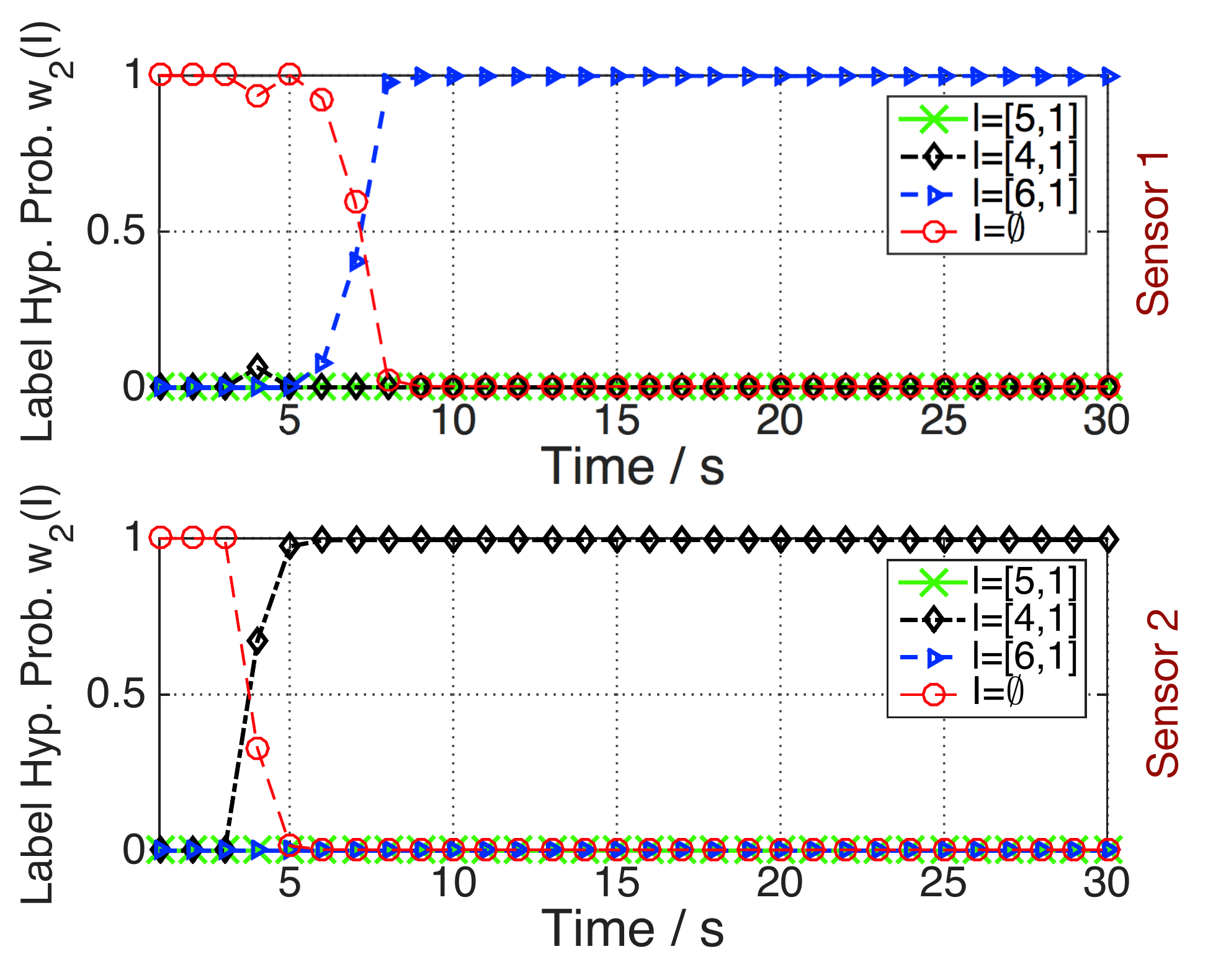}}
%  gci_divergence_and_yes_object_probability.eps
%  \vspace{1.5cm}
  \centerline{\small{\small{(b)}}}\medskip
\end{minipage}
\caption{(a) Tracks and measurements at  two sensors; (b) The respective posterior probabilities $w_1\!(I)$  and $w_2(I)$  for different label set hypotheses at  sensors 1 and 2.}
\label{pbp_mismatching_phenomenon}
\end{figure}
%\footnote{Note that the existence probability $r^{(\ell)}$ at LMB filter plays the same role as label set probability $w(I)$ in the general label multi-object filter.}
Fig. 1 (b) shows the posterior probabilities   $w_1(I)$ and $w_2(I)$ for different label set hypotheses $I=\{(4,1)\}, \{(5,1)\}, \{(6,1)\}, \emptyset$ at sensor nodes 1 and 2, respectively. A large disparity can be observed between the probabilities of the same label set hypothesis at the two sensor nodes, which is the main reason why $\hat k_1$ and $\hat k_2$ are different. Furthermore, if we prune all the label set hypotheses with probabilities less than a threshold $\Gamma=10^{-6}$, at times $k>8\,$s, only one label set hypothesis survives at each sensor node: (6,1) at sensor 1, and (4,1) at sensor 2. GCI fusion  will be completely erroneous in this case.
%, leading to zero probability for all label hypotheses
\end{Exa}
 %Merging
\section{Motivating Analysis}\label{chp:4}
In the previous section, we revealed the drawback of GCI fusion for labeled multi-object densities through an intuitive observation of the fusion formulas,  three common phenomena and a typical numerical experiment. In this section, we  present a principled analysis of how  label inconsistencies affect the fusion performance  from the perspective of PMDI and the declaration  of the object existence.
% {\color{blue}{by exploiting the GCI divergence}}.
 The performance analysis will also motivate the proposed solution that is robust to label mismatches. 

\subsection{Mathematical Tools}
We firstly  introduce the following  notions 
%{\color{blue}{provide mathematical definitions}} 
which are the basis of the subsequent   theoretical analyses.
\subsubsection{Conditional Multi-label Distribution}
We firstly review the  mathematical treatment to transform a labeled multi-object density to its unlabeled  version. The unlabeled version of  a labeled RFS on $\mathbb{X}\times \mathbb{L}$ is given by $\mathcal{K}(\bX)=\{\mathcal{K}(\bx):\bx\in\bX\}$,
where $\mathcal{K}:\mathbb{X}\times\mathbb{L}\rightarrow\mathbb{X}$ is the projection defined by $\mathcal{K}((x,\ell))=x$. Given a labeled RFS $\bX$ distributed according to $\bpi$, $X=\mathcal{K}(\bX)$ is distributed according to the following marginal:~\cite{LMB_Vo2} 
\begin{equation}\label{unlabel-marginal}
{\small{\begin{split}
&\!\!\!\!\!\!\!\!\pi(\{x_1,\!\cdots\!,x_n\})\!=\!\!{\sum}_{\!(\ell_1,\!\cdots\!,\ell_n)\in\mathbb{L}^n}\bpi(\{(x_1,\ell_1),\!\cdots\!,(x_n,\!\ell_n)\}).\!\!\!\!\!\!
\end{split}}}
\end{equation}
\begin{Def}
Given a labeled multi-object density $\bpi$ on $\mathbb{X}\times\mathbb{L}$, and any positive integer $n$, the joint  probability distribution of labels $\ell_1,\cdots,\ell_n$ on $\mathbb{L}^n$  conditional on their corresponding (unlabeled) states $x_1,\cdots,x_n$ is given by
\begin{equation}
\label{varpi}
\begin{split}
\!\!\!\varpi(\{(\ell_1|x_1),\!\cdots\!,(\ell_n|x_n)\})\triangleq\frac{\bpi(\{(x_1,\ell_1),\!\cdots\!,(x_n,\ell_n)\})}{\pi(\{x_1,\!\cdots\!,x_n\})}
\end{split}
\end{equation}
where $\pi(\!\{x_1,\!\cdots\!,x_n\}\!)$ is given in (\ref{unlabel-marginal}).
By convention, $\!\varpi(\emptyset)\!\triangleq\!\!1$.
\end{Def}
%The conditional distribution $\varpi(\{(\ell_1|x_1),\cdots,(\ell_n|x_n)\})$ provides the joint probability 
%that, given a set of unlabeled states $\{x_1,\cdots,x_n\}$, $x_i$ is labeled with $\ell_i$, $i=1,\cdots,n$.  According to Bayes rule, it can also be computed  by
%\begin{equation}\label{label-distribution}\begin{split}
%&\varpi(\{(\ell_1|x_1),\cdots,(\ell_n|x_n)\})=\\
%&\frac{w(\{\ell_1,\!\cdots\!,\ell_n\})p(\{(x_1,\ell_1),\!\cdots\!,(x_1,\ell_n)\})}{{\sum}_{(\ell_1,\!\cdots\!,\ell_n)\in\mathbb{L}^n}w(\{\ell_1,\!\cdots\!,\ell_n\})p(\!\{(x_1,\ell_1),\!\cdots\!,(x_1,\ell_n)\}\!)}
%\end{split}
%\end{equation}
%where  $w(\{\ell_1,\cdots,\ell_n\})$ and $p(\{(x_1,\ell_n),\cdots,(x_n,\ell_n)\})$ are given in Definition 1. 
Note that \eqref{varpi} can be rewritten as:
\begin{equation}\label{labeled-RFS-2}
\begin{split}
&\bpi(\{(x_1,\ell_1),\cdots,(x_n,\ell_n)\})=\\
&\,\,\,\,\,\,\,\,\,\,\,\,\varpi(\{(\ell_1|x_1),\cdots,(\ell_n|x_n)\})\pi(\{x_1,\cdots,x_n\})
\end{split}
\end{equation}
where $n$ takes values from  the field of real number $\mathbb{N}$, which actually provides a new decomposition of labeled multi-object density. Thus, the conditional multi-label distribution $\varpi(\{(\ell_1|x_1),\!\cdots\!,(\ell_n|x_n)\})$ encapsulates all label-related information embedded in the labeled multi-object density $\bpi(\cdot)$. It is this information that makes it possible to estimate the labels of kinematic states and  produce tracks in labeled set filters. 
\subsubsection{GCI Divergence}
In this subsection, \textit{GCI divergence} is introduced as a new measure of discrepancy to quantify the degree of similarity between multiple densities, and evaluate the minimal information gain of GCI fusion with multiple  densities. 
 Consider a set of multi-object densities associated with their corresponding weights,  denoted by  
$$\Pi=\{(\pi_s(X),\omega_s): s\in\mathcal{N}\},$$
where each $\pi_s(X)$ is defined on the same space $\chi$, each weight $\omega_s$ is a given confidence for $\pi_s$, and the weights satisfy (\ref{weight}). For any  multi-object density $\pi(X)$ on $\chi$ (possibly $\mathbb{X}$ or $\mathbb{X}\times\mathbb{L}$),  the weighted \textit{average information gain} (AIG) from $\pi$ to  densities in $\Pi$ is defined as
 \begin{equation}\label{ADI}
{\small{\begin{split}D_{\text{AIG}}(\pi;\Pi) \triangleq {\sum}_{s\in\mathcal{N}}\,\omega_s D_{\text{KL}}(\pi;\pi_s).\end{split}}}
\end{equation}
According to (\ref{KLA}),  the GCI fusion rule \textit{by principle} minimizes the weighted AIG. The resulting minimal weighted AIG  over all densities on $\chi$ is given by:~\cite{Battistelli}
\begin{equation}\label{MADI}
{\small{\min_{\pi} D_{\text{{AIG}}}(\pi;\Pi)=-\log c(\Pi)}}
\end{equation}
where
\begin{equation}\label{c}
{\small{\begin{split}c(\Pi)= \int {\prod}_{s\in\mathcal{N}}[\pi_s(X)]^{\omega_s}\delta X.\end{split}}}
\end{equation}
The quantity $c(\Pi)$ is referred to as \textit{GCI coefficient}. It is always in [0,1] and presents a measure of similarity between densities in $\Pi$.  In this paper, we will call the minimal weighted AIG given in~\eqref{MADI}~and~\eqref{c} as \emph{GCI divergence} denoted by $G(\Pi)$:
\begin{equation}\label{GCI-divergence}
{\small{\begin{split}G(\Pi)=-\log c(\Pi)=-\log \int {\prod}_{s\in\mathcal{N}}[\pi_s(X)]^{\omega_s}\delta X.\end{split}}}
\end{equation} 

\emph{GCI divergence} is  a tool to quantify the degree of similarity between multiple densities. The larger the GCI divergence among densities in $\Pi$ is, the more the corresponding GCI fusion is violating the PMDI. Hence a large GCI divergence can be an indication that the information contained in the densities to be fused are not coherent. The extreme case  $G(\Pi)\rightarrow+\infty\,(c(\Pi)\rightarrow 0^+)$ occurs when the densities are compeletly incompatible and have different supports.
\subsubsection{Yes-Object Probability}
We define \emph{yes-object probability} and \emph{no-object probability} as $P_{\text{y}}(\pi) \triangleq1-\pi(\emptyset)$ and $P_{\text{n}}(\pi)\triangleq\pi(\emptyset)$, respectively, for a given multi-object posterior $\pi$~\cite{book_mahler}. 
Usually, object existence can be declared only if the yes-object probability is greater than  a threshold $\tau$ ( usually $0.5\leqslant\tau<1$).
The yes-object probability is the basis of multi-object state estimation in the sense that only if the existence of objects can be declared, the object states can be extracted; otherwise the best 
estimate of  multi-object state set will be an empty set (no-object inference).
%In case of having two densities $\pi_1$ and $\pi_2$, their GCI divergence is proportional to Renyi divergence~\cite{Uney-2} between the two densities, i.e. 
%\begin{equation}\label{RD-interperation}\notag
%G(\{(\pi_1,\omega_1),(\pi_2,\omega_2)\})=\omega_{2}\mathcal{R}_{\omega_1}(\pi_1||\pi_2)=\omega_1\mathcal{R}_{\omega_2}(\pi_2||\pi_1).
%\end{equation}
%Also, Bhattacharyya distance \cite{Bhattacharyya-distance} is a special case of GCI divergence for two densities with equal weights $\omega_1=\omega_2=\frac{1}{2}$, i.e. 
%\begin{equation}\notag
%D_{\text{B}}(\pi_1,\pi_2)=G(\{(\pi_1,0.5),(\pi_2,0.5)\}).
%\end{equation}
\subsection{Theoretical Analysis}
This section provides a thorough theoretical analysis of how the label inconsistencies between  different labeled multi-object densities affect the performance of GCI  fusion.

Consider a set of labeled multi-object densities $\bPi=\{(\bpi_s,\omega_s)\}_{s\in\mathcal{N}}$ with each $\bpi_s$ defined on space $\mathbb{X}\times\mathbb{L}$, and a set of unlabeled multi-object densities $\Pi=\{(\pi_s,\omega_s)\}_{s\in\mathcal{N}}$ with each $\pi_s$ the marginal of $\bpi_s$ on $\mathbb{X}$. The following proposition states the relationship between GCI divergences of the two sets of multi-object densities.
%having the form of (\ref{labeled-RFS-2})
\begin{Pro}
If each labeled multi-object density $\bpi_s(\cdot)$ is 
\begin{equation}\label{labeled-RFS-3-sensor}
\begin{split}
&\bpi_s(\{(x_1,\ell_1),\cdots,(x_n,\ell_n)\})=\\
&\,\,\,\,\,\,\,\,\,\,\,\,\varpi_s(\{(\ell_1|x_1),\cdots,(\ell_n|x_n)\})\pi_s(\{x_1,\cdots,x_n\})\end{split}
\end{equation}
of form (\ref{labeled-RFS-2}),  then the GCI divergence for densities in $\bPi$ is given by  
\begin{equation}\label{GCI-label-decomposion}
G(\bPi)=G(\Pi)-\log \emph{E}_{\pi_\omega}[\mu(X)] 
\end{equation}
where $G(\Pi)$ is the GCI divergence  for densities in $\Pi$, and
\begin{equation}\label{GCI-coefficient-label}
{\small{\begin{split}
&\mu(\{x_1,\cdots,x_n\})=\\
&{\sum}_{(\ell_1,\cdots,\ell_n)\in\mathbb{L}^n} {\prod}_{s\in\mathcal{N}}[\varpi_s(\{(\ell_1|x_1),\cdots,(\ell_n|x_n)\})]^{\omega_s}
\end{split}}}
\end{equation}
is GCI coefficient for the set of conditional multi-label distributions $\{(\varpi_s(\{(\ell_1|x_1),\!\cdots\!,(\ell_n|x_n)\}),\omega_s)\}_{s\in\mathcal{N}}$, and $\emph{E}_{\pi_\omega}(\cdot)$ denotes expectation
with respect to  $\pi_\omega$,  with  $\pi_\omega$  being the fused density returned by GCI fusion of all densities in $\Pi$.
\end{Pro}
\begin{proof}
%\noindent\emph{Proof. }$\,$
Substituting the densities in~\eqref{labeled-RFS-3-sensor} in the definition of GCI divergence~(\ref{GCI-divergence}) leads to
\begin{small}
\begin{equation}\label{GCI-divergence-proof1}
{\small{\begin{split}
&G(\bPi) = -\log {\sum}_{n=0}^{\infty}\frac{1}{n!}\int \!\!{\prod}_{s\in\mathcal{N}}[\pi_s(\{x_1,\cdots,x_n\})]^{\omega_s}\\
&\!\!\!\!\!\!\sum_{(\ell_1,\!\cdots\!,\ell_n)\in\mathbb{L}^n}\prod_{s\in\mathcal{N}}[\varpi_s(\{(\ell_1|x_1),\!\cdots\!,(\ell_n|x_n)\})]^{\omega_s} d(x_1,\!\cdots\!,x_n).\\
\end{split}}}
\end{equation}
\end{small}
Substituting
\begin{equation}
{\small{\begin{split}
&c(\Pi)\!=\!\sum_{n=0}^{\infty}\frac{1}{n!}\!\int\! \prod_{s\in\mathcal{N}}[\pi_s(\{x_1,\!\cdots\!,x_n\})]^{\omega_s}d(x_1,\!\cdots\!,x_n)
\end{split}}}
\end{equation}
and  (\ref{GCI-coefficient-label}) into (\ref{GCI-divergence-proof1}), $G(\bPi)$ can be rewritten as
\begin{equation}\label{GCI-divergence-proof2}
{\small{\begin{split}
G(\bPi)\!=\!&-\log c(\Pi){\sum}_{n=0}^{\infty}\frac{1}{n!}\!\int\!\frac{{\prod}_{s\in\mathcal{N}}[\pi_s(\{x_1,\!\cdots\!,x_n\})]^{\omega_s}}{c(\Pi)}\\
&\times\mu(\{x_1,\cdots,x_n\}) d(x_1,\cdots,x_n).\\
%=&G(\pi)-\log\sum_{n=0}^{\infty}\frac{1}{n!}
%\int \pi_\omega(\{x_1,\cdots,x_n\})c(\{x_1,\cdots,x_n\})d(x_1,\cdots,dx_n)
\end{split}}}
\end{equation}
Based on the GCI fusion rule, the GCI fusion for $\Pi$ is
\begin{equation}\notag
{\small{ \begin{split}
\pi_\omega(X)={\prod}_{s\in\mathcal{N}}[\pi_s(X)]^{\omega_s}\big/c(\Pi).
\end{split}}}
\end{equation}
Hence, (\ref{GCI-divergence-proof2}) can be further represented as
\begin{equation}\notag
{\small{\begin{split}
G(\bPi)=&-\log c(\Pi)-\log{\sum}_{n=0}^{\infty}\frac{1}{n!} \int \pi_\omega(\{x_1,\cdots,x_n\})\\
&{} \times \mu(\{x_1,\cdots,x_n\})d(x_1,\cdots,x_n)\\
=&G(\Pi)-\log \text{E}_{\pi_\omega}[\mu(X)].
\end{split}}}
\end{equation}
%\qed
\end{proof}
%\vspace{-3mm}
The above result reveals that  $G(\bPi)$ can be decomposed into two parts: one part is  $G(\Pi)$ which reflects the discrimination information between kinematic states of different sensors;   the other part is  $-\log E_{\pi_\omega} [\mu(X)]$ with $E_{\pi_\omega} [\mu(X)]$  being the statistical average of the GCI coefficient of conditional multi-label distribution which reflects the discrimination  information between label distributions of different sensors.

\begin{Def}
We define the ``label inconsistency indicator'' with respect to a set of labeled multi-object densities $\bPi$ as 
\begin{equation}
d_G(\bPi)\triangleq G(\bPi)-G(\Pi)=-\log E_{\pi_\omega} [\mu(X)].
\label{dG_def}
\end{equation}
\end{Def}
The label inconsistency indicator $d_{G}(\bPi)$ is a  measure to quantify the inconsistencies of label information embedded into  multiple labeled densities. A larger value of $d_{G}(\bPi)$ indicates a higher level of  label inconsistencies between densities in $\bPi$.  Moreover, the quantity $d_{G}(\bPi)$  can reflect the difference between GCI divergences of $\bPi$ and $\Pi$.  The following corollary establishes upper and lower bounds on $d_G(\bPi)$. 
\begin{Cor}
The following inequalities hold,
\begin{align}
0 \leqslant d_G \leqslant -\log \pi_\omega(\emptyset)
%-\log\emph{E}_{\pi_\omega}[\mu(X)]&\geqslant0.
\end{align}
with $\pi_\omega$  the fused density returned by GCI fusion of $\Pi$.
\end{Cor}
\noindent\emph{Proof. }$\,$
%\begin{proof}
By definition, $\varpi_s(\emptyset)=1$ is always true. Therefore, 
\begin{equation}\label{mu-emptyset}
\mu(\emptyset)={\prod}_{s\in\mathcal{N}}[\varpi_s(\emptyset)]^{\omega_s}=1.
\end{equation}
For $X\neq \emptyset$, the term $\mu(X)$ denotes a GCI coefficient and is therefore, within [0,1]. The quantity $\text{E}_{\pi}[\mu(X)]$ is given by
\begin{equation}\label{E-mu-1}\begin{split}
\!\!\text{E}_{\pi_\omega}[\mu(X)]\!=&\pi_\omega(\emptyset)\mu(\emptyset)\!+\!{\sum}_{n=1}^{\infty}\frac{1}{n!}\int\!\pi_\omega(\{x_1,\!\cdots\!,x_n\}) \\
&\times\mu(\{x_1,\!\cdots\!,x_n\})d(x_1,\!\cdots\!,x_n).
\end{split}\end{equation}
%Substitution of (\ref{mu-emptyset}) into (\ref{E-mu-1}), yields
%\begin{equation}
%\begin{split}
%\text{E}_{\pi_\omega}[\mu(X)]=&\pi_\omega(\emptyset)+{\sum}_{n=1}^{\infty}  \frac{1}{n!}\int \pi_\omega(\{x_1,\cdots,x_n\}) \\
%&\times\mu(\{x_1,\cdots,x_n\})d(x_1,\cdots,x_n).
%\end{split}
%\end{equation}
Since each $\mu(\cdot)$ term within summing integrations is less than or equal to 1, the upper bound of $\text{E}_{\pi_\omega}[\mu(X)]$ is given by
%\begin{array}{rcl}
\begin{small}
\begin{equation}\notag
\begin{split}
\text{E}_{\pi_\omega}[\mu(X)]&\leqslant \pi_\omega(\emptyset)\!+\!{\sum}_{n=1}^{\infty} \frac{1}{n!}\!\int\! \pi_\omega(\{x_1,\cdots,x_n\}) d(x_1,\cdots,x_n)\\  
&=\int \pi_\omega(X)\delta X = 1.
\end{split}
\end{equation}
\end{small}
This also establishes the lower bound of $-\log \text{E}_{\pi_\omega}[\mu(X)]$, 
$$d_G = -\log \text{E}_{\pi_\omega}[\mu(X)] \geqslant 0.$$
On the other hand, since $\mu(\emptyset)=1$ and each $\mu(\{x_1,\cdots,x_n\}) (n\geqslant1)$ term is non-negative, we have: 
$$\text{E}_{\pi_\omega}[\mu(X)]\geqslant \pi_\omega(\emptyset)$$
which also establishes an upper bound on $-\log \text{E}_{\pi_\omega}[\mu(X)]$, 
$$-\log \text{E}_{\pi_\omega}[\mu(X)] \leqslant -\log \pi_\omega(\emptyset).$$
Together, we have: 
\begin{equation}
0\leqslant d_G=-\log \text{E}_{\pi_\omega}[\mu(X)]\leqslant -\log\pi_\omega(\emptyset)
\end{equation}
\qed

%\end{proof}
From { Corollary} 1, we conclude that the upper bound on the quantity $d_G$ depends on $\pi_\omega(\emptyset)$. 
In the limit case $\pi_\omega(\emptyset)=0$, upper bound can reach $+\infty$.  

%\textcolor{red}{Furthermore,  $d_G$ reach its upper bound only if the GCI coefficient $\mu(X)=1$, in other words, the GCI divergence between $\varpi_s(\cdot), s\in\mathcal{N}$ is $+\infty$ holds for each $X\in\mathcal{F}(\bX)$.}

To further investigate how the label inconsistencies affect the fusion performance,  we build up the functional relationship between $d_{G}(\bPi)$ and   the so-called \textit{yes-object probability} and its opposite counterpart, \textit{no-object probability}  in the following corollary.
%\begin{Lam}
%For a labeled multi-object posterior and its unlabeled version, their yes/no-object probabilities are the same, i.e. 
%$$P_{\text{y}}(\bpi)=1-\bpi(\emptyset)=1-\pi(\emptyset)=P_{\text{y}}(\pi).$$
%$$P_{\text{n}}(\bpi)=\bpi(\emptyset)=\pi(\emptyset)=P_{\text{n}}(\pi).$$
%\end{Lam}
% 
 
\begin{Cor}
The yes-object probability of the labeled GCI fusion $\bpi_\omega$ can be written in terms of the label inconsistency indicator $d_{G}(\bPi)$ and the yes-object probability of the corresponding unlabeled 
GCI fusion $\pi_\omega$ as follows
\begin{equation}\label{Re-P-y}
P_{\text{y}}(\bpi_\omega)=1-e^{d_G(\bPi)} [1-P_{\text{y}}(\pi_\omega)].
\end{equation}
\end{Cor}%\begin{proof}
\noindent\emph{Proof. }$\,$According to (\ref{G-CI}), we can get the no-object probabilities after GCI fusion with $\bPi$ and  GCI fusion with $\Pi$ as, respectively, 
\begin{equation}
\begin{split}
P_{\text{n}}(\bpi_\omega)=&{\prod}_{s\in\mathcal{N}}[\bpi_s(\emptyset)]^{\omega_s}\big/ c(\bPi)\\
P_{\text{n}}(\pi_\omega)=&{\prod}_{s\in\mathcal{N}}[\pi_s(\emptyset)]^{\omega_s}\big/c(\Pi).
\end{split}
\end{equation}
As for each  labeled density $\bpi_s$  and its corresponding unlabeled version $\pi_s$, $\bpi_s(\emptyset)=\pi_s(\emptyset)$, we have
\begin{equation}
P_{\text{n}}(\bpi_\omega)=\frac{c(\Pi)}{c(\bPi)}P_{\text{n}}(\pi_\omega)=e^{G(\bPi)-G(\Pi)}P_{\text{n}}(\pi_\omega).
\end{equation}
Hence, we can get the following yes-object probability
\begin{equation}
P_\text{y}(\bpi_\omega)=1-P_{\text{n}}(\bpi_\omega)=1-e^{d_G}[1-P_{\text{y}}(\pi_\omega)].
\end{equation}
\qed
%\end{proof}
\begin{figure}[t]
\centering
\includegraphics[width=6cm]{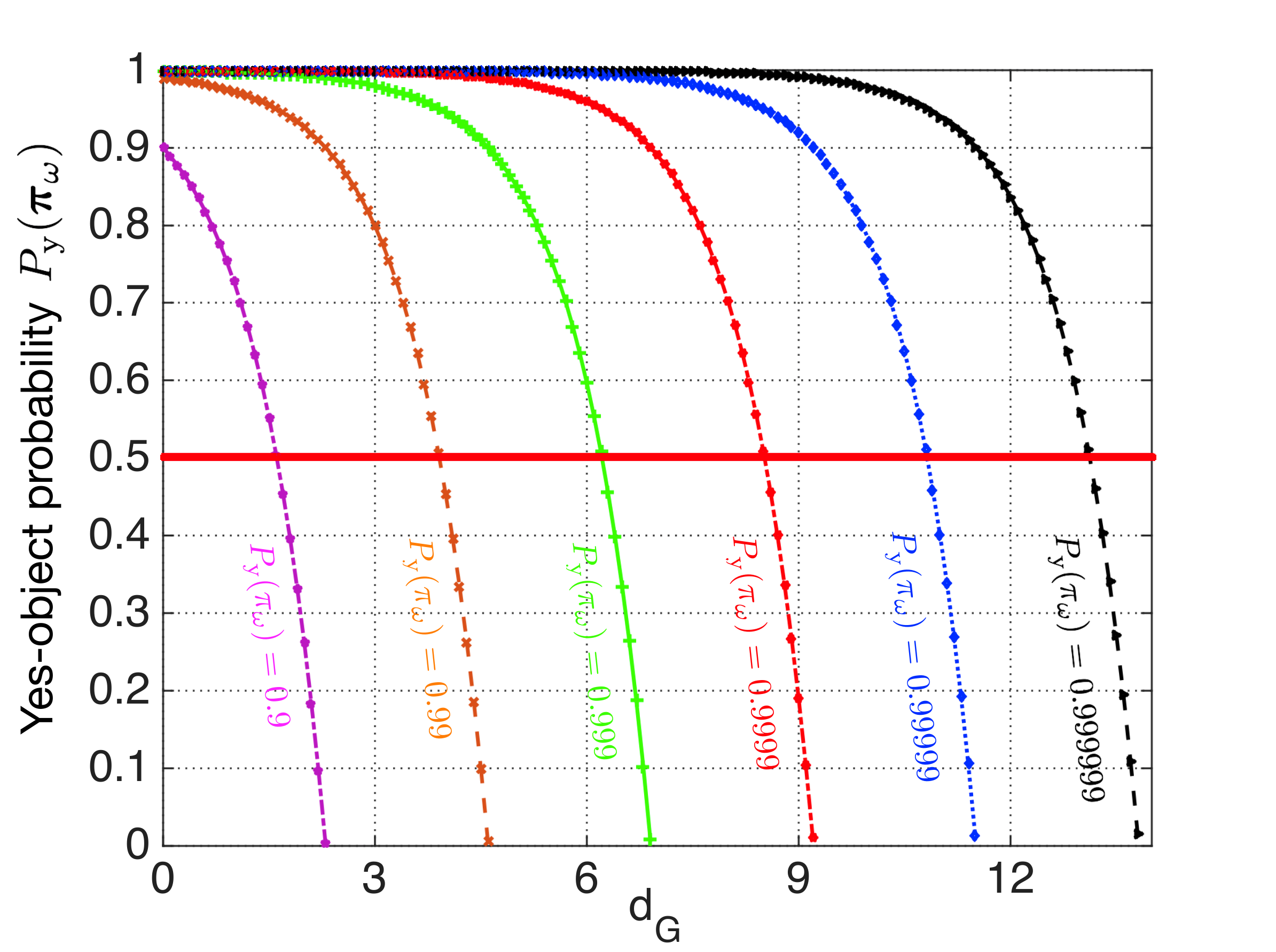}
\caption{Sharp decline of $P_\text{y}(\bpi_\omega)$ with increasing $d_G(\bPi)$ for various given values of $P_\text{y}(\pi_\omega)$. \label{fig2}}
\end{figure}

The above result shows that  given  $P_\text{y}(\pi_\omega)$,  $P_\text{y}(\bpi_\omega)$ monotonically decreases with $d_G$. As shown in Fig.~\ref{fig2}, there can be a sharp decline for $P_\text{y}(\bpi_\omega)$  if $P_\text{y}(\pi_\omega)$ is close to one. Particularly, when
 $$d_G>-\log \pi_\omega(\emptyset)+\log(1-\tau),$$ 
 we have
 $$P_\text{y}(\bpi_\omega)<\tau,$$ and as a result,  existence of no object can be inferred from GCI fusion with $\bPi$. Hence, when the label inconsistency indicator $d_G(\bPi)$ is too large, the GCI fusion with labeled densities can  lose all tracks.   Indeed, from Fig.~\ref{fig2}, we observe that with the unlabeled yes-object probability is $P_\text{y}(\pi_\omega)=0.999$,  if $d_G$ is larger than 6.2, GCI fusion of labeled multi-object posteriors in $\bPi$ will (mis)lead us to the no-object inference since the yes-object probability $P_\text{y}(\bpi_\omega)$ will be less than 0.5.
%  A large gap $d_G$ means a significant disparity among the label information embedded in various labeled densities in $\bPi$. This happens in Example 1 given in the previous section.

\begin{Rem}
Having $P_{\text{y}}(\bpi_\omega)<\tau$ is a sufficient condition for no-object inference, but a necessary one. A filter can also lead to an empty set estimate (no-object inference) if $P_{\text{n}}(\bpi_\omega) = \bpi_\omega(\emptyset)$ becomes larger than the probabilities of any other cardinalities ($n\neq 0$).
\end{Rem}

\subsection{Summary and Motivation}
The analysis presented in this section shows that the label inconsistency indicator $d_G(\bPi)$ can reflect the impact of disparities between label information in different sensor nodes on the performance of GCI fusion with labeled densities. Given the labeled posteriors $\bPi$ and their unlabeled versions $\Pi$, a larger $d_G(\bPi)$ means
\begin{itemize}
\item a larger GCI divergence for $\bPi$, and thus, less optimality of GCI fusion (as it will violate PMDI more); and
\item a smaller yes-object probability for the fused labeled density.
\end{itemize}
This shows that GCI fusion with labeled densities is highly sensitive to the label inconsistency indicator $d_G(\bPi)$ and has little tolerance (if not none) to inconsistencies between label information embedded in different local sensor nodes. 

In practice, even if all the local sensors work well, the label inconsistency indicator $d_G(\bPi)$ can be still large. This can happen with any of the common phenomena listed in previous section, i.e., when we have an adaptive birth process in place within each local filter, or due to uncertainties in observations (excessive noise, false alarm rate, low probability of detection), or due to local pruning operations. 

Revisiting example 1, we calculate the divergences $G(\bPi)$, $G(\Pi)$, the label inconsistency indicator $d_G(\bPi)$ and its the upper bound $-\log\pi_\omega(\emptyset)$. Fig.~\ref{fig3}(a) shows those values plotted versus time. It can be seen that while GCI fusion of the unlabeled densities consistently returns small GCI divergence values, those values returned by the GCI fusion of labeled densities are large to the extent that the label inconsistency indicator $d_G(\bPi)$ is very close to its upper bound. This provides an in-depth explanation for the reason why GCI fusion of the LMB densities fails in this case. We also compute yes-object probabilities returned by each local filter, by GCI fusion of labeled posteriors, and by GCI fusion of unlabeled posteriors. The results are plotted in Fig.~\ref{fig3}(b) showing that while the two local filters and GCI fusion of unlabeled posteriors work well (in terms of successfully returning one object state estimate), GCI fusion of labeled posteriors fails.
\begin{figure}[htb]
\begin{minipage}[htbp]{0.49\linewidth}
  \centering
  \centerline{\includegraphics[width=4.95cm]{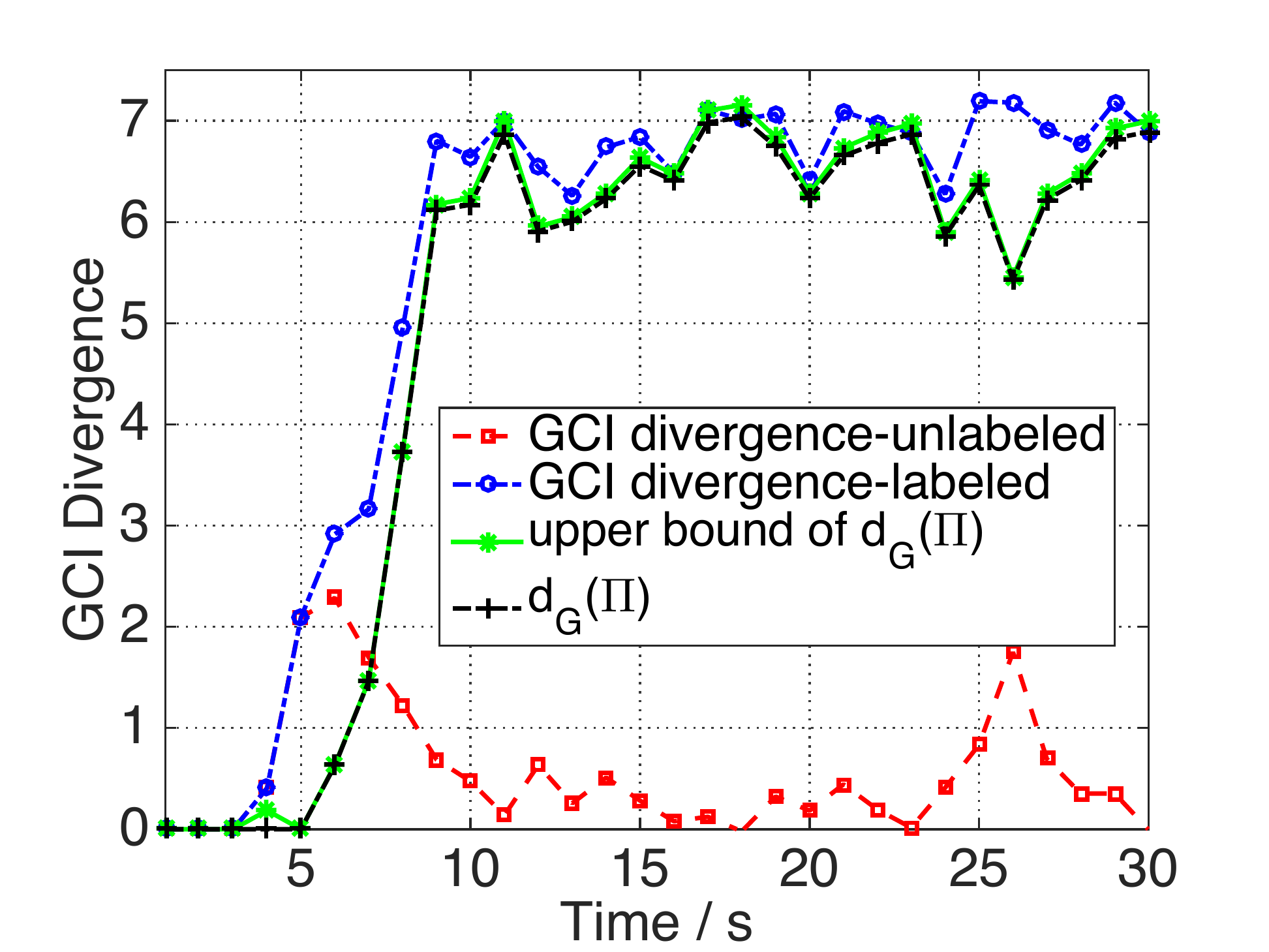}}
%  \vspace{1.5cm}
  \centerline{\small{\small{(a)}}}\medskip
  \end{minipage}
  \hfill
\begin{minipage}[htbp]{0.49\linewidth}
  \centering
  \centerline{\includegraphics[width=4.95cm]{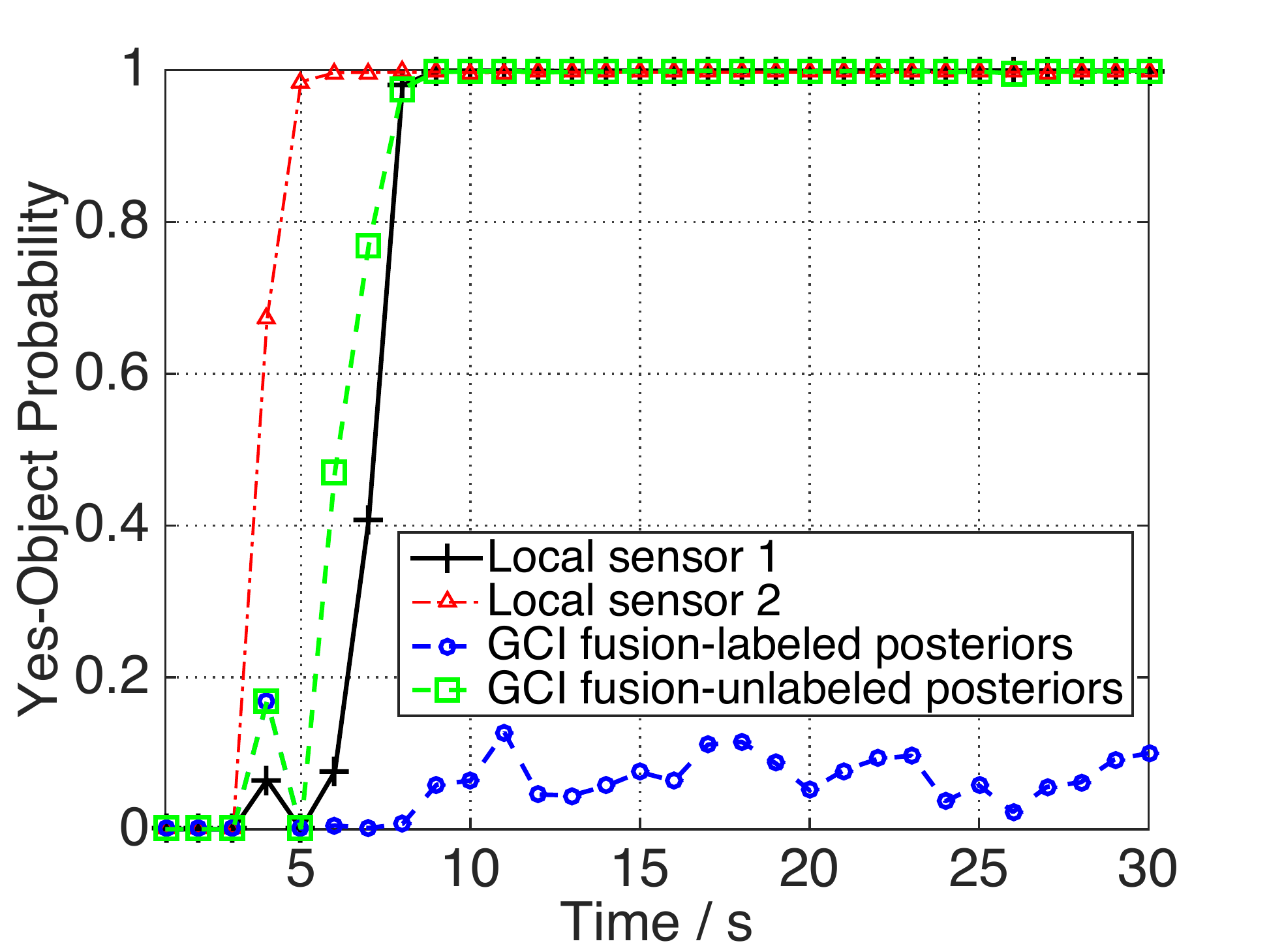}}
%  gci_divergence_and_yes_object_probability.eps
%  \vspace{1.5cm}
  \centerline{\small{\small{(b)}}}\medskip
\end{minipage}
\caption{(a) GCI divergences vs time; (b) Yes-object probabilities vs time.}
\label{fig3}
\end{figure}
\section{Robust Distributed Fusion: Proposed Method}\label{chp:5}
The discussions and analyses presented in previous sections show that the performance of GCI fusion with labeled densities is highly sensitive to disparities in local label information embedded in labeled posteriors at each sensor node, which is quantified by the label inconsistency indicator $d_G(\bPi)$. Moreover, in Section IV, the functional relationships between the GCI fusion with $\bPi$ and $\Pi$ are also  founded in terms of the GCI divergence and the yes-object probability. Revisiting Proposition 1 and Corollary 2, we can easily conclude that
\begin{align}
G(\Pi)\leqslant &G(\bPi)\\P_\text{y}(\pi_\omega)\geqslant &P_\text{y}(\bpi_\omega)
\end{align} 
which mean that GCI fusion with unlabeled posteriors $\Pi$ is expected to achieve a smaller minimal AIG value and a larger yes-object probability than GCI fusion with labeled posteriors $\bPi$. Indeed, $G(\Pi)\!=\!G(\bPi)$ and $P_\text{y}(\pi_\omega)\!=\!P_\text{y}(\bpi_\omega)$ hold only if 
$$
\forall X\in\mathcal{F}(\mathbb{X}),\ \ \ \mu(X) = 1
$$
and this occurs only if there is 100\% consistency between label-object relationships inferred from local sensor node posteriors; a condition that is hardly satisfied in practice. Hence, GCI fusion is generally expected to perform better on unlabeled posteriors assembled in $\Pi$ rather than the labeled posteriors in $\bPi$. Quality of performance of GCI fusion algorithms with unlabeled random set densities, has been already verified in several works involving particular classes of RFSs such as Poisson, i.i.d clusters and MB RFSs~\cite{Uney-2,Battistelli,GCI-MB}.
 
As for the task of producing tracks (equivalent to get $
 \varpi_\omega(\cdot)$), the  fusion with the conditional multi-label distribution $\varpi_s(\cdot), s\in\mathcal{N}$ does not make sense when the  discrimination between the statistics of labels at different sensors is large. Hence, our strategy is that if the fusion is performed at sensor $s_0$,  the conditional multi-label distribution for the fused density $\varpi_\omega(\cdot)$ is produced by performing the GCI fusion with $\varpi_{s_0}(\cdot)$ with uniform distributions based on (\ref{G-CI}), i.e. 
 
\begin{equation}\label{fused-label-distribution}\begin{split}&\varpi_\omega(\{(\ell_1|x_1),\cdots,(\ell_n|x_n)\})=\\
&\frac{[\varpi_{s_0}(\{(\ell_1|x_1),\cdots,(\ell_n|x_n)\})]^{\omega_{s_0}}}{\sum_{(\ell_1,\cdots,\ell_n)\in\mathbb{L}^n_{s_0}} [\varpi_{s_0}(\{(\ell_1|x_1),\cdots,(\ell_n|x_n)\})]^{\omega_{s_0}}}\end{split}\end{equation}
where $\mathbb{L}_{s_0}$ is the label space at sensor $s_0$.

The proposed  robust solution to the distributed fusion of labeled posteriors is summerized  as follows. Suppose that  at the current time  step, sensor $s_0$ has received    agent densities from its neighbours, 

\begin{enumerate}[\textbf{Step} 1:]
\item   All agent  labeled posteriors $\bpi_s(\bX)$ are marginalized to their unlabeled versions $\pi_s(X)$;  

\item    Get the GCI fusion $\pi_\omega(X)$   with  the unlabeled posteriors $\pi_s(X)$   based on (\ref{G-CI});

\item   The labeled fused density is produced by augmenting $\pi_\omega(X)$ with the labels in $\mathbb{L}_{s_0}$, i.e. 
%\begin{equation}\label{construction-fused-label}
% \bpi_\omega(\bX)=\varpi_\omega(\bX)\pi_\omega(\mathcal{K}(\bX))
%\end{equation}
\begin{equation}\label{construction-fused-label}
\begin{split}
&\!\!\!\bpi_\omega(\{(x_1,\ell_1),\cdots,(x_n,\ell_n)\})=\\
&\!\!\!\varpi_\omega(\{(\ell_1|x_1),\cdots,(\ell_n|x_n)\})\pi_\omega(\{x_1,\cdots,x_n\})
\end{split}
\end{equation}
where $\varpi_\omega(\cdot)$ is given in (\ref{fused-label-distribution}).
\end{enumerate} 
\begin{Rem} That the label spaces of the algorithms at different sensors are the same, i.e.  $\mathbb{L}_1=\cdots=\mathbb{L}_{N_s}=\mathbb{L}$, is an inherent assumption for GCI fusion with labeled densities.   However, in the practice, the real used  label spaces for different sensors  may be different, for example when the filters use pruning and merging strategies, or if the sensors use the adaptive birth process. Note that the proposed solution  can work even if  the real used  label spaces  are not the same, because the fusion is performed on the state space $\mathbb{X}$. 
\end{Rem}

% In the following, to distinguish the label spaces of different sensors, we use $\mathbb{L}_s, s\in\mathcal{N}$ for different sensors.

\subsection{The Unlabeled Version of Labeled RFS}
According to the proposed robust distributed fusion solution, the first step is to compute  the marginals of local labeled densities on kinematic space $\mathbb{X}$. In this subsection, we    investigate the unlabeled versions of common  labeled densities, i.e.  GLMB densities and its subclass. Firstly, we present a class of RFS defined on unlabeled state space, named as generalized multi-Bernoulli (GMB) RFS. Then we provide the mathematic representations of the unlabeled version of GLMB densities and its subclass, which turned out to be the same unlabeled RFS family, namely, GMB RFS family.
\subsubsection{GMB RFS}
The GMB RFS was first proposed in \cite{GCI-MB}, and this paper provides a formal definition of GMB RFS. 
\begin{Def}\label{defination_gmb}
A GMB RFS is an RFS on state space $\mathbb{X}$ distributed according to
\begin{equation}\label{GMB_ori}
  \pi(\!\{x_1,\!\cdots\!,x_n\}\!)\!\!=\!\!\!\sum_{\sigma}\!\!\!\!\sum_{(\mathcal{I},\phi)\in \mathcal{F}_n(\mathbb{I})\!\times\!\Phi}\!\!\!\!\!\!w^{(\mathcal{I},\phi)}\!\!\prod_{i=1}^{n}p^{(\phi),\mathcal{I}^v\!(i)}(x_{\sigma\!(i)})
\end{equation}
where  $\sigma$ denotes one permutation of $\mathcal{I}$, $\sigma(i)$ denotes the $i$th element of the permutation,  the summation $\sum_\sigma$ is taken over all permutations on the numbers ${1,\cdots,n}$, $\Phi$ is a discrete space, $\mathbb{I}$ is the  set of indexes of densities, $\mathcal{I}^v\in \mathbb{I}^{|\mathcal{I}|}$ is a vector constructed by sorting the elements of the set $\mathcal{I}$, $w^{(\mathcal{I},\phi)}$ and  $p^{(\phi),\imath}(x)$ satisfy

\begin{equation}\label{paramters}
{\small{\begin{split}
{\sum}_{\mathcal{I}\in\mathcal{F}(\mathbb{I})}{\sum}_{\phi\in\Phi}w^{(\mathcal{I},\phi)}&=1,\\
\int p^{(\phi),\imath}(x)dx&=1, \imath \in \mathbb{I}.
\end{split}}}
\end{equation}
\end{Def}
A  GMB distribution is constructed by  a set of hypotheses, $\{(\mathcal{I},\phi):(\mathcal{I},\phi)\in \mathcal{F}(\mathbb{I})\times\Phi\}$. We define a set of densities for each $\phi\in\Phi$ as:
\begin{equation}\label{set of function}
\mathfrak{P}^{(\phi)}\triangleq\{p^{(\phi),\imath}(x)\}_{\imath\in\mathbb{I}}.
\end{equation}
Under each hypothesis, the corresponding weight is $w^{(\mathcal{I},\phi)}$ and the corresponding density set is  $\mathfrak{P}^{(\phi)}$. Thus a GMB density is completely characterized by the set of parameters $\{(w^{(\mathcal{I},\phi)},\mathfrak{P}^{(\phi)}):(\mathcal{I},\phi)\in \mathcal{F}(\mathbb{I})\times\Phi\}$. 
%Note that the number of $\omega^{(\mathcal{I},\phi)}$ and $\mathfrak{P}^{(\phi)}$ which need to store and compute is $|\mathcal{F}(\mathbb{\mathcal{I}})\times\Phi|$ and $|\Phi|$, respectively.

%\begin{Pro}
%The cardinality distribution of GMB RFS in the form of (\ref{GMB_ori}) is given by
%\begin{align}\label{GMB-car}
%\begin{split}
%p(n)=\sum_{(\mathcal{I},\phi)\in \mathcal{F}_n(\mathbb{I})\!\times\!\Phi}\!\!\!\!\!w^{(\mathcal{I},\phi)}
%\end{split}
%\end{align}
%\end{Pro}
The first-order moment of a GMB RFS distributed according to (\ref{GMB_ori}) is given by: \cite{GCI-MB}
 \begin{equation}\label{GMB-PHD}
\begin{split}
  v(x)={\sum}_{\imath\in\mathbb{I}} \,r^{(\imath)} p^{(\imath)}({x}).
 \end{split}
\end{equation}
where
\begin{align}
\label{r_l}r^{(\imath)}&={\sum}_{\mathcal{I}\in \mathcal{F}(\mathbb{I})}{\sum}_{\phi\in\Phi}1_{\mathcal{I}}(\imath)w^{(\mathcal{I},\phi)} ,\\
\label{p_l}p^{(\imath)}({x})&={\sum}_{\mathcal{I}\in \mathcal{F}(\mathbb{I})}{\sum}_{\phi\in\Phi}1_{\mathcal{I}}(\imath)w^{(\mathcal{I},\phi)} p^{(\phi),\imath}({x})\bigg/r^{(\imath)}.
\end{align}
\subsubsection{Unlabeled Versions of Common Labeled Densities}
We present the mathematical representations of the unlabeled versions for  GLMB density and its subclass, LMB density. 
%Note that Proofs of Proposition 2 and Corollary 3 are given in  ..................We also provide the unlabeled version of M$\delta$-GLMB density which is another subclass of GLMB density in ..............................
\begin{Pro}\label{P_GLMB}
If a labeled RFS $\bX$ on $\mathbb{X\times}\mathbb{L}$ is a GLMB RFS distributed according to (\ref{GLMB}), then  $X=\mathcal{K}(\bX)$ is distributed as
\begin{align}\label{GMB-GLMB}
\begin{split}
  \pi(\{&x_1,\cdots,x_n\})=\\&{\sum}_{\sigma}{\sum}_{(I,c)\in\mathcal{F}_n(\mathbb{L})\times\mathbb{C}}w^{(I,c)}{\prod}_{i=1}^{n}p^{(c),I^v(i)}(x_{\sigma(i)})
\end{split}
\end{align}
where
\begin{equation}\label{para-GMB-GLMB}
\begin{split}
w^{(I,c)}&\triangleq w^{(c)}(I), I\in\mathcal{F}(\mathbb{L})\\
p^{(c),\ell}(x)&\triangleq p^{(c)}(x,\ell), \ell\in\mathbb{L}.
\end{split}
\end{equation}
\end{Pro}
\proproof{A}
%\noindent \emph{Proof. }
%\begin{Pro}\label{P_delta_GLMB}
%If a labeled RFS $\bX$ on $\mathbb{X}\times\mathbb{L}$ is a $\delta$-GLMB RFS distributed according to (\ref{delta-GLMB}), then $X=\mathcal{K}(\bX)$ is distributed as:
%\begin{equation}\label{delta-GMB}
%\pi(\{x_1,\ldots,x_n\})=\sum_{\sigma}\sum_{(I,\xi)\in\mathcal{F}_n(\mathbb{L})\times\Xi}\omega^{(I,\xi)}\prod_{i=1}^{n}p^{(\xi),I^v(i)}(x_{\sigma(i)})
%\end{equation}
%where
%\begin{equation}\label{para-delta-GMB}
%\begin{split}
%  &p^{(\xi),\ell}(x)\triangleq p^{(\xi)}(x,\ell), \ell\in\mathbb{L}.
%\end{split}
%\end{equation}
%\end{Pro}

%
%\begin{Pro}\label{P_Mdelta_GLMB}
%If a labeled RFS $\bX$ on $\mathbb{X}\times\mathbb{L}$ is an M$\delta$-GLMB RFS distributed according to (\ref{Mdelta-GLMB}), then $X=\mathcal{K}(\bX)$ is distributed as:
%\begin{equation}\label{Mdelta-GMB}
%\begin{split}
%&\pi(\{x_1,\ldots,x_n\})\\
%&=\sum_\sigma \sum_{(I,I')\in\mathcal{F}_n(\mathbb{L})\times\mathcal{F}(\mathbb{L})}\delta_{I'}(I)w^{(I)}\prod_{i=1}^{n}p^{(I),I^v(i)}(x_{\sigma(i)})\\
%&=\sum_{\sigma}\sum_{I\in\mathcal{F}_n(\mathbb{L})}w^{(I)}\prod_{i=1}^{n}p^{(I),I^v(i)}(x_{\sigma(i)})
%\end{split}
%\end{equation}
%where
%\begin{equation}\label{para-Mdelta-GMB}
%\begin{split}
%  &p^{(I),\ell}(x)\triangleq p^{(I)}(x,\ell), \ell\in\mathbb{L}.
%\end{split}
%\end{equation}
%\end{Pro}

\begin{Pro}\label{P_LMB}
If a labeled RFS $\bX$ on $\mathbb{X}\times\mathbb{L}$ is an LMB RFS distributed according to (\ref{LMB}), then $X=\mathcal{K}(\bX)$ is distributed as:
\begin{equation}\label{MB-GMB}
\begin{split}
\pi(\{x_1,\!\cdots\!,x_n\})\!=\!\sum_\sigma\!\sum_{I\in \mathcal{F}_n(\mathbb{L})}\!w^{(I)}{\prod}_{i=1}^n p^{(I^v(i))}(x_{\sigma(i)})\\
\end{split}
\end{equation}
where
\begin{equation}\label{para-MB-GMB}
\begin{split}
w^{(I)}&\triangleq w(I), I\in\mathcal{F}(\mathbb{L})\\
p^{(\ell)}(x)&\triangleq p(x,\ell), \ell\in\mathbb{L}.
\end{split}
\end{equation}
\end{Pro}
\proproof{B}
%\noindent \emph{Proof. }
\begin{Rem}
\normalfont{
Propositions~\ref{P_GLMB} and \ref{P_LMB} explicitly describe the relationship between the parameters of GLMB and LMB densities and the parameters of their unlabeled versions, respectively.  Specifically, the unlabeled version of GLMB is a GMB  with $\mathbb{I}=\mathbb{L}$ and $\Phi=\mathbb{C}$;   the unlabeled version of LMB is a GMB with and $\mathbb{I}=\mathbb{L}$ and the discrete space $\Phi$ only has one point. Note that (\ref{MB-GMB}) is an MB density with a set of parameters $\{r^{(\ell)},p^{(\ell)}(x)\}_{\ell\in\mathbb{L}}$ in nature.
%the unlabeled version of $\delta$-GLMB is a GMB with $\mathbb{I}=\mathbb{L}$ and $\Phi=\Xi$;
%the unlabeled version of M$\delta$-GLMB is a GMB with $\mathbb{I}=\mathbb{L}$ and $\Phi=\mathcal{F}({\mathbb{L}})$;
}
\end{Rem}
\subsection{GCI Fusion with GMB Distributions}
Once the labeled multi-object posteriors in GLMB RFS family are marginalized to the GMB densities based on Propositions~\ref{P_GLMB}-\ref{P_LMB}, the subsequent task is to perform  GCI fusion with GMB densities according to the proposed robust solution. We present the formula of GCI fusion with two GMB densities  in this subsection, and for the case of more than two densities, a common method is to perform the pair-wise fusion \cite{GCI-MB, Battistelli}.   Actually, a manipulatable formula for
GCI fusion has two implicit demands: one is  the formula should make the computation of the set integral in (\ref{G-CI}) tractable; the other is  the fused distribution
should belong to the same RFS family of local posteriors or its unlabeled version,  to enable the pair-wise fusion of more than two densities. In the following, we are devoted to deriving an explicit formula for GCI fusion with GMB distributions to satisfy these two demands.

The following Definition 4 first provided in \cite{GCI-MB} describes the degree of separation for different tracks from the view of highest posterior density (HPD) \cite{HPD} region.
With this definition, \cite{GCI-MB} also provides an approximation technique  to  simplify the fractional order exponential
power of MB distribution. Since they are pertinent to the derivation that follows,  we firstly  introduce them here.
\begin{Def}\label{defination_separation}
\normalfont
Consider an MB posterior $\pi=\left\{\left(r^{(\imath)},p^{(\imath)}(\cdot)\right)\right\}_{\imath\in\mathbb{I}}$, where $\mathbb{I}$ is the index set of Bernoulli components. If $\mathbb{X}_\imath$ is the HPD of confidence $\lambda$ for $p^{(\imath)}(\cdot)$, then the Bernoulli components of $\pi(X)$ are said to be  \underline{mutually $\lambda\times 100\%$  separated} if,
$$
\forall\imath\neq\imath^{\prime},  \ \ \mathbb{X}_{\imath}\cap\mathbb{X}_{\imath^{^{\prime}}}=\emptyset.
$$
\end{Def}

The approximation technique proposed in \cite{GCI-MB} is that
if the Bernoulli components of an MB posterior density denoted by $\pi=\left\{\left(r^{(\imath)},p^{(\imath)}(\cdot)\right)\right\}_{\imath\in\mathbb{I}}$ are mutually $\lambda\times 100\%$ separated and $\lambda$ is very close to $1$ (e.g. $\lambda \geqslant 0.9$), then
\begin{equation}\label{fuse-2}
\begin{split}
\pi&(\!\left\{x_1\!,\!\ldots\!,\!x_n\right\}\!)^{\omega}
%\,\,=\sum_{\mathcal{I}_s^n\in {\mathbb{L}_s}^n}Q_s^{\mathcal{I}_s^n}P_s^{\mathcal{I}_s^n}\left((\bx_1,\cdots,\bx_n)\right)\\&
\!\approx\!\!\sum_{\sigma}\!\!\sum_{{\mathcal{I}}\in \mathcal{F}_n(\mathbb{I})}\!\!\!\left[\!w^{(\mathcal{I})}\right]^{\omega}\!\!\left[\!\prod_{i=1}^{n}\!p^{({\mathcal{I}}^v(i))}(\!x_{\sigma(i)})\!\right]^{\omega}\!
\end{split}
\end{equation}
where $w^{(\mathcal{I})}=\prod_{\imath\in \mathcal{I}}{r^{(\imath)}}\prod_{\imath'\in \mathbb{I}-\mathcal{I}}(1-r^{(\imath')})$.

%$\sigma$ denotes one permutation of $\mathcal{I}$, and $\mathcal{I}^v\in \mathbb{I}^{|\mathcal{I}|}$ is a vector constructed by sorting the elements of the set $\mathcal{I}$.

%\subsection{Two-step Approximation}of form (\ref{GMB_ori})
 Consider two GMB distributions $\pi_s$  , $s=1,2$, which are  the unlabeled versions of posteriors output by two sensors in a network, i.e. 
\begin{equation}\label{GMBs}
\begin{split}
  &\pi_s(\{x_1,\cdots,x_2\}) =\\
 &\sum_{\sigma_s}\!\!\sum_{(\mathcal{I},\phi\!)\in\mathcal{F}_n(\mathbb{I}_s)\!\times\!\Phi_s}\!\!\!\omega_s^{(\mathcal{I},\phi)}\prod_{i=1}^{n}p_s^{(\phi),\mathcal{I}^v(i)}\!(x_{\sigma_s(i)}\!), s=1,2.
 \end{split}
\end{equation}
According to GCI fusion rule, by substitution of (\ref{GMBs}) into (\ref{G-CI}), we are faced with  a tough task that is to simplify the following expression involving the fractional order exponential power,
 \begin{equation}\label{challange}
\pi_s^{\omega_s}\!\!=\!\!\!\left[\!\sum_{\sigma}\!\!\sum_{(\mathcal{I},\phi)\in\mathcal{F}_n(\mathbb{I}_s\!)\!\times\!\Phi_s}\!\!\!\!\omega_s^{(\mathcal{I},\phi)}\!{\prod}_{i=1}^{n}p_s^{(\phi),\mathcal{I}^v(i)}(\!x_{\sigma(i)}\!)\!\right]^{\omega_s}.\end{equation}
Obviously, (\ref{challange}) is computationally intractable due to the sum over the discrete space $\mathcal{F}_n(\mathbb{I}_s)\times\Phi$ and over $\sigma$.

Herein, we adopt a two-step approximation strategy that is to approximate the GMB distribution as a more tractable distribution preserving its  key statistical properties firstly and then compute the GCI fusion of the approximated distributions. On one hand, in the RFS based multi-object tracking algorithms, approximating the multi-object posterior as a simple distribution preserving its key statistical properties can usually obtain a great reduction in computation burden with a slight compromise in accuracy.    For instance, in LMB filter, the  multi-object posterior  is approximated as an LMB distribution which preserves its first-order moment, and the performance of LMB filter has been well demonstrated in \cite{LMB_Vo2}. On the other hand, by adopting the approximation technique in \cite{GCI-MB}, the MB density can accommodate   simplification of the fractional order exponential power in (\ref{challange}) manageable as shown in (\ref{fuse-2}).
Inspired by these two aspects, we  seek  the  MB approximation which preserve the first-order moment  of the original GMB distribution in this paper.

For the subsequent development, we firstly give the definition of fusion map.
\begin{Def}\label{defination_fusion_maps}
Without loss of generality, assume that $|\mathbb{I}_1|\leq|\mathbb{I}_2|$.  A fusion map (for the current time) is a function
$\tau: \mathbb{I}_1\!\rightarrow \!\mathbb{I}_2$ such that  $\tau(i)\!=\!\tau(i^{\ast})$ implies $i\!=\!i^{\ast}$.
 The set of all such fusion maps  is called fusion map space denoted by $\mathcal{T}$. The subset of $\mathcal{T}$ with domain $\mathcal{I}$ is denoted by $\mathcal{T}(\mathcal{I})$. For  notation  convenience, we define $\tau(I)\triangleq  
\{\tau(i), i\in I\}$.
\end{Def}
\begin{Rem} Note that each fusion map denotes a hypothesis that the tracks in sensor 2 are one-to-one matching with the tracks  in sensor 1. The fusion map plays the similar role as the measurement-track association map
  in $\delta$-GLMB filter \cite{delta_GLMB}.
  \end{Rem}

The two-step approximation strategy is described  as  follow:

%\begin{enumerate}[$>$]
% \item
\noindent$>$ \ul{\textit{The MB approximation of GMB density preserving the first-order moment:}}  According to (\ref{GMB-PHD}),  the first-order moments of the GMB densities of form (\ref{GMBs}) are given by
 \begin{equation}
\begin{split}
    &\,\,v_s(x) ={\sum}_{\imath\in\mathbb{I}_s}\widetilde  r_s^{\,(\imath)} \widetilde p_s^{\,(\imath)}({x}), s=1,2
 \end{split}
\end{equation}
where
\begin{align}
\label{r MB}
\widetilde{r}_s^{(\imath)}&={\sum}_{\mathcal{I}\in\mathcal{F}(\mathbb{I}_s)}{\sum}_{\phi\in\Phi_s}1_\mathcal{I}(\imath)w_s^{(\mathcal{I},\phi)}\\
\label{p MB}
\widetilde{p}_s^{(\imath)}(x)&=\!{\sum}_{\mathcal{I}\in\mathcal{F}(\mathbb{I}_s\!)}{\sum}_{\phi\in\Phi_s}\!1_\mathcal{I}(\imath)w_s^{(\mathcal{I},\phi)}
p_s^{(\phi),\imath}(x)\bigg/\widetilde{r}_s^{\,(\imath)}
\end{align}

The above expression  can be interpreted as a weighted sum of the densities of all individual tracks, with each track denoted by a Bernoulli RFS $\{(\widetilde r^{(\imath)}_{s},\widetilde p^{(\imath)}_{s})\}$, $\imath\in\mathbb{I}_s$. Hence, the MB distributions that match exactly the first-order moments of  the GMB densities of form (\ref{GMBs})  are 
\begin{equation}\label{MB-fuse}
\widetilde{\pi}_s=\{(\widetilde{r}_s^{\,(\imath)},\widetilde{p}_s^{\,(\imath)}(x))\}_{\imath\in\mathbb{I}_s}, s=1,\,2.
\end{equation}
% \item 

\noindent$>$  \underline{\textit{The GCI fusion with MB approximations:}} By adopting the approximation  technique in (\ref{fuse-2}) proposed in \cite{GCI-MB}, we have on condition that the Bernoulli components of $\widetilde{\pi}_s=\{(\widetilde{r}_s^{(\imath)},\widetilde{p}_s^{(\imath)}(x))\}_{\imath\in\mathbb{I}_s}$ are mutually $\lambda\times 100\%$ separated with $\lambda$ close to 1, the fractional order exponential power of the MB distribution can be  simplified as
 \begin{equation}\label{approx-1-v2}
 \widetilde\pi_s(\!\{x_1,\!\!\cdots\!\!,x_n\}\!)^{\omega_s}\!\approx\!\!\sum_{\sigma_s}\!\!\sum_{\mathcal{I}_s\in\! \mathcal{F}_n(\mathbb{I}_s)}\!\!{\!\!\left[\widetilde w_s^{(\mathcal{I}_s)}\right]}^{\omega_s}\!\!{\left[{\prod}_{i=1}^{n}\widetilde p_s^{(\mathcal{I}_s^v(i))}\!(x_{i})\right]}^{\omega_s}
 \end{equation}
 where $\widetilde w_s^{(\mathcal{I}_s)}=\prod_{\imath\in\mathcal{I}_s}r_s^{(\imath)}\prod_{\imath'\in\mathbb{I}_s-\mathcal{I}_s}(1-\widetilde r_s^{(\imath)})$. By substitution of (\ref{approx-1-v2}) for $s=1,2$ into (\ref{G-CI}), and utilizing Definition~5, according to Proposition 3 given in  \cite{GCI-MB}, the fused density can be computed as
\begin{equation}\label{fuse-1}
\begin{split}
&\pi_\omega(\left\{x_1,\ldots,x_n\right\})\\
 =&\frac{1}{C}\prod_{s=1,2} \sum_{\sigma_s}\!\sum_{\mathcal{I}_s\in\! \mathcal{F}_n(\mathbb{I}_s)}\!{\left[\widetilde w_s^{\,(\mathcal{I}_s)}\right]}^{\omega_s}{\left[{\prod}_{i=1}^{n}\widetilde p_s^{(\mathcal{I}_s^v(i))}\!(x_{i})\right]}^{\omega_s}\\
 =&\sum_{\sigma}\sum_{(\mathcal{I},\tau)\in \mathcal{F}_n(\mathbb{I}_1) \times\mathcal{T}(\mathcal{I})} w_\omega^{(\mathcal{I},\tau)}{\prod}_{i=1}^{n}p_\omega^{(\tau),\mathcal{I}^v(i)}(x_{\sigma(i)})\\
\end{split}
\end{equation}
where
%\begin{align}
%  \label{fuse-w}w_\omega^{(\mathcal{I},\tau)}&=\check{w}_\omega^{(\mathcal{I},\tau)}/C,\\
% \label{fuse-p}p_{\omega}^{(\tau),\imath}(x)&=[\widetilde{p}_1^{(\imath)}(x)]^{\omega_1}[\widetilde{p}_2^{\,(\tau(\imath))}(x)]^{\omega_2}\big/ Z_\omega^{(\tau),\imath},\,\,\imath\in\mathbb{I}_1\\
% \label{Z}
%Z_\omega^{(\tau),\imath}&=\int [\widetilde{p}_1^{(\imath)}(x)]^{\omega_1}[\widetilde{p}_2^{\,(\tau(\imath))}(x)]^{\omega_2}dx\\
%\label{fuse-w-non}\check{w}_\omega^{(\mathcal{I},\tau)}& = [\widetilde w_1^{(\mathcal{I})}]^{\omega_1}[\widetilde  w_2^{(\tau(\mathcal{I}))}]^{\omega_2}{\prod}_{\imath\in \mathcal{I}} Z_\omega^{(\tau),\imath}\\
%\label{fuse-C}
%C&=\!\int \pi_\omega(X)\delta X={\sum}_{(\mathcal{I},\tau)\in\mathcal{F}(\mathbb{I}_1)\times\mathcal{T}(\mathcal{I})}\check{w}^{(\mathcal{I},\tau)}.
%\end{align}
\begin{align}
  \label{fuse-w}w_\omega^{(\mathcal{I},\tau)}&=\check{w}_\omega^{(\mathcal{I},\tau)}/C,\\
% \label{fuse-p}p_{\omega}^{(\tau),\imath}(x)&=[\widetilde{p}_1^{(\imath)}(x)]^{\omega_1}[\widetilde{p}_2^{\,(\tau(\imath))}(x)]^{\omega_2}\big/ Z_\omega^{(\tau),\imath},\,\,\imath\in\mathbb{I}_1\\
% \label{Z}
%Z_\omega^{(\tau),\imath}&=\int [\widetilde{p}_1^{(\imath)}(x)]^{\omega_1}[\widetilde{p}_2^{\,(\tau(\imath))}(x)]^{\omega_2}dx\\
\label{fuse-p}p_{\omega}^{(\tau),\imath}(x)&=\widetilde p^{(\imath,\tau(\imath))}_{\omega} (x)\\
\label{Z}
\eta_\omega^{(\tau),\imath}&=\widetilde\eta_{\omega}^{(\imath,\tau(\imath))}\\
\label{fuse-w-non}\check{w}_\omega^{(\mathcal{I},\tau)}& = [\widetilde w_1^{(\mathcal{I})}]^{\omega_1}[\widetilde  w_2^{(\tau(\mathcal{I}))}]^{\omega_2}{\prod}_{\imath\in \mathcal{I}} \eta_\omega^{(\tau),\imath}\\
\label{fuse-C}
C&=\!\int \pi_\omega(X)\delta X={\sum}_{(\mathcal{I},\tau)\in\mathcal{F}(\mathbb{I}_1)\times\mathcal{T}(\mathcal{I})}\check{w}^{(\mathcal{I},\tau)}.
\end{align}
with 
\begin{align}
\label{p-temp}\widetilde p_{\omega}^{(\imath,\jmath)}(x)=&[\widetilde p^{(\imath)}_{1}(x)]^{\omega_{1}}[\widetilde p^{(\jmath)}_{2}(x)]^{\omega_{2}}\big/ \widetilde \eta_{\omega}^{(\imath,\jmath)},  (\imath,\jmath)\in \mathbb{I}_{1}\times\mathbb{I}_{2}\\
 \label{eta_temp}\widetilde \eta_{\omega}^{(\imath,\jmath)}=&\int  [\widetilde p^{(\imath)}_{1}(x)]^{\omega_{1}}[\widetilde p^{(\jmath)}_{2}(x)]^{\omega_{2}} dx,  (\imath,\jmath)\in \mathbb{I}_{1}\times\mathbb{I}_{2} 
\end{align}
%\begin{align}
%\label{fuse-C}
%C=&\int \pi_\omega(X)\delta X=\sum_{(\mathcal{I},\tau)\in\mathcal{F}(\mathbb{I}_1)\times\mathcal{T}(\mathcal{I})}\check{w}^{(\mathcal{I},\tau)}.
%\end{align}
%As a result, (\ref{G-CI}) can be represented as 
%\begin{equation}\label{fuse-1}
%\begin{split}
%&\pi_\omega(\left\{x_1,\ldots,x_n\right\})\\
%  =&\sum_{\sigma}\sum_{(\mathcal{I},\tau)\in \mathcal{F}_n(\mathbb{I}_1) \times\mathcal{T}(\mathcal{I})} w_\omega^{(\mathcal{I},\tau)}\prod_{i=1}^{n}p_\omega^{(\tau),\mathcal{I}^v(i)}(x_{\sigma(i)})\\
%\end{split}
%\end{equation}
%with 
%\begin{equation}
%\label{fuse-w}w_\omega^{(\mathcal{I},\tau)}=\frac{1}{C}\widetilde\omega^{(\mathcal{I},\tau)}.
%\end{equation}
%\end{enumerate}

Finally, via the aforementioned two-step approximation, we derive the explicit formula for GCI fusion with  two GMB distributions as shown in (\ref{fuse-1}). It can be observed that (\ref{fuse-1}) is another GMB distribution, 
it can enable the distributed fusion with GLMB RFS family in a sensor network owning  more than two sensors by applying pair-wise fusion.

\begin{Rem}
\normalfont{
Usually the true single object state corresponding to each Bernoulli component $\{(\widetilde r^{(\imath)}_{s},\widetilde p^{(\imath)}_{s})\}$ determines the center of its HPD region. Furthermore, the width of HPD region of a Bernoulli component is smaller with lower maneuverability and higher SNR. In such practical scenarios, the MB approximation preserving the first order moment of the original GMB posterior  can be easily assumed to be mutually separated with very high confidence, and the condition of the approximation in (\ref{approx-1-v2}) can be easily satisfied.}
\end{Rem}
%\vspace{-3mm}
\subsection{Construction of the Labeled Fused Posterior}
After fusing the unlabeled densities and obtained the GMB fused density $\pi_\omega$, one  can extract the multiple object states directly from the fused GMB density if the trajectories are not required; otherwise if the system does need distinguish object identities and form the trajectories, we can also construct the labeled density  based on 
(\ref{fused-label-distribution}) and (\ref{construction-fused-label}). 

The conditional multi-label distribution $\varpi_{s_0}(\cdot)$ in (\ref{fused-label-distribution})  can be computed from  $\bpi_{s_0}(\bX)$, i.e.,  the local labeled density of sensor $s_0$, by combination of (\ref{unlabel-marginal}) and  (\ref{varpi}), however, it is hard to obtain an analytic solution even for GLMB densities.  As a result, the  calculations of (\ref{fused-label-distribution}) and (\ref{construction-fused-label}) are even more difficult. However, we find that if  $\varpi_{s_0}(\cdot)$ is replaced by $\overline\varpi_{s_0}(\cdot)$ of the first-order moment preserved LMB approximation of $\bpi_{s_0}(\bX)$, the calculation of (\ref{construction-fused-label}) can be dramatically simplified, and more importantly,   the labeled fused density is another GLMB.

%\begin{equation}\label{GLMB-s0}
%\begin{split}
%\bpi(\bX)=\Delta(\bX)\sum_{c\in\mathbb{C}}w^{(c)}(\mathcal{L}(\bX))[p^{(c)}]^\bX
%\end{split}
%\end{equation}
 Assume that the labeled posterior $\bpi_{s_0}(\bX)$ is a GLMB density of form (\ref{GLMB}), i.e.,
 \begin{small}
\begin{align}\label{GLMB-2}
\begin{split}
\bpi_{s_0}(\bX)=\Delta(\bX){\sum}_{c\in\mathbb{C}_{s_0}}w_{s_0}^{(c)}(\mathcal{L}(\bX))[p_{s_0}^{(c)}]^\bX
\end{split}
\end{align}
\end{small}
According to Proposition 4 in \cite{delta_GLMB}, the  LMB approximation which preserves the first-order moment of $\bpi_{s_0}(\bX)$ is  
\begin{equation}
\overline\bpi_{s_0}(\bX)=\{(\overline{r}_{s_0}^{(\ell)},\overline{p}^{(\ell)}_{s_0}(x))\}_{\ell\in\mathbb{L}_{s_0}}
\end{equation}
with 
%\label{r_s0}
%
\begin{align}
\label{r_s0}  \!\!\!\overline r_{s_0}^{(\ell)}&\!=\!\!{\sum}_{I\in \mathcal{F}(\mathbb{L}_{s_0})}{\sum}_{c\in\mathbb{C}_{s_0}}1_{I}(\ell)w_{s_0}^{(c)}(I) ,\\
\label{p_s0} \!\!\!\overline p_{s_0}^{(\ell)}({x})&\!=\!\!{\sum}_{I\in \mathcal{F}(\mathbb{L}_{s_0}\!)}{\sum}_{c\in\mathbb{C}_{s_0}}\!\!\!1_{I}(\ell)w_{s_0}^{(c)}(I) p_{s_0}^{(c)}(x,\ell)\!\!\bigg/\overline r_{s_0}^{(\ell)}
\end{align} 
 According to~(\ref{unlabel-marginal}) and \eqref{varpi}, the conditional joint probability density of $\ell_1,\cdots,\ell_n$ of $\overline \bpi_{s_0}(\bX)$ can be computed by  
\begin{equation}
\begin{split}
&\overline \varpi_{s_0}(\{(\ell_1|x_1),\cdots,(\ell_n|x_n)\})=\\
&\frac{\overline w_{s_0}(\{\ell_1,\cdots,\ell_n\})\prod_{i=1}^n\overline p_{s_0}^{(\ell_i)}(x_i)}{\sum_{(\ell_1,\cdots,\ell_n)\in\mathbb{L}_{s_0}^n} \overline w_{s_0}(\{\ell_1,\cdots,\ell_n\})\prod_{i=1}^n\overline p_{s_0}^{(\ell_i)}(x_i)} 
\end{split}
\end{equation}
where $\overline w_{s_0}(I)=\prod_{\ell\in I}{\overline r_{s_0}^{(\ell)}}\prod_{\ell'\in \mathbb{L}_{s_0}-I}(1-\overline r_{s_0}^{(\ell')})$.
Replacing $\varpi_{s_0}(\cdot)$ with $\overline \varpi_{s_0}(\cdot)$, according to (\ref{fused-label-distribution}), by  simple deduction we can  obtain the corresponding $\overline \varpi_\omega(\cdot)$  as
\begin{equation}\label{app-fused-label}
\begin{split}
&\overline \varpi_{\omega}(\{(\ell_1|x_1),\cdots,(\ell_n|x_n)\})=\\
&\frac{\left[\overline w_{s_0}(\{\ell_1,\!\cdots,\!\ell_n\})\right]^{\omega_{s_0}}\prod_{i=1}^n\left[\overline p_{s_0}^{(\ell_i)}(x_i)\right]^{\omega_{s_0}}}{\sum_{(\ell_1,\cdots,\ell_n)\in\mathbb{L}_{s_0}^n}\!\left[\overline w_{s_0}(\{\ell_1,\!\cdots\!,\ell_n\})\right]^{\omega_{s_0}}\!\prod_{i=1}^n\!\left[\overline p_{s_0}^{(\ell_i)}(x_i)\!\right]^{\omega_{s_0}}}.
\end{split}
\end{equation}

Without loss of generalisation,  let $s_0=1$, we can obtain the following proposition: 
\begin{Pro}
If the   labeled fused density is constructed as
\begin{equation}\label{label-density-0}
\begin{split}
\bpi_\omega(&\{(x_1,\ell_1),\cdots,(x_n,\ell_n)\})=\\
&\overline\varpi_\omega(\{(\ell_1|x_1),\cdots,(\ell_n|x_n)\})\pi_\omega(\{x_1,\cdots,x_n\})
\end{split}
\end{equation}
where $\pi_\omega(\cdot)$ and $\overline \varpi_w(\cdot)$ is given in (\ref{fuse-1}) and (\ref{app-fused-label}) (with $s_0=1$) respectively, then the labeled fused density is given by
\begin{equation}\label{fused GLMB}
\bpi_\omega(\bX)
=\!\!\Delta (\mathbf{X})\!{\sum}_{\tau \in \mathcal{T}}1_{\mathcal{T}(\mathcal{L}(\bX))}(\tau)w_\omega^{(\tau)}(\mathcal{L}(\bX)) {[p_\omega^{(\tau)}]}^\bX
\end{equation}
defined on $\mathbb{L}_\omega=\mathbb{I}_1=\mathbb{L}_1$, where
\begin{align}
\label{p_identity} p_\omega^{(\tau)}(\cdot,\ell)\triangleq&p_\omega^{(\tau),\ell}(\cdot)\\
\label{w_identity} w_\omega^{(\tau)}(I)\triangleq&w_\omega^{(I,\tau)}.
\end{align}
\end{Pro}
\proproof{C}
%\begin{proof}
%As given in the previous section,  the fused unlabeled density $\pi_\omega(\cdot)$ is the GCI fusion with $\widetilde \pi_s(\cdot)$ in  (\ref{MB-fuse}), $s=1,2$ after using the approximation technique in (\ref{approx-1-v2}), i.e. 
%\begin{equation}
%\begin{split}
%&\pi_\omega(\{x_1,\cdots,x_n\})=\\
%&\frac{1}{C}\prod_{s=1,2} \sum_{\sigma_s}\sum_{\mathcal{I}_s\in\! \mathcal{F}_n(\mathbb{I}_s)}{\left(Q^{\mathcal{I}_s}\right)}^{\omega_s}{\left(\prod_{i=1}^{n}p_s^{\mathcal{I}_s^v(i)}\!(x_{i})\right)}^{\omega_s}
%\end{split}
%\end{equation}
\begin{Rem}
Note that constructing the labeled fused density in (\ref{fused GLMB}) from the fused GMB density in (\ref{fuse-1}) does not require any additional computation because (\ref{fused GLMB})  is completely
determined by the parameters of fused GMB density according to Proposition 4. Moreover, it is obvious that  the constructed GLMB fused density has the same cardinality distribution and the (unlabeled) first-order moment as the fused GMB  density. 
\end{Rem}
\subsection{Summary and Discussions}
In this section, we proposed a novel  distributed fusion solution  for label densities. More specifically, we proposed a  robust GCI fusion algorithms for GLMB  densities (R-GCI-GLMB) according to the proposed solution. The schematic of the  R-GCI-GLMB fusion algorithm is shown in Fig. 4.  One should note that R-GCI-GLMB fusion algorithm also suitable  for the fusion with LMB filters or M$\delta$-GLMB filters. For the LMB filter, approximating GMB density as a first-order moment preserved MB density is not required, because the unlabeled version of an LMB posterior is just an MB density. After the fused GLMB density is constructed, if feedback is required, the fused GLMB density should be approximated as the same class of local posterior, such as the LMB or M$\delta$-GLMB densities.
\begin{figure}[h]
\centering
\includegraphics[width=3.5in]{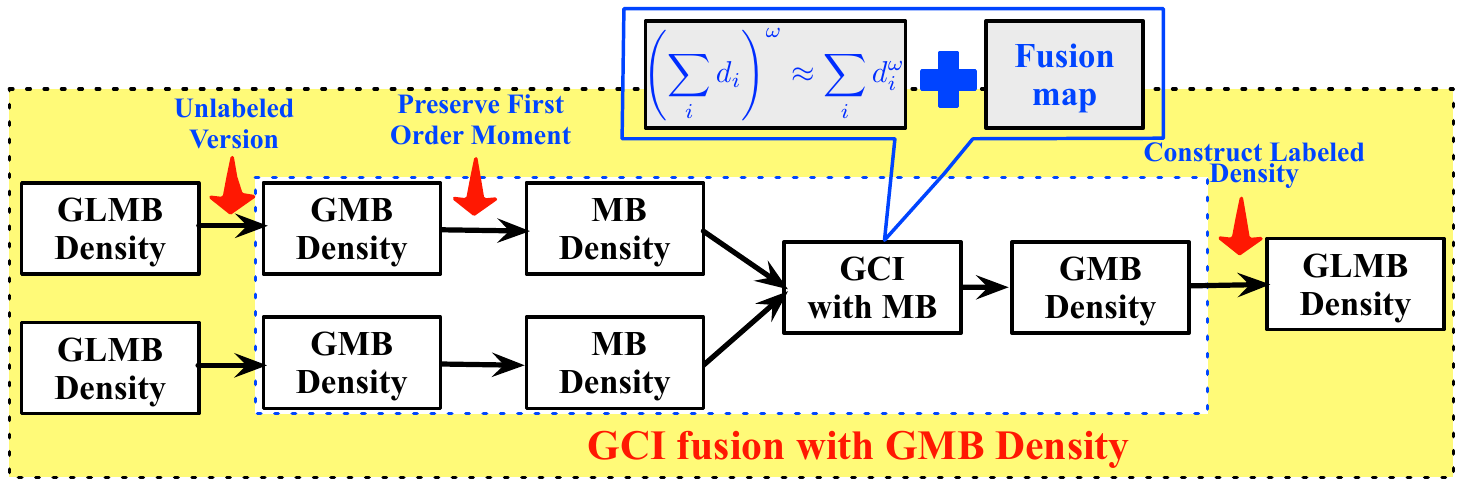}
\caption{\small{A schematic diagram of the R-GCI-GLMB fusion algorithm.}}
\label{relation_GLMB_GMB}
\end{figure}

In the following, we provide the pseudocode of the R-GCI-GLMB fusion algorithm  in three parts (algorithms). \text{Algorithm 1}  shows the pseudocode for approximating a GMB density as an MB density  matching the first-order moment.
\text{Algorithm 2}  shows the pseudocode of GCI fusion with GMB densities from $N_{s}$ ($N_s\geqslant2$) sensors, where the pair-wise fusion strategy is adopted and the ordering of pair-wise fusions is irrelevant. Taking the case that local sensors perform   GLMB filtering as an example,  \text{Algorithm 3} shows  the pseudocode of the whole fusion algorithm.

The computational cost of the R-GCI-GLMB algorithm  mainly lies in the GCI fusion of GMB densities, since   marginalizing the GLMB density to its unlabelled version and constructing the fused  GLMB density from  the fused GMB density are only conceptual operations with no need to  calculate any quantities. 
Taking the Gaussian mixture (GM) implementation as an example, we analyse the computational complexity of the R-GCI-GLMB fusion algorithm  by the following comparison of the GCI fusion with  the CPHD filter (GCI-CPHD) proposed in \cite{Battistelli,Uney-2}. 
 Suppose that  each single-object density for the local GLMB filter   is approximated by $M_{G}$ Gaussian components, and the location density of each local CPHD filter is approximated by $n_{\max}M_{G}$ Gaussian components, where $n_{\max}$ denotes the maximum number of objects. The computational cost of the R-GCI-GLMB fusion mainly depends on two parts: one is for computing the fused single-object densities $\widetilde p_{\omega}^{(\imath,\jmath)}$ which has the computational complexity  $\mathcal{O}\left[(|\mathbb{I}_{\max}|M_{G})^{N_{s}}\right]$, where $\mathbb{I}_{\max} =\arg \max_{\mathbb{I}\in \{\mathbb{I}_{s}:s\in\mathcal{N}\}}|\mathbb{I}|$; the other part is for the calculation of the fused weights $w_{\omega}^{(\mathcal{I},\tau)}$ which has the computational complexity $\mathcal{O}\left[N_{s}N_{H}\right]$ where $N_{H}=|\mathcal{F}(\mathbb{I}_{\max})\times\mathcal{T}_{\max}|$ with   $\mathcal{T}_{\max}$ being  the fusion map space defined on $\mathbb{I}_{\max}$. Overall, the R-GCI-GLMB fusion has the   computational complexity $\mathcal{O}\left[\max\{N_s N_{H}, (|\mathbb{I}_{\max}|M_{G})^{N_{s}}\}\right]$. By contrast,  the  computational complexity of the GCI-CPHD fusion is $\mathcal{O}\left[(n_{\max} M_{G})^{N_{s}}\right]$. 
As it is common with the
 GM
 implementation of stochastic filters, the fused Gaussian components with negligible coefficients are pruned to reduce the computational cost when implementing the GCI-CPHD fusion  \cite{Battistelli}.
When implementing the R-GCI-GLMB fusion, we  apply the same strategy  to reduce the computational cost of the first part. Moreover, to keep the computational cost of the second part moderate, some common efficient  strategies,  such as the truncation of the fused GMB density with the ranked assignment strategy \cite{GLMB,LMB_Vo2}, are adopted. 
%The proposed R-GCI-GLMB fusion can be easily extended to $N_s\geqslant2$ sensors by sequentially applying the three-step  strategy  $N_s-1$ times, where the ordering of pairwise fusions is irrelevant. Similar approach has been widely used in distributed fusion, such as GCI fusion with CPHD filters \cite{Battistelli}, MB filters \cite{GCI-MB} and LMB filters \cite{Fantacci-BT}. The pseudocode of R-GCI-GLMB fusion is given in \textbf{Algorithm 1}.
%
%
%
%
%The implementation of R-GCI-GLMB fusion is equivalent to : Step 1. Transform the labeled RFSs to its unlabeled version; Step 2. Calculate the first order moment matching approximation; Step 3. Calculate GCI fusion with MB densities. 
%
\begin{algorithm}[h]\label{algorithm: R-GCI-GLMB}
\caption{Approximate a GMB density as an MB density matching the first-order moment.}
%\caption{GMB2MB.}
%\begin{minipage}{0.9\columnwidth}
\small
{\underline{INPUT:}  A GMB density with the parameter set  $\pi_{s}=\{(w_s^{(\mathcal{I},\phi)},\mathfrak{P}_s^{(\phi)}):(\mathcal{I},\phi)\in \mathcal{F}(\mathbb{I}_s)\times\Phi_s\}$\;

\underline{OUTPUT:} An MB density with  the parameter set  $\widetilde{\pi}_s=\!\{(\widetilde{r}_s^{(\imath)},\widetilde {p}_s^{(\imath)}(\cdot))\}_{\imath\in\mathbb{I}_s}$.

\textbf{function} GMB2MB\,($\pi_{s}$)

\For{$\imath\in\mathbb{I}_s$}
{
$\widetilde r_s^{(\imath)}:=\sum_{\mathcal{I}\in\mathcal{F}(\mathbb{I}_s)}\sum_{\phi\in\Phi_s}1_\mathcal{I}(\imath)w_s^{(\mathcal{I},\phi)}$\;
$\widetilde p_{s}^{(\imath)}:=\sum_{\mathcal{I}\in\mathcal{F}(\mathbb{I}_s\!)}\sum_{\phi\in\Phi_s}\!1_\mathcal{I}(\imath)w_s^{(\mathcal{I},\phi)}
p_s^{(\phi),\imath}(x)\bigg/\widetilde r_s^{\,(\imath)}$
}
\textbf{return} $\widetilde{\pi}_s=\!\{(\widetilde{r}_s^{(\imath)},\widetilde {p}_s^{(\imath)}(\cdot))\}_{\imath\in\mathbb{I}_s}$.}%\end{minipage}
\end{algorithm}
\begin{algorithm}[h]\label{algorithm: R-GCI-GLMB}
\caption{GCI fusion with GMB Densities.}
\small{
%\caption{GCI\_GMB\_Fusion.}
\underline{INPUT: } $N_{s}$ GMB Densities    with the repective parameter set $\pi_s=\{(w_s^{(\mathcal{I},\phi)},\mathfrak{P}_s^{(\phi)}):(\mathcal{I},\phi)\in \mathcal{F}(\mathbb{I}_s)\times\Phi_s\}, s=1,\cdots,N_{s}$\;
\underline{OUPUT:} The fused GMB density  with the  parameter set $\pi_\omega=\{(w_\omega^{(\mathcal{I},\tau)},\mathfrak{P}_\omega^{(\tau)}):(\mathcal{I},\tau)\in \mathcal{F}(\mathbb{I}_1)\times\mathcal{T}\}$.

\textbf{function} GCI\_GMB\_Fusion\,($\pi_{1},\cdots,\pi_{N_{s}}$)

\For{$s=2:N_s$}
{
 $\{(\widetilde{r}_s^{(\imath)},\widetilde{p}_s^{(\imath)}(\cdot))\}_{\imath\in\mathbb{I}_s}$: = GMB2MB\,($\pi_s$)\;
 $\{(\widetilde{r}_1^{(\imath)},\widetilde{p}_1^{(\imath)}(\cdot))\}_{\imath\in\mathbb{I}_1}$: = GMB2MB\,($\pi_1$)\;
 \For{$(\imath,\jmath)\in\mathbb{I}_{1}\times\mathbb{I}_{s}$}
 {
Calculate $\widetilde p_{\omega}^{(\imath,\jmath)}$ according to (\ref{p-temp})\;
Calculate $\widetilde \eta_{\omega}^{(\imath,\jmath)}$ according to (\ref{eta_temp})\;
}
Create a fusion map space $\mathcal{T}_{s}:=\{\tau| \tau:\mathbb{I}_{1}\rightarrow\mathbb{I}_{s}\}$ according to \textbf{Definition 6}\;
%$\mathfrak{P}_{\omega}^{(\tau)}=\emptyset$\;
\For{$\tau\in\mathcal{T}_{s}$}
{
%%\For{$\imath\in\mathbb{I}_{1}$}
%%{$p_{\omega}^{(\tau),\imath}=\widetilde p_{\omega}^{(\imath,\tau(\imath))}$\;
%%$\mathfrak{P}_{\omega}^{(\tau)}=\mathfrak{P}_{\omega}^{(\tau)}\cup p_{\omega}^{(\tau),\imath}$\;
$\mathfrak{P}_{\omega}^{(\tau)}:=\{\widetilde p_{\omega}^{(\imath,\tau(\imath))}\}_{\imath\in\mathbb{I}_{1}}$\;
%}
}
%\For { $\tau \in \mathcal{T}$}
%{ 
%\For{$\imath\in\mathbb{I}_1$}  
%{
%% Evaluate  $\widetilde{p}_{\omega}^{(\tau),\imath}$ according to (\ref{fuse-p})\;
% Evaluate  $Z_\omega^{(\tau),\imath}$ according to (\ref{Z})\;
%  Evaluate  $p_{\omega}^{(\tau),\imath}$ according to (\ref{fuse-p})\;
%}
%$\mathfrak{P}_\omega^{(\tau)}=\{p_{\omega}^{(\tau),\imath}\}_{\imath\in\mathbb{I}_1}$\;
%}

\For{$(\mathcal{I},\tau)\in\mathcal{F}(\mathbb{I}_1)\times\mathcal{T}(\mathcal{I})$}  
{
%\For{$\imath \in \mathcal{I}_1 (\tau(\imath)\in \mathcal{T}(\mathcal{I}))$}
%{
%$Z_\omega^{(\tau),\imath}=\int \widetilde{p}_{\omega}^{(\tau),\imath} dx$\;
%$p_{\omega}^{(\tau),\imath}=\widetilde{p}_{\omega}^{(\tau),\imath}\big{/}Z_\omega^{(\tau),\imath} $\;
%}
Calculate the un-normalized weight $\check{w}_\omega^{(\mathcal{I},\tau)}$  according to (\ref{Z}) and (\ref{fuse-w-non})\;
}
Calculate the normalized factor $C$ according to  (\ref{fuse-C})\;
\For{$(\mathcal{I},\tau)\in\mathcal{F}(\mathbb{I}_1)\times\mathcal{T}(\mathcal{I})$}  
{
Calculate the normalized   weight $w_\omega^{(\mathcal{I},\tau)}$ according to (\ref{fuse-w})\;
}
$\pi_1: =\{(w_\omega^{(\mathcal{I},\tau)},\mathfrak{P}_\omega^{(\tau)}):(\mathcal{I},\tau)\in \mathcal{F}(\mathbb{I}_1)\times\mathcal{T}\}$\;}

$\pi_\omega=\pi_1$\;

\textbf{return}: $\pi_{\omega}$}
\end{algorithm}
%\begin{algorithm}[tb]\label{algorithm: R-GCI-GLMB}
%\caption{R-GCI-GLMB Fusion Algorithm.}
%\small{
%%\begin{minipage}{0.9\columnwidth}
%\underline{INPUT:} GLMB Densities from $N_{s}$ sensors with parameter sets $\bpi_s=\{(w^{(c)}(I),p_s^{(c)}):(I,c)\in\mathcal{F}(\mathbb{L}_s)\times\mathbb{C}_s\}$, $s=1,\cdots,N_s$ \;
%\underline{OUPT:} The fused GLMB density with parameter set $\bpi_\omega=\{(w_\omega^{(\tau)}(I),p_\omega^{(\tau)}):(I,\tau)\in\mathcal{F}(\mathbb{L}_\omega)\times\mathcal{T}\}$.
%
%
%Marginalize labeled RFS to its unlabeled versions $\pi_s=\{(w^{(I,c)},\mathfrak{P}_s^{(\phi)}:(I,c)\in\mathcal{F}(\mathbb{L}_s)\times\mathbb{C}_s\}, s=1,\cdots,N_s$:
%$w_s^{(I,c)} \triangleq w_s^{(c)}(I)$, $ \mathfrak{P}_s^{(\phi)}\triangleq \{p_s^{(c),\ell}(x)\}_{\ell\in\mathbb{L}_s}$ with  $p_s^{(c),\ell}(x) \triangleq p_s^{(c)}(x,\ell)$\;
%
%$\{(w_\omega^{(\mathcal{I},\tau)},\mathfrak{P}_\omega^{(\tau)}):(\mathcal{I},\tau)\in \mathcal{F}(\mathbb{I}_1)\times\mathcal{T}\}$: =GCI\_GMB\_Fusion($\pi_1,\cdots,\pi_{N_s}$)\;
%Construction of the labeled fused posterior $\bpi_\omega=\{(w_\omega^{(\tau)}(I),p_\omega^{(\tau)}):(I,\tau)\in\mathcal{F}(\mathbb{L}_\omega)\times\mathcal{T}\}$: $\mathbb{L}_\omega=\mathbb{I}_1$,  $p_\omega^{(\tau)}(\cdot,\ell)\triangleq p_\omega^{(\tau),\ell}(\cdot)$, and $ w_\omega^{(\tau)}(I)\triangleq w_\omega^{(I,\tau)}$\;
% }
%\end{algorithm}
\begin{algorithm}[h]\label{algorithm: R-GCI-GLMB}
\caption{R-GCI-GLMB fusion algorithm.}
\small{
%\begin{minipage}{0.9\columnwidth}
\underline{INPUT:} GLMB densities from $N_{s}$ sensors, $\bpi_s$, $s=1,\cdots\!,N_s$ \;
\underline{OUPT:} The fused GLMB density $\bpi_\omega$.

\textbf{function} main\_loop\,($\bpi_{1},\cdots,\bpi_{s}$)

\For{ $s=1:N_{s}$}{
%Calculate the parameters of the unlabeled version of the GLMB density $\bpi_{s}$, denoted as a set $\pi_{s}$ according to \textbf{Proposition} 2. 

Marginalize the GLMB density $\bpi_{s}$ to its unlabeled version $\pi_{s}$ according to \textbf{Proposition} 2;}

$\pi_{\omega}$: =\,GCI\_GMB\_Fusion\,($\pi_1,\cdots,\pi_{N_s}$)\;
Construct  the labeled fused posterior $\bpi_\omega$ according to \textbf{Proposition} 4\;
 
\textbf{return} $\bpi_{\omega}$}
\end{algorithm}

\section{Performance Assessment}\label{chp:6}
In this section, the performance of the proposed R-GCI-GLMB fusion is examined by comparison with the state of art over  distributed sensor  networks.
%Gaussian mixture
The GM implementations are adopted for all the distributed tracking algorithms.
%We choose the Metropolis weights [?]in the GCI fusion for convenience.

%$ and the parameters $\omega_s$ in (\ref{G-CI}) are chosen as 0.5 for $s=1,2$$

The standard object and observation models \cite{LMB_Vo2} are used. The object state variable is a vector of planar position and velocity $x_k=[p_{x,k}\,\,p_{y,k}\,\,\dot{p}_{x,k}\,\,\dot{p}_{y,k}]^{\top}$, where ``$^\top$'' denotes matrix transpose. The single-object transition model is the linear Gaussian $$f_k(x_k|x_{k-1}) = \mathcal{N}(x_k; \bF_k x_{k-1},\bQ_k)$$ with its parameters given for a nearly constant velocity model: 
%The number of objects is time varying due to births and deaths. 
% with the following values
\[
\displaystyle{
 \bF_k=\left[
\begin{array}{cc}
\bI_2& \Delta \bI_2 \\ \mathbf{0}_2 & \bI_2
\end{array}
\right],\,\,\,\,\,\,\, \bQ_k=\sigma_v^{2}\left[
\begin{array}{cc}
\frac{1}{4}\bI_2& \frac{1}{2}\Delta \bI_2 \\ \frac{1}{3}\mathbf{0}_2 & \bI_2
\end{array}
\right]
}
\]
where $\bI_2$ and $\mathbf{0}_2$ denote the $2\times 2$ identity and zero matrices, $\Delta =1$\,s is the sampling period, and $\sigma_\nu=5\,\text{m/s}^{2}$ is the standard deviation
of the process noise. The  state independent survival  probability of the object  is given by $P_{S,k}=0.98$.

Two types of birth procedures are considered in different experiments. One is the prior knowledge-based birth procedure. At each time $k$, the birth process is an LMB RFS with the parameter set $\bpi_{B}=\{(r_{B}^{(k,i)}, p_{B}^{(k,i)})\}_{i=1}^{3}$ where $r_{B}^{(k,i)}=0.04$ and $p_{B}^{(k,i)}=\mathcal{N}(x; m_{B}^{(i)}, P_{B})$ with $m_{B}^{(1)}=[200\,\, 400\,\,0\,\,0]^\top$, $m_{B}^{(2)}=[-\!150\,-\!\!310\,\,\,0\,\,0]^\top$, $m_{B}^{(3)}=[0\,\,400\,\,0\,\,0]^\top$, and $P_{B}=\mbox{diag}([900\,\,900\,\,400\,\,400])$.

 The other birth procedure  is the adaptive birth procedure proposed in \cite{delta_GLMB}. The LMB birth process at   time step $k+1$￼ ￼ depends on the measurement ￼ set $Z$￼ of the current time step $k$ and is given by $$\bpi_{B}=\{r_{B}^{(k+1,i)}(z), p_{B}^{(k+1,i)}(x|z)\}_{z\in Z}.$$ More specifically, the existence probability $r_{B}^{(k+1,i)}(z)$ is proportional to the probability that $z$ is not assigned to any track during the update at time step $k$:
% depending on the measurement $z\in Z^{k}$
\begin{align} \notag
%\label{adaptive_birth}
r_{B}^{(k+1,i)}(z)=\min\left(r_{B,\max},\frac{1-r_{U,k}(z)}{\sum_{\xi\in Z}1-r_{U,k}(\xi)}\cdot \lambda_{B,k+1}\right)
\end{align}
%where
%\begin{align}\label{adaptive_birth_r_1}
%r_{U,k}=\sux_{(I_+,\theta)\in \mathcal{F}(\mathbb{L}_+)\times \Theta_{I_+}}1_{\theta}(z) w_k^{(I_+,\theta)}
%\end{align}
%with the weight $w_k^{(I_+,\theta)}$ is a quantity for the local filter given by Eq. (59) of \cite{delta_GLMB}, 
where $r_{U,k}(z)$ denotes the  probability that  a measurement $z$ ￼ is associated to a track in the hypotheses  at time step $k$, $\lambda_{B,k+1}$ is the expected number of object births at time step  $k+1$, and $r_{B,\max}\in[0, 1]$ is the maximum existence probability of a newly born object. Each density $p_{B}^{(k+1,i)}(x|z)=\mathcal{N}(x; m_{B}(z), P_{B})$ with $m_{B}(z)=[z(1)\,\, z(2)\,\,0\,\, 0]^\top$, $P_{B}=\mbox{diag}([900\,\,900\,\, 400\,\, 400])$. The parameters $\lambda_{B,k+1}$ and $r_{B,\max}$ are set to be $0.8$ and $0.3$, respectively. The details about how to compute  the probability $r_{U,k}(z)$ are given in \cite{delta_GLMB}.

Each sensor node detects an object independently with the same probability $P_{D,k}$.  The single-object observation model is linear Gaussian 
$$g_k(z|\bx_{k}) = \mathcal{N}(z; \bH_k \bx_{k},\bR_k)$$
with parameters
\[
\displaystyle{
 \bH_k=\left[
\begin{array}{cc}
\bI_2  & \mathbf{0}_2
\end{array}
\right],
\,\,\,\,\,\,\,\mathbf{R}_k=\sigma_\varepsilon^{2}\bI_2
}
\]
where $\sigma_\varepsilon=25$\,m is the standard deviation of the measurement noise. The number of clutter reports
in each scan is Poisson distributed with $\lambda=10$. Each clutter report is sampled uniformly over the whole surveillance region.

The optimal sub-pattern assignment (OSPA) error \cite{MeMBer_Vo1}
serves as the main performance metric with the cut-off value $c=100$ m and the order parameter $p=1$. 
All performance metrics are averaged over 200 Monte Carlo (MC) runs.

\subsection{Scenario 1}
The robustness of  the proposed R-GCI-GLMB fusion algorithm is verified  by comparison with the classical GCI fusion of LMB posteriors (C-GCI-LMB) \cite{Fantacci-BT} in two experiments with  the ABP and the PBP used respectively. To this end, we consider a scenario involving three objects on a two dimensional surveillance region $[-500, 500]\,\text{m}\times[-500,500] \,\text{m}$, which is shown in Fig.~\ref{scenario_1}. For both fusion algorithms, the LMB filter is chosen as the local filter. 
For  GM implementations of  local LMB filter and  fusion algorithms, the parameters are chosen as follows: the truncation threshold for Bernoulli components is  $\gamma_t=10^{-4}$; pruning  and merging thresholds for Gaussian components are $\gamma_p=10^{-5}$ and  $\gamma_m=4$, respectively; the maximum number of Gaussian components is $N_{\max}=10$.
The duration of this scenario is $T=65$\,s. The probability of detection $P_{D,k}$ for each sensor is 0.99.

\begin{figure}
\centering
\includegraphics[width=6cm]{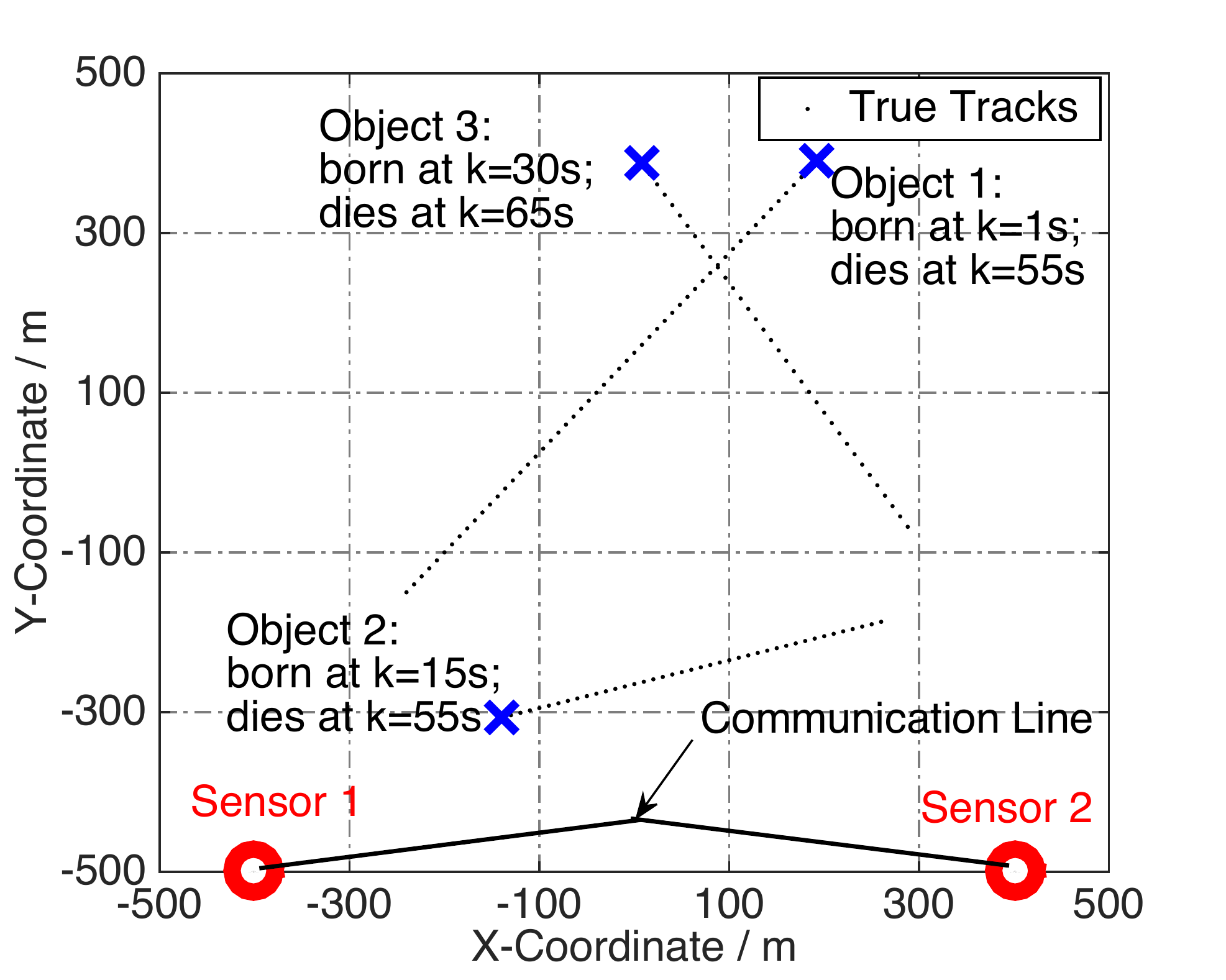}
\caption{The scenario of a distributed  sensor network with two sensors tracking three objects.\label{scenario_1}}
\end{figure}

\noindent\textit{\underline{Experiment 1}}: The performance metrics  for  R-GCI-GLMB  and  C-GCI-LMB fusions in presence of an ABP are shown in Fig. \ref{fig:ABP}.  Specifically, the cardinality estimates (Est.) and the corresponding standard deviations (Std.)  are presented in Fig. \ref{fig:ABP}(a), while the OSPA errors   are provided in Fig. \ref{fig:ABP}(b).

  Not surprisingly,  C-GCI-LMB fusion completely fails (returns highly erroneous estimates) when the ABP is in place due to the resulting inconsistencies between label assignments in each local filters and the reliance of labeled GCI fusion on label consistency between filters. This is while R-GCI-GLMB  fusion leads to errors that are significantly lower than errors returned by each local filter after each transient.  These results  highlight the robustness of the proposed R-GCI-GLMB algorithm when label mismatches happen.
\begin{figure}
	\begin{minipage}[!htb]{0.49\linewidth}
		\centering
		\centerline{\includegraphics[width=4.94cm]{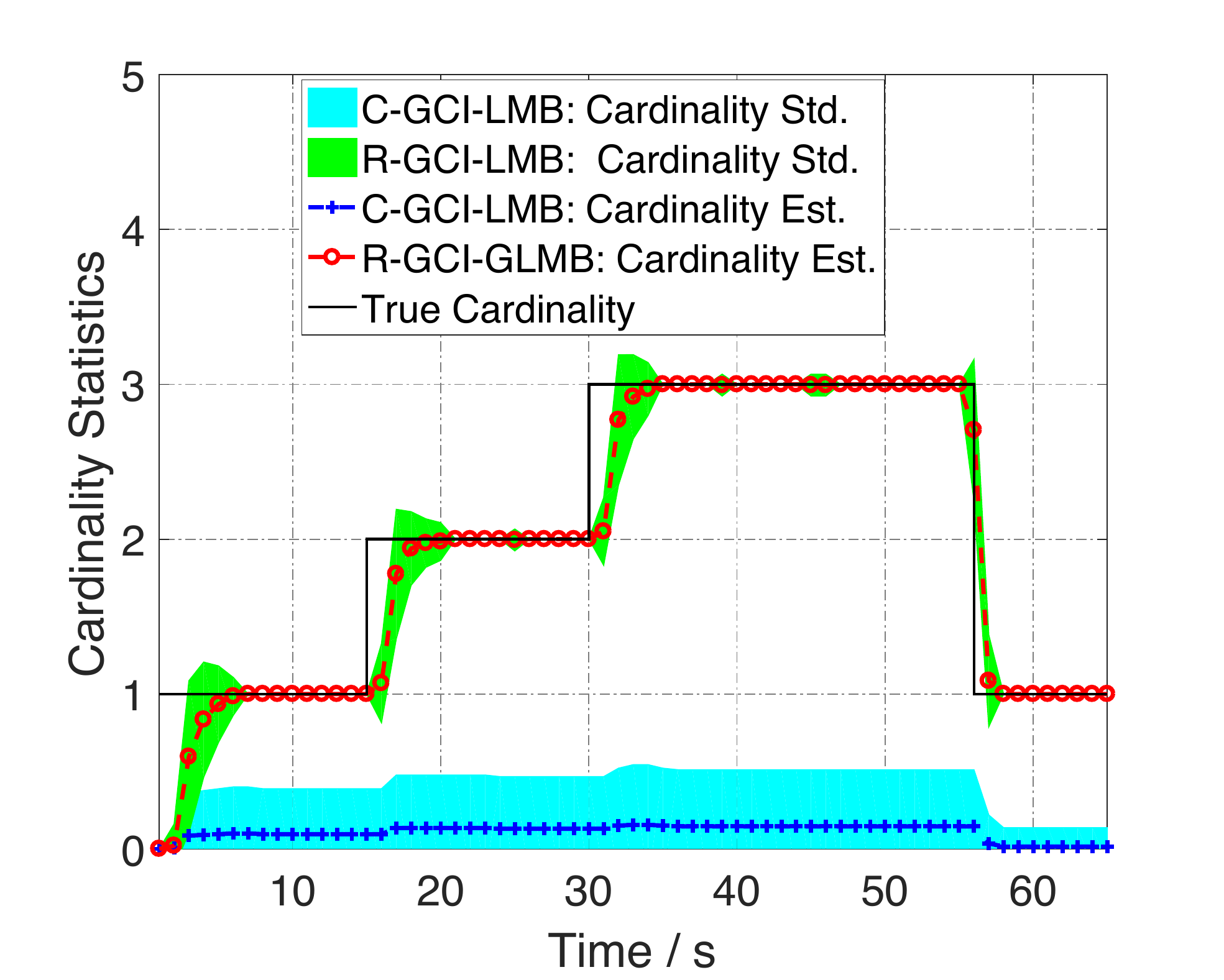}}
		\centerline{\small{(a)} }\medskip
	\end{minipage}
	\hfill
	\begin{minipage}[!htb]{0.49\linewidth}
		\centering
		\centerline{\includegraphics[width=4.94cm]{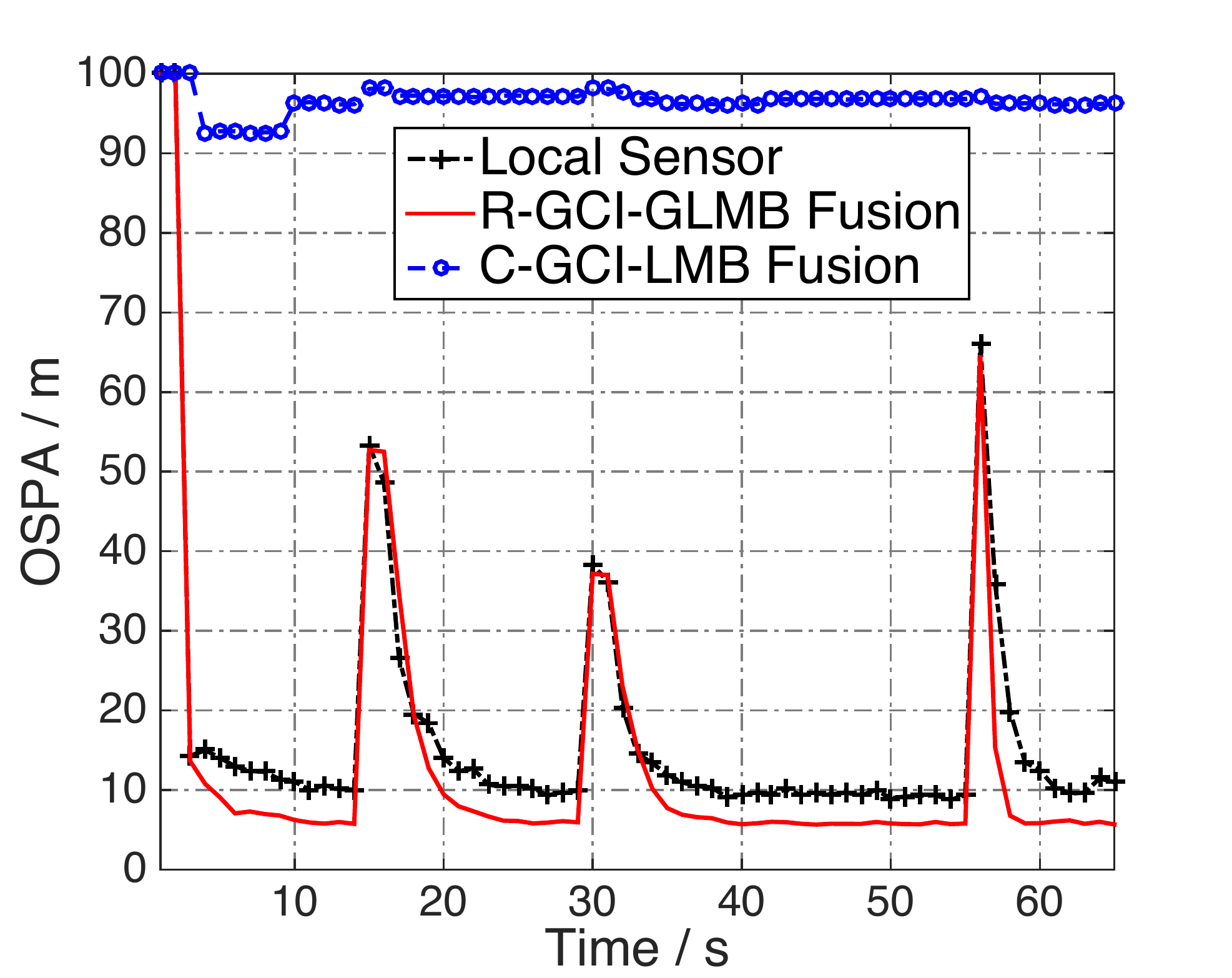}}
		\centerline{\small{(b)}}\medskip
	\end{minipage}
	\caption{Tracking performances of local sensor filter, R-GCI-GLMB and C-GCI-LMB fusion algorithms in scenario 1 in presence of an adaptive birth process: (a) cardinality statistics, (b) OSPA errors. \label{fig:ABP}}
\end{figure}
\begin{figure}
	\begin{minipage}[!htb]{0.49\linewidth}
		\centering
		\centerline{\includegraphics[width=4.95cm]{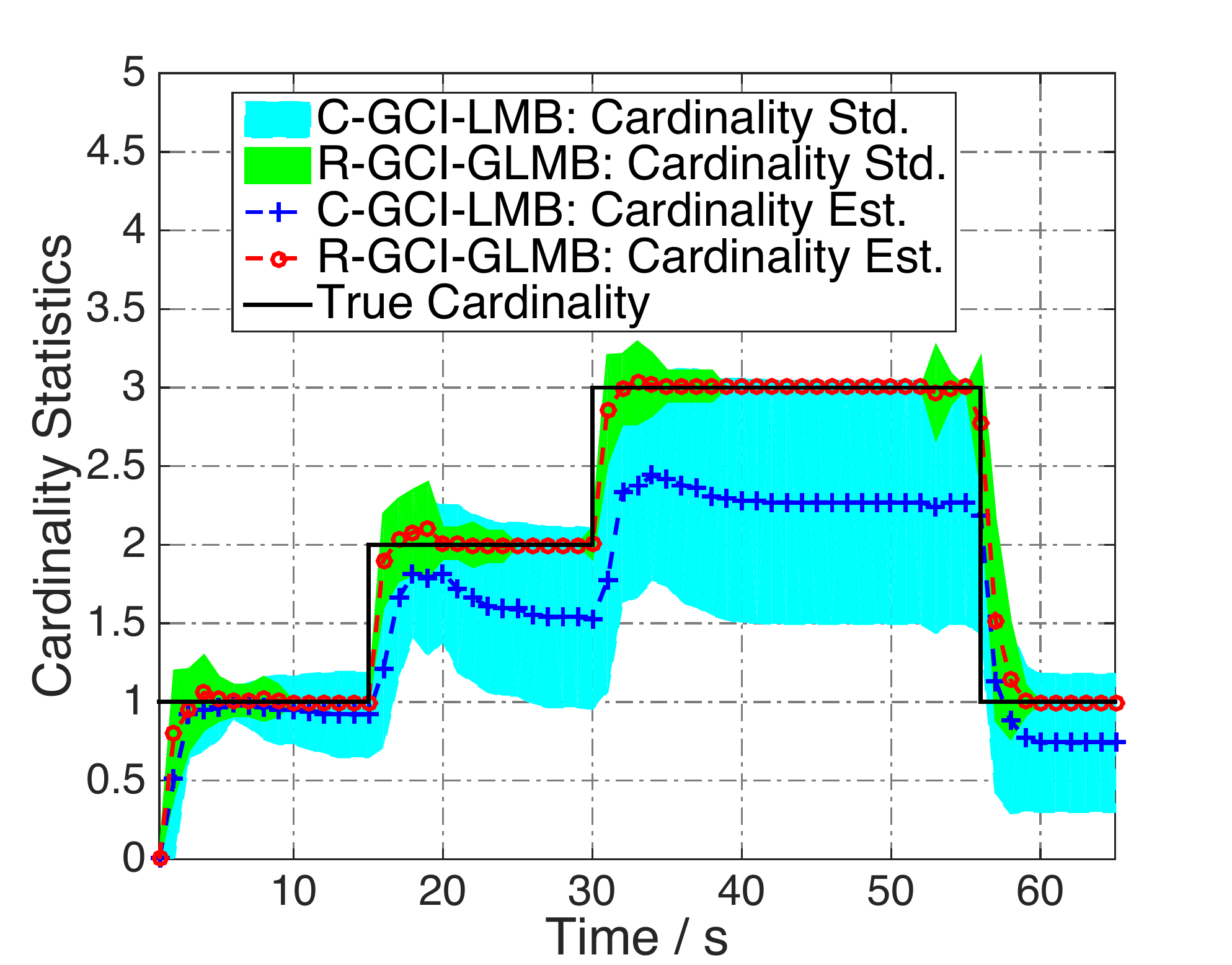}}
		\centerline{\small{(a)} }\medskip
	\end{minipage}
	\hfill
	\begin{minipage}[!htb]{0.49\linewidth}
		\centering
		\centerline{\includegraphics[width=4.95cm]{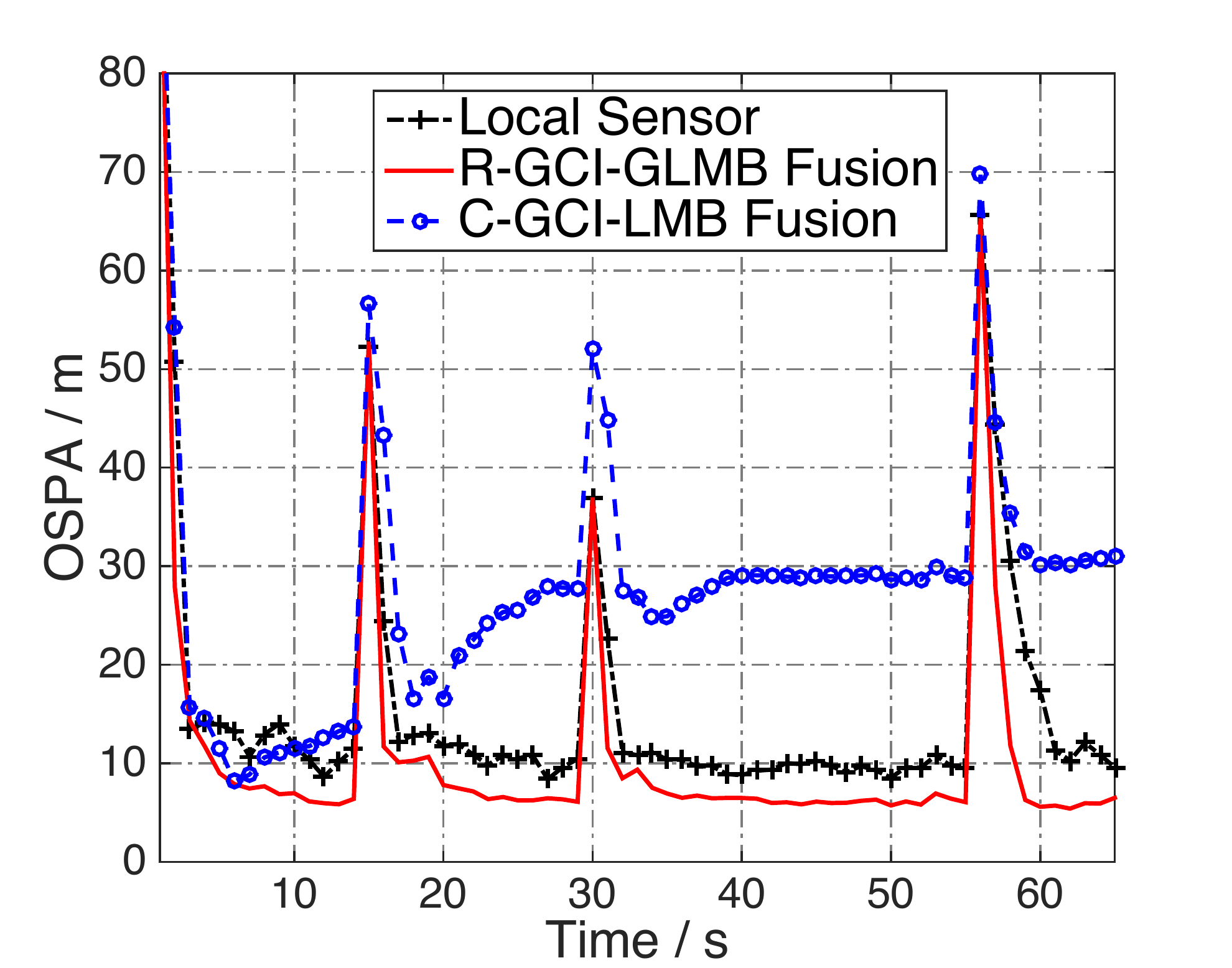}}
		\centerline{\small{(b)}}\medskip
	\end{minipage}
	\caption{Tracking performances of the local sensor filter, R-GCI-GLMB and C-GCI-LMB fusion algorithms in scenario 1 in presence of  a non-adaptive birth process designed based on prior information: (a) cardinality statistics, (b) OSPA errors. \label{fig:PBP}}
\end{figure}

\noindent\textit{\underline{Experiment 2}}:$\,$ In this experiment, the performance of R-GCI-GLMB and C-GCI-LMB fusions are compared in presence of a non-adaptive birth model that is based on prior information -- PBP model~\cite{delta_GLMB}. The comparisons in terms of the cardinality   statistics and OSPA errors between R-GCI-GLMB fusion and C-GCI-LMB fusion  are presented in Figs.~\ref{fig:PBP}(a) and (b).

It can be seen from  Fig.~\ref{fig:PBP}(a) that  cardinality estimates of  the C-GCI-LMB fusion are biased with  large standard deviations, while  the cardinality estimates returned by the R-GCI-GLMB fusion are much more accurate with  less deviations (high level of confidence).  From the results shown in Fig.~\ref{fig:PBP}(b), we observe that whenever an object appears or disappears, the OSPA errors returned by all methods sharply increase. Another observation is  while the local filter and the proposed R-GCI-GLMB fusion handle the change well (and their tracking errors gradually retract after every jump), the  classical GCI fusion does not survive the impact of a sudden change in number of objects (especially for the births) and its error increases. A third observation is that our R-GCI-GLMB fusion significantly outperforms the other methods in terms of the
 OSPA error.

%\begin{table}[!h]
%\renewcommand{\arraystretch}{1.5}
%	\caption{{\color{red}{(newly added)}} Average standard deviations (m) of OSPA Errors \label{tab_scenario_1}}
%\begin{center}
%		\footnotesize
%\begin{tabular*}{0.48\textwidth}{@{\extracolsep{\fill}}c| c c c }
%		\hline\hline
%	Algorithm & R-GCI-GLMB & C-GCI-LMB &Local Sensor \\
%		\hline
%		 Std. (ABP)& 2.7362          & 15.0260          & 9.8481           \\
%		 Std. (PBP)& 4.2579 & 23.8633 & 8.7811 \\
%		\hline
%\end{tabular*}
%\label{tabone}
%	\normalsize
%\end{center}
%\end{table} 
The above observations are in line with the result of Example 1 and the mathematical analysis  presented earlier. Each time a new object is born,  the average disparity between  the label information embedded in various labeled posteriors is enhanced because one more object may have different estimated labels in different sensors, which leads to a larger label  inconsistency indicator $d_G(\bPi)$   in turn resulting in a degraded  performance for C-GCI-LMB fusion.
%The above observations are in line with the result of Example 1 and the mathematical analysis that was presented earlier in the paper. Every sudden change in the number of targets (a birth or death) clearly leads to an increase in the level of confidence represented by conditional multi-label distributions $\varpi_s(\cdot)$. This intuitively means that these densities become in a sense ``wider'' and their products and its statistical expectation over $\mathbb{X}^n$ become smaller, hence leading to a larger gap distance $d_G$. From equation~\eqref{dG_def}, a larger gap distance means larger GCI divergence for labeled densities which in turn means a degraded performance with fusion of those densities.
%\begin{figure}
%	\begin{minipage}[!htb]{0.495\linewidth}
%		\centering
%		\centerline{\includegraphics[width=4cm]{ABP_GCI_GMB_vs_LMB_OSPA_v1.pdf}}
%		\centerline{\small{(a)} }\medskip
%	\end{minipage}
%	\hfill
%	\begin{minipage}[!htb]{0.495\linewidth}
%		\centering
%		\centerline{\includegraphics[width=4cm]{PBP_GCI_GMB_vs_LMB_OSPA_v1.pdf}}
%		\centerline{\small{(b)}}\medskip
%	\end{minipage}
%	\caption{OSPA errors of local sensor filter, R-GCI-GLMB and C-GCI-LMB fusion algorithms in scenario 1 in presence of (a) an adaptive birth process, and (b) a non-adaptive birth process designed based on prior information. \label{fig:tracks}}
%\end{figure}

\subsection{Scenario 2}
In order to further demonstrate the performance of the proposed R-GCI-GLMB fusion in challenging scenarios, a sensor network scenario with three sensors and eight objects is considered as shown in Fig.~\ref{fig:tracks2}(a).  The objects appear and disappear at different times as listed in Table~\ref{tab1}.
\begin{figure}[h!]
\begin{minipage}
{0.49\linewidth}
  \centering
%   \centerline{\includegraphics[width=4.5cm]{the_OSPA_of_two_target1.eps}}
\centerline{\includegraphics[width=4.93cm]{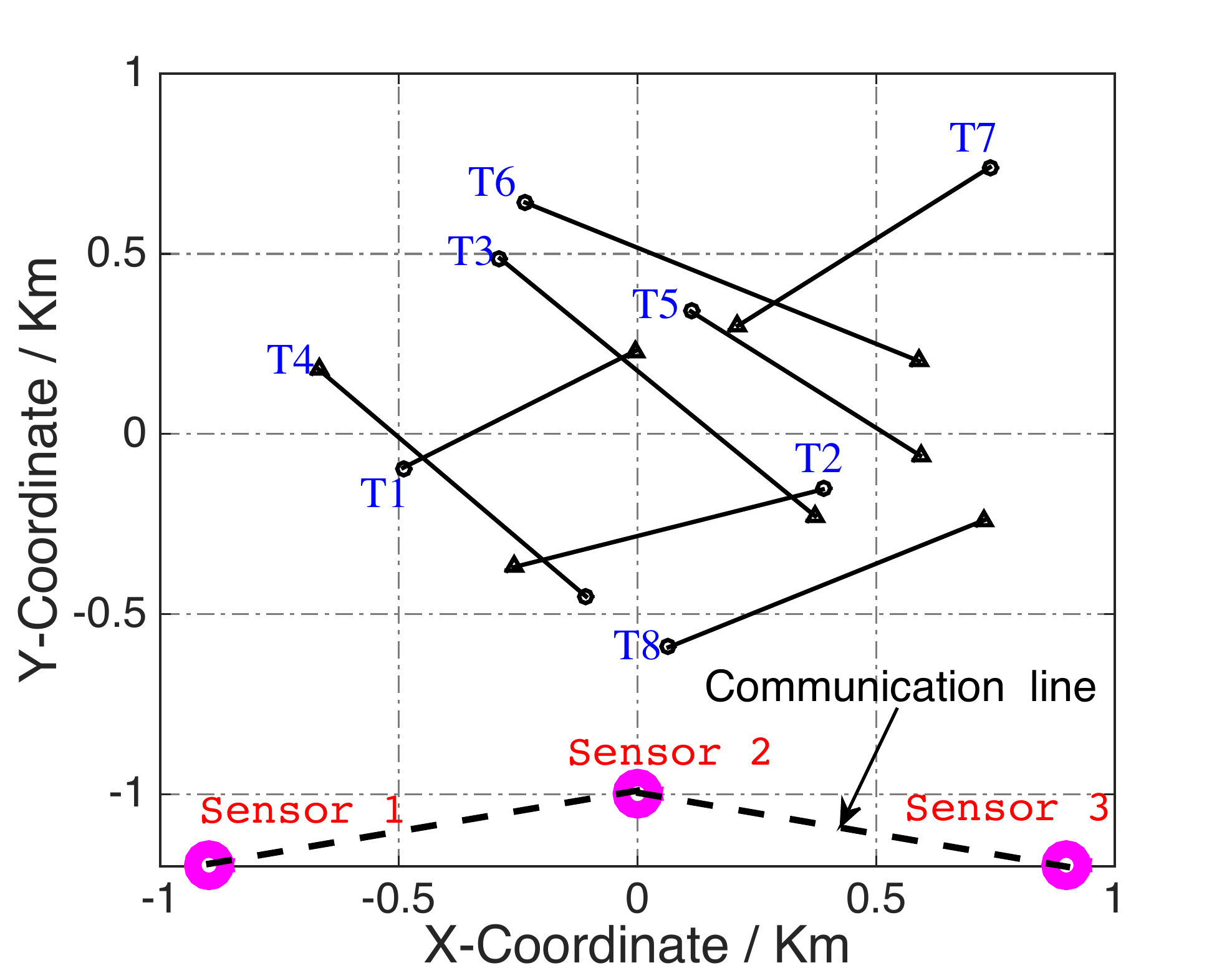}}
%  \vspace{1.5cm}
  \centerline{\small{(a)} }\medskip
\end{minipage}
\hfill
\begin{minipage}
{0.49\linewidth}
  \centering
\centerline{\includegraphics[width=4.94cm]{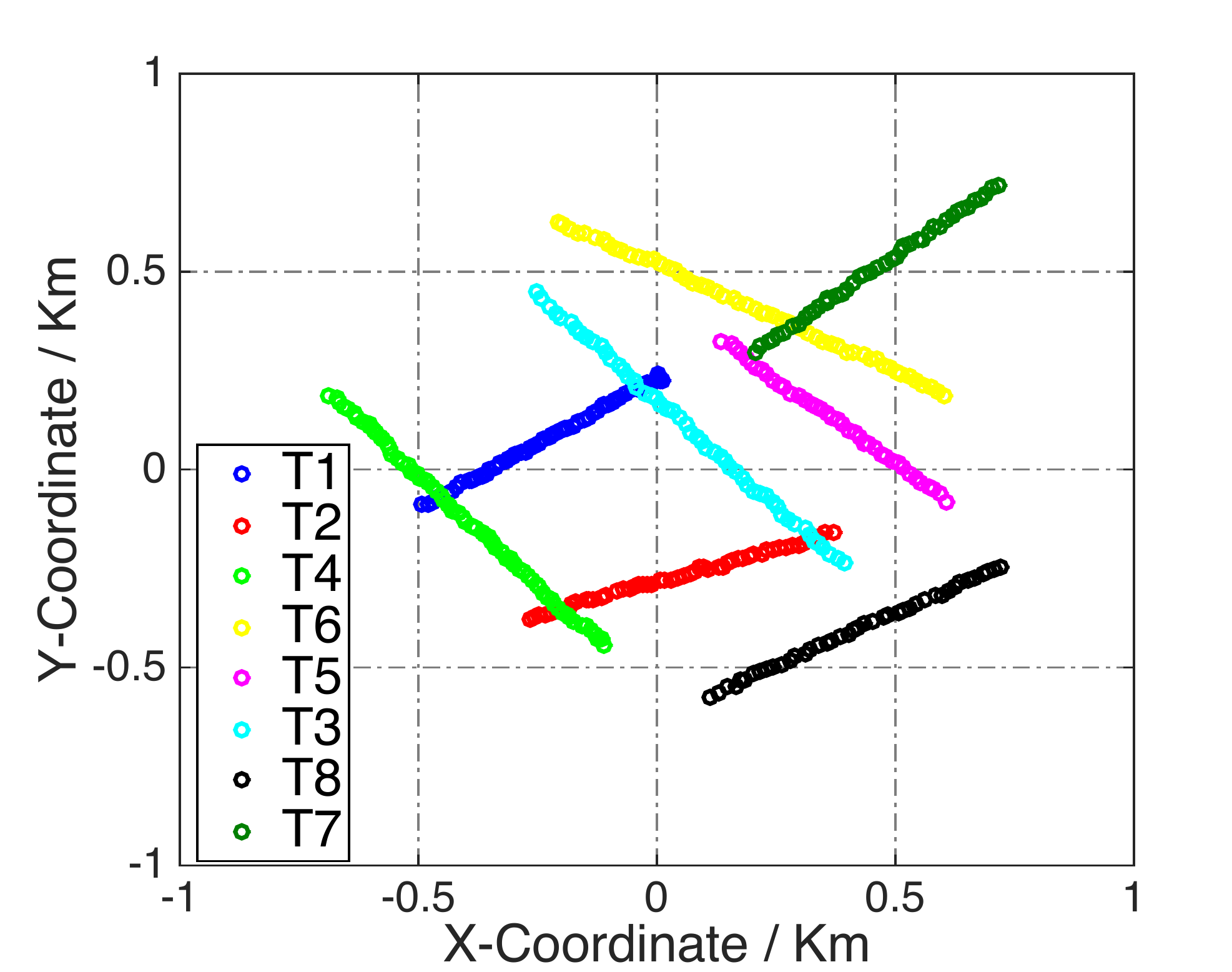}}
%  \vspace{1.5cm}
  \centerline{\small{(b)}}\medskip
\end{minipage}
\caption{(a) Scenario 2: a distributed  sensor network involving three sensors tracking eight objects on a two dimensional surveillance
region. (b) Tracking results of the R-GCI-GLMB fusion algorithm with different colors denoting different identities of objects under $P_{D,k}=0.98$.}
\label{fig:tracks2}
\end{figure}

%\vspace{-3mm}
\begin{table}[htb]
	\caption{Target birth and death times in Scenario 2.\label{tab1}}
	\centerline
	{
		\footnotesize
	\begin{tabular}{cccccc}
		\hline
		Target & Birth & Death & Target & Birth & Death \\
		\hline
		T1 & 1\,s & 56\,s & T5 & 25\,s & 66\,s \\
		T2 & 1\,s & 56\,s & T6 & 25\,s & 81\,s \\
		T3 & 10\,s & 66\,s & T7 & 56\,s & $>$100\,s \\
		T4 & 10\,s & 81\,s & T8 & 56\,s & $>$100\,s \\
		\hline
	\end{tabular} 
	\normalsize
	}
\end{table}
The performance of the R-GCI-GLMB fusion is compared to the GCI-CPHD fusion~\cite{Battistelli,Uney-2}. The CPHD filter and the LMB filter are chosen as the local filter  for GCI-CPHD fusion and  R-GCI-GLMB fusion, respectively. Since the objects appear at unknown positions, LMB filters
use an ABP introduced in~\cite{delta_GLMB} and CPHD filter  uses the adaptive birth distribution introduced in~\cite{Ristic_PHD}.  Pruning  and merging thresholds for GM implementations of  local CPHD filter and  GCI-CPHD fusion algorithms are chosen as  $\gamma_p=10^{-5}$ and  $\gamma_m=4$, respectively, and the maximum number of Gaussian components is $N_{\max}=30$. For  the GM implementation of the R-GCI-GLMB fusion, the parameters  are set to be the same as Scenario 1.  The duration of this scenario is $T= 100$\,s.

The sensors have the same detection parameters and each sensor can only exchange posteriors with its neighbour(s). Therefore, sensors 1 and 3 perform  fusion with two posteriors from sensor 2 and their local filters, and sensor 2 performs fusion with three posteriors from sensor 1, sensor 3 and the local filter by sequentially applying the pairwise fusion twice.

Fig.~\ref{fig:tracks2}(b) shows the estimated tracks returned by R-GCI-GLMB fusion for a single run under $P_{D,k}=0.98$.  It can be seen that  R-GCI-GLMB fusion performs
accurately and consistently for the entire scenario in the sense
that it maintains locking on all tracks,  estimates
object positions accurately, and  recognizes object identities correctly. Fig.~\ref{fig.ospa_cardinality}(a) presents the cardinality estimates and the corresponding standard deviations  returned by R-GCI-GLMB  fusion and GCI-CPHD fusion algorithms at sensor 2 under $P_{D,k}=0.98$. It shows that cardinality estimates given by R-GCI-GLMB fusion are more accurate with less variations (higher level of confidence) than GCI-CPHD fusion. Note that since
two objects are born at time 56\,s and two objects die at time 56\,s (as shown in Table I), the cardinality curves have a notch at  time 56\,s.

Under $P_{D,k}=0.98$, the OSPA errors for tracking results returned by the algorithms are shown in Fig.~\ref{fig.ospa_cardinality}(b). Further, we compute the corresponding standard deviations of OSPA errors and average the post-transient values over 200 MC runs and 100 time steps, and the results are provided in Table II. They demonstrate the performance difference between the R-GCI-GLMB  and GCI-CPHD fusions at sensor 2. 
%With the R-GCI-GLMB, 
%track initiations are much faster than the GCI-CPHD when new targets are born at times 10\,s, 25\,s and 56\,s. Moreover, 
OSPA errors of the R-GCI-GLMB fusion are significantly lower than the GCI-CPHD fusion with lower standard deviations after each transient. Moreover, when objects die at time 66\,s and 81\,s,   OSPA error of the R-GCI-GLMB filter  retracts to a stable value  much faster than the GCI-CPHD fusion method.  
\begin{figure}[h]
\begin{minipage}{0.49\linewidth}
  \centering
%   \centerline{\includegraphics[width=4.5cm]{the_OSPA_of_two_target1.eps}}
  \centerline{\includegraphics[width=4.94cm]{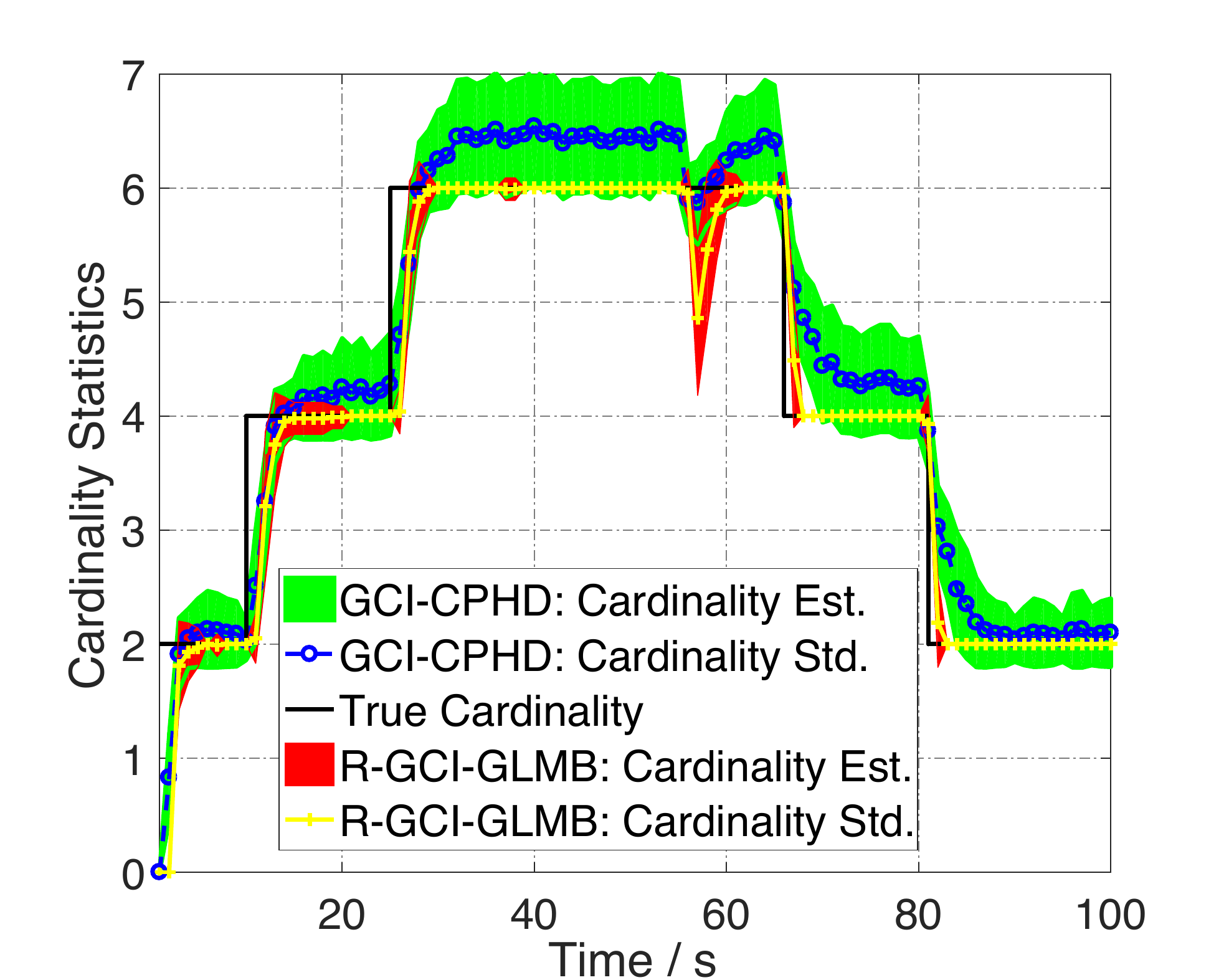}}
%  \vspace{1.5cm}
  \centerline{\small{(a)} }\medskip
\end{minipage}
\hfill
\begin{minipage}{0.49\linewidth}    \centering
  \centerline{\includegraphics[width=4.94cm]{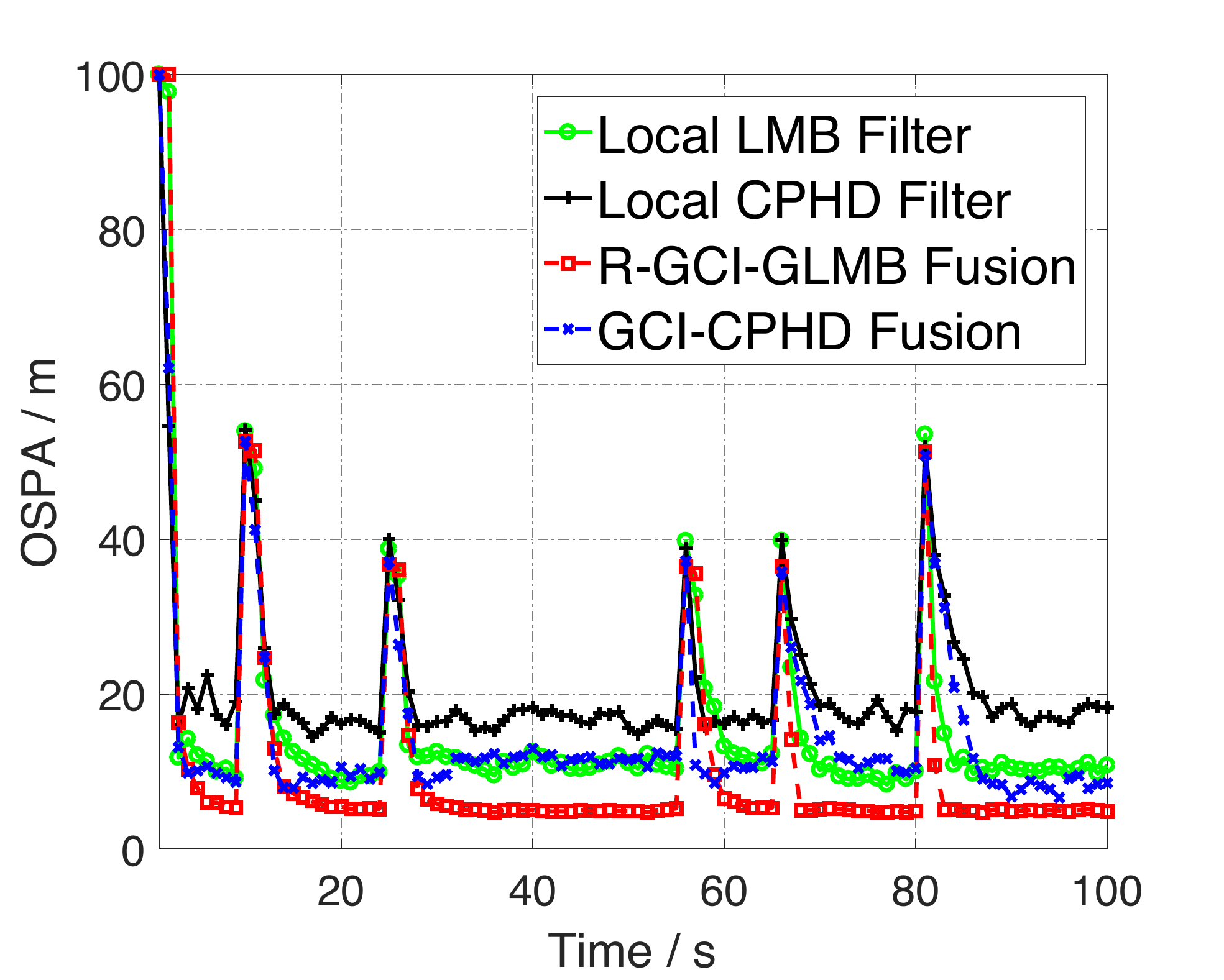}}
%  \vspace{1.5cm}
\centerline{\small{(b)}}\medskip

\end{minipage}

\caption{Tracking performances of R-GCI-GLMB and GCI-CPHD fusion algorithms under $P_{D,k}=0.98$ in Scenario 2: (a) cardinality statistics (b) OSPA errors.}
\label{fig.ospa_cardinality}
\end{figure}

\begin{table}[h]
\renewcommand{\arraystretch}{1.5}
	\caption{Average standard deviations of OSPA errors under $P_{D,k}=0.98$.\label{tab_scenario_2}}

\begin{center}
	\footnotesize
\begin{tabular*}{0.48\textwidth}{@{\extracolsep{\fill}}c| c c c c}
\hline\hline
Algorithm	&	 R-GCI-GLMB          & GCI-CPHD         &Local LMB            &Local CPHD \\
		\hline
	\text{Std. (m)}       &      1.673 &8.083     &6.838            &10.580         \\
		\hline
\end{tabular*}
\label{tabone}
	\normalsize
\end{center}
\end{table}

To assess the computational efficiency of the  algorithms, the average execution times of the R-GCI-GLMB (adopting the efficient implementation strategy) and GCI-CPHD fusions under $P_{D,k}=0.98$ are depicted in Fig. \ref{time_scenario_2}. It can be seen that the execution time of  the R-GCI-GLMB fusion is only slightly longer than the GCI-CPHD fusion with the R-GCI-GLMB fusion providing the enhanced performance (as demonstrated previously) and also automatically accounts for track labelling.

Further, we assess the performance of  R-GCI-GLMB and GCI-CPHD fusion methods under different $P_{D,k}$ values  in terms of the averaged post-transient values of OSPA errors (over 200 MC runs and 100 time steps)  as shown in Table \ref{PD_scenario_2}. Not surprisingly, while  the performances of both algorithms degrade as the $P_{D,k}$ value decreases,  the  R-GCI-GLMB fusion performs remarkably better than the GCI-CPHD fusion under each $P_{D,k}$ value with the performance difference  stable. 

%Note that the R-GCI-GLMB performs worse than GCI-PHD when objects die at 65\,s and  80\,s because the local LMB filter performs worse than the local PHD filter when handling target deaths.

\begin{figure}[h]
\centering
\includegraphics[width=6cm]{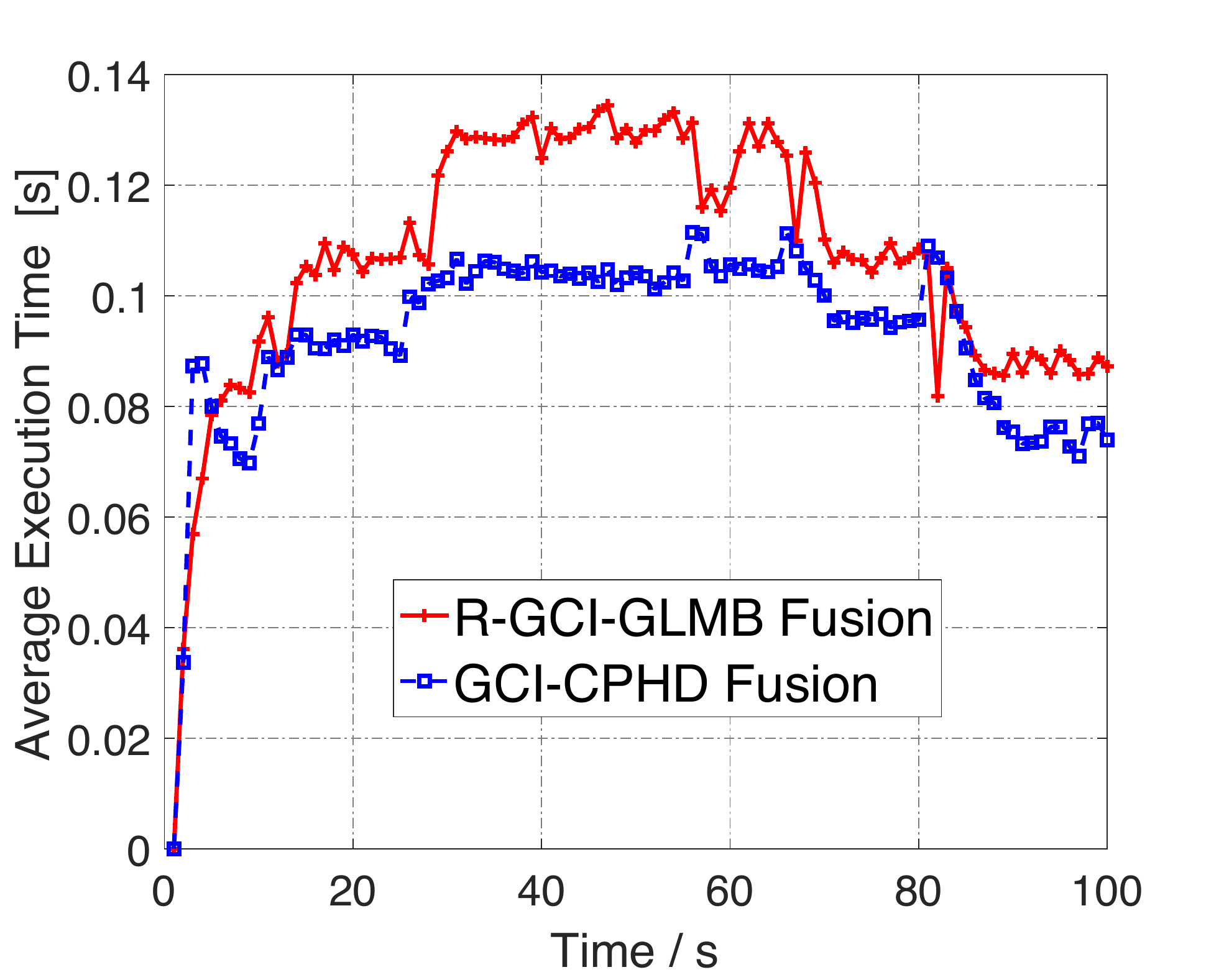}
\caption{Average execution times for R-GCI-GLMB and GCI-CPHD fusion algorithms under $P_{D,k}=0.98$ in Scenario 2.}
\label{time_scenario_2}
\end{figure}
\begin{table}[!h]
\renewcommand{\arraystretch}{1.5}
	\caption{Average OSPA errors (m) of R-GCI-GLMB and GCI-CPHD fusions under different $P_{D,k}$ values.}\label{PD_scenario_2}
\begin{center}
		\footnotesize
\begin{tabular*}{0.48\textwidth}{@{\extracolsep{\fill}}c| c c c }
		\hline\hline
%$P_{D,k}$                        & 0.99                & 0.88     &  0.78  \\
%\midrule
% R-GCI-GLMB               &  4.949               & 5.652   &  6.178
%\\
%GCI-PHD              &  12.777            & 24.489  &  29.912
%\\
% R-GCI-GLMB               &  4.949               & 5.652   &  6.178
%\\
%GCI-PHD              &  12.777            & 24.489  &  29.912
%\\
$P_{D,k}$                        & 0.98                & 0.88     &  0.78  \\
%\midrule
\hline
 R-GCI-GLMB               &  5.160              & 6.026   &  7.688
\\
GCI-CPHD              &  11.079            & 12.411  &  13.659
\\
\hline
\end{tabular*}
\label{tabone}
	\normalsize
\end{center}
\end{table}

To demonstrate how the performance advantage gained from sensor fusion increases with the number of sensors, we compute the OSPA errors returned by R-GCI-GLMB fusion and average the post-transient values over 200 MC runs and 100 time steps. Table~\ref{tabone} shows the recorded values in presence of one, two and three sensors.  The results demonstrate the  efficacy of the proposed sensor fusion algorithm in the sense that estimation accuracy improves with more sensors.
\begin{table}[!htb]
\footnotesize
\renewcommand{\arraystretch}{1.2}
\setlength{\abovecaptionskip}{1pt}
\setlength{\belowcaptionskip}{-3pt}\caption{Average OSPA Errors \textit{vs}  No. of Sensors ($P_{D,k}=0.98$)}
\begin{center}
\begin{tabular*}{0.4\textwidth}{@{\extracolsep{\fill}}c c c c}
\toprule
%&\multicolumn{3}{c}{Number of sensors}\\
%\cmidrule(r){2-4}
 Number of sensors                         & One                  &   Two   &  Three  \\
\midrule
OSPA Errors (m)                &  10.711               & 5.948     &  5.160
\\
%OSPA of MS-PHD (m)                         &  3.1242              &    0.9530  &   0.799\\
\bottomrule
\end{tabular*}
\end{center}
\label{tabone}
\end{table}
These results demonstrate that the two-step approximation used in the derivation of GCI fusion with GMB distributions is reasonable, and the significant enhancement in performance (in terms of OSPA errors) also verifies the robustness and effectiveness of the R-GCI-GLMB fusion devised and presented in this work. 
\section{Conclusion}\label{chp:7}
In this paper, we addressed the problem of distributed multi-object tracking with labeled set filters based on generalized Covariance Intersection (GCI). Firstly, we showed that the performance of GCI fusion with labeled multi-object densities is highly sensitive to inconsistencies between label information from the local labeled posteriors. We provided a mathematical analysis from the perspective of Principle of Minimum Discrimination Information and yes-object probability. Secondly, inspired by the analysis, a novel and general solution was proposed for distributed fusion of labeled multi-object posteriors that is robust to label inconsistencies between different sensor node posteriors.  Thirdly, for the case of fusing generalized labeled multi-Bernoulli (GLMB) filter family  including the GLMB, $\delta$-GLMB, marginalized $\delta$-GLMB and labeled multi-Bernoulli (LMB) filters, we formulated the robust fusion solution. Simulation results for Gaussian mixture (GM) implementation demonstrated the  robustness and effectiveness of the proposed fusion algorithms in  challenging tracking scenarios.
\appendices
\section{Proof of Proposition 2}
According to (\ref{unlabel-marginal}), the unlabeled version of a GLMB density of form (\ref{GLMB}) is distributed according to

\begin{equation}\label{label2unlabel-GLMB}
{\small{
\begin{split}
\!\!\!&\!\!\!\pi(\!\{x_1,\!\cdots\!,x_n\}\!)\!=\!\!\!\sum_{(\ell_1,\cdots,\ell_n)\in\mathbb{L}^n}\!\!\bpi(\!\{(x_1,\ell_1),\!\cdots\!,(x_n,\ell_n)\}\!)\\
\!\!=&\!\!\sum_{(\ell_1,\cdots,\ell_n)\in\mathbb{L}^n}\sum_{c\in\mathbb{C}}w^{(c)}(\{\ell_1,\cdots,\ell_n\})\prod_{i=1}^{n}p^{(c)}(x_{i},\ell_i)\\
\!\!=&\sum_{\sigma}\sum_{I\in\mathcal{F}_n(\mathbb{L})}\sum_{c\in\mathbb{C}}w^{(c)}(I)\prod_{i=1}^{n}p^{(c)}(x_{\sigma(i)},I^v(i)).
\end{split}}}
\end{equation}
where $\sigma$ denotes one permutation of $I$, $\sigma(i)$ denotes the $i$th element of the permutation, and $I^v$ denotes a vector constructed by sorting the elements of the set $I$.

Let 
\begin{equation}\label{para-GMB-GLMB}
\begin{split}
w^{(I,c)}&\triangleq w^{(c)}(I), I\in\mathcal{F}(\mathbb{L})\\
p^{(c),\ell}(x)&\triangleq p^{(c)}(x,\ell), \ell\in\mathbb{L}.
\end{split}
\end{equation}
Equation (\ref{label2unlabel-GLMB}) can be further represented as
\begin{equation}\label{GMB-GLMB}
{\small{\begin{split}\notag
\pi(\{&x_1,\cdots,x_n\})\!=\!\!\sum_{\sigma}\!\sum_{(I,c)\in\mathcal{F}_n(\mathbb{L})\times\mathbb{C}}\!\!w^{(I,c)}{\prod}_{i=1}^{n}p^{(c),I^v(i)}(x_{\sigma(i)})
\end{split}}}
\end{equation}
\section{Proof of Proposition 3}
According to (\ref{unlabel-marginal}), the unlabeled version of an LMB density of form (\ref{LMB}) is distributed according to
\begin{equation}\label{label2unlabel-LMB}
{\small{\begin{split}
\!\!\!&\!\!\!\pi(\{x_1,\!\cdots\!,x_n\})\!=\!\!\!\sum_{(\ell_1,\cdots,\ell_n)\in\mathbb{L}^n}\bpi(\!\{(x_1,\ell_1),\!\cdots\!,(x_n,\ell_n)\}\!)\\
\!\!\!=&\!\sum_{(\ell_1,\cdots,\ell_n)\in\mathbb{L}^n}w(\{\ell_1,\cdots,\ell_n\})\prod_{i=1}^{n}p(x_{i},\ell_i)\\
\!\!\!=&\sum_{\sigma}\sum_{I\in\mathcal{F}_n(\mathbb{L})}w(I)\prod_{i=1}^{n}p(x_{\sigma(i)},I^v(i)).
\end{split}}}
\end{equation}
Let
\begin{equation}\label{para-GMB-GLMB}
\begin{split}
w^{(I)}&\triangleq w^{(I)}, I\in\mathcal{F}(\mathbb{L})\\
p^{(c),\ell}(x)&\triangleq p^{(\ell)}(x), \ell\in\mathbb{L}.
\end{split}
\end{equation}
Equation (\ref{label2unlabel-LMB}) can be further represented as
\begin{equation}\label{GMB-GLMB}
{\small{\begin{split}
  \pi(\{&x_1,\cdots,x_n\})=\sum_{\sigma}\sum_{I\in\mathcal{F}_n(\mathbb{L})}w^{(I)}\prod_{i=1}^{n}p^{(I^v(i))}(x_{\sigma(i)})
\end{split}\notag}}
\end{equation}
\section{Proof of Proposition 4}
Combination of (\ref{fuse-1}) and (\ref{app-fused-label}) yields  
\begin{equation}\label{label-density}
{\small{\begin{split}
&\!\!\!\bpi_\omega(\{(x_1,\ell_1),\cdots,(x_n,\ell_n)\})\!=\!\\
%&=\frac{[\widetilde w(\{\ell_1,\cdots,\ell_n\})]^{\omega_1} \prod_{i=1}^n [p(x_i,\ell_i)]^{\omega_1}}{\sum_{(\ell_1,\cdots,\ell_n)\in\mathbb{L}^n}[\widetilde w(\{\ell_1,\cdots,\ell_n\})]^{\omega_1}\widetilde \prod_{i=1}^n [p(x_i,\ell_i)]^{\omega_1}}\cdot \\
%&\,\,\,\,\,\,\,\frac{\prod_{s={1,2}}[\widetilde \pi_s(\{x_1,\cdots,x_n\})]^{\omega_s}}{\int \prod_{s={1,2}}[\pi_\omega(X)]^{\omega_s}\delta X}\\
&\!\!\!\!\frac{\left[\overline w_{1}(\{\ell_1,\!\cdots\!,\ell_n\})\right]^{\omega_1}\prod_{i=1}^n\left[\overline p_{1}^{(\ell_i)}(x_i)\right]^{\omega_1}}{\sum_{(\ell_1,\cdots,\ell_n)\in\mathbb{L}_1^n} \left[\overline w_{1}(\{\ell_1,\!\cdots\!,\ell_n\})\right]^{\omega_1}\prod_{i=1}^n\left[\overline p_{1}^{(\ell_i)}(x_i)\right]^{\omega_1}} \times \\
&\!\!\!\!\frac{1}{C}\prod_{s=1,2} {\sum}_{\sigma_s}{\sum}_{\mathcal{I}_s\in\! \mathcal{F}_n(\mathbb{I}_s)}{\left[\widetilde w_s^{(\mathcal{I}_s)}\right]}^{\omega_s}{\left[{\prod}_{i=1}^{n}\widetilde p_s^{(\mathcal{I}_s^v(i))}\!(x_{i})\right]}^{\omega_s}.\\
%&=\frac{1}{C}[\widetilde w(\{\ell_1,\cdots,\ell_n\})]^{\omega_1} \prod_{i=1}^n [p(x_i,\ell_i)]^{\omega_1}\cdot\\
%&\,\,\,\,\,\,\sum_{\sigma_2}\sum_{\mathcal{I}_2\in\mathcal{F}_n(\mathbb{I}_2)}{\left(Q^{\mathcal{I}_2}\right)}^{\omega_2}{\left(\prod_{i=1}^{n}p_2^{\mathcal{I}_2^v(i)}\!(x_{i})\right)}^{\omega_2}\\
%&=\sum_{\tau\in \mathcal{T}(\{\ell_1,\cdots,\ell_n\})} w_\omega^{(\{\ell_1,\cdots,\ell_n\},\tau)}\prod_{i=1}^n p^{(\tau),\ell_i}(x_i)
\end{split}}}
\end{equation}
As shown in Proposition 2, the unlabeled version of GLMB density in (\ref{GLMB-2}) is a GMB density with $\mathbb{I}_{s_0}=\mathbb{L}_{s_0}$ and $\Phi_{s_0}=\mathbb{C}_{s_0}$. Hence,  (\ref{r_s0}) and (\ref{p_s0}) (with $s_0=1$) can be   rewritten as

\begin{align}
\label{-r}\overline r^{(\ell)}_1&=\sum_{\mathcal{I}\in\mathcal{F}(\mathbb{I}_1)}\sum_{\phi\in\Phi_1}1_\mathcal{I}(\ell)w_1^{(\mathcal{I},\phi)}=\widetilde r_1^{(\ell)}
\end{align}

\begin{align}
\label{-p}\!\!\overline p^{(\ell)}_1(x)&\!=\!\frac{1}{\widetilde{r}_1^{(\ell)}}\!\sum_{\mathcal{I}\in\mathcal{F}(\mathbb{I}_1\!)}\!\sum_{\phi\in\Phi_1}\!1_\mathcal{I}(\ell)w_1^{(\mathcal{I},\phi)}
p_1^{(\phi),\ell}(x)\!=\!\widetilde p_1^{(\ell)}(x)
\end{align}
where $\widetilde r_1^{(\ell)}$ and $\widetilde p_1^{(\ell)}(x)$ are shown in (\ref{r MB}) and (\ref{p MB}) respectively.
As a result, we have the following equality,
\begin{equation}\label{equality}
\begin{split}
&\!\sum_{(\ell_1,\cdots,\ell_n)\in\mathbb{L}_1^n} \!\!\left[\overline w_{1}(\{\ell_1,\cdots,\ell_n\})\right]^{\omega_1}\!{\prod}_{i=1}^n\left[\overline p_{1}^{(\ell_i)}(x_i)\right]^{\omega_1}\!\!=\\
&\!{\sum}_{\sigma}{\sum}_{\mathcal{I}_1\in\! \mathcal{F}_n(\mathbb{I}_1)}{\left[\widetilde w_1^{(\mathcal{I}_1)}\right]}^{\omega_1}{{\prod}_{i=1}^{n}\left[\widetilde p_1^{(\mathcal{I}_1^v(i))}\!(x_{\sigma({i})})\right]}^{\omega_1}\end{split}
\end{equation}
Substitution of (\ref{-r}), (\ref{-p}) and (\ref{equality}) into (\ref{label-density}), we have
\begin{equation}\label{label-density-2}
\begin{split}
&\bpi_\omega(\{(x_1,\ell_1),\cdots,(x_n,\ell_n)\})\\
=&\frac{1}{C}[\widetilde w_1^{(\{\ell_1,\cdots,\ell_n\})}]^{\omega_1} {\prod}_{i=1}^n [\widetilde p_1^{(\ell_i)}(x_i)]^{\omega_1}\times\\
&{\sum}_{\sigma}{\sum}_{\mathcal{I}_2\in\mathcal{F}_n(\mathbb{I}_2)}{\left[\widetilde w_2^{(\mathcal{I}_2)}\right]}^{\omega_2}{{\prod}_{i=1}^{n}\!\left[\widetilde p_2^{(\mathcal{I}_2^v(i))}(x_{\sigma({i})})\right]}^{\omega_2}\\
%&=\sum_{\tau\in \mathcal{T}(\{\ell_1,\cdots,\ell_n\})} w_\omega^{(\{\ell_1,\cdots,\ell_n\},\tau)}\prod_{i=1}^n p^{(\tau),\ell_i}(x_i)
\end{split}
\end{equation}
Substitution of  (\ref{fuse-p}), (\ref{Z}), (\ref{fuse-w}) and (\ref{fuse-w-non}) into (\ref{label-density-2}), and utilizing Definition 5, we can obtain
\begin{small}
\begin{equation}
\begin{split}
&\bpi_\omega(\{(x_1,\ell_1),\cdots,(x_n,\ell_n)\})\\
=&{\sum}_{\tau\in \mathcal{T}(\{\ell_1,\cdots,\ell_n\})}\frac{1}{C} [\widetilde w_1^{(\{\ell_1,\cdots,\ell_n\})}]^{\omega_1} [\widetilde w_2^{(\tau(\{\ell_1,\cdots,\ell_n\})}]^{\omega_2} \times
\\
&{\prod}_{i=1}^n \eta^{(\tau),\ell_i}_\omega  {\prod}_{i=1}^n p_\omega^{(\tau),\ell_i}(x_i)\\
=&{\sum}_{\tau\in \mathcal{T}(\{\ell_1,\cdots,\ell_n\})}   w_\omega^{(\{\ell_1,\cdots,\ell_n\},\tau)} {\prod}_{i=1}^n p_\omega^{(\tau),\ell_i}(x_i)
\end{split}
\end{equation}
\end{small}
Using the definition of $p_\omega^{(\tau)}(\cdot,\ell)$ in (\ref{p_identity}) and $w_\omega^{(\tau)}(I)$ in (\ref{w_identity}), we can obtain (\ref{fused GLMB}).
\bibliographystyle{IEEEtran}
\bibliography{GCI_GMB}
\begin{IEEEbiography}
[{\includegraphics[width=0.8\columnwidth,draft=false]{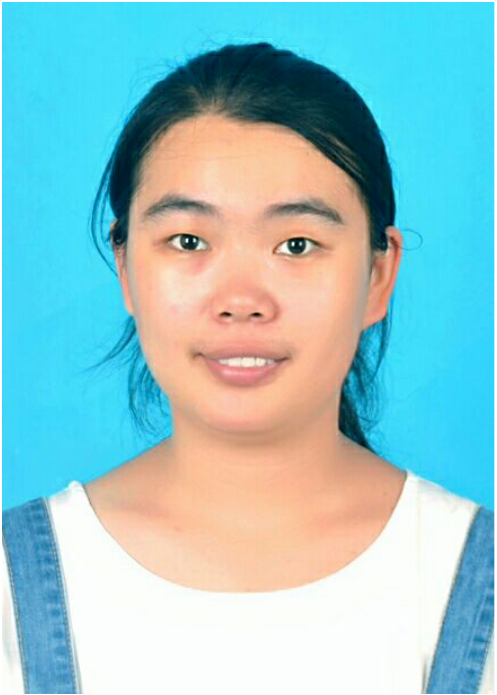}}]{Suqi Li}  is born in 1990.  She received the B.E. degree in electronic engineering from the University of Electronic Science and Technology of China, Chengdu, in 2011. 
Since September 2011, she has been pursuing the Ph.D. degree at the School of Electronic Engineering, University of Electronic Technology and Science of China.  
Currently, she is a visiting student with the Department of Information Engineering, University of Florence, Italy. Her research interests include random finite set, multi-target tracking, nonlinear filtering. 
\end{IEEEbiography}
 \begin{IEEEbiography}
[{\includegraphics[width=0.85\columnwidth,draft=false]{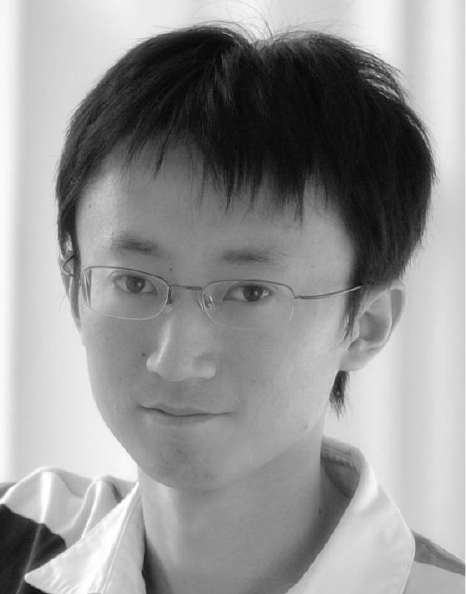}}]{Wei Yi}
 received the B.E. degree in electronic engineering from the University of Electronic Science and Technology of China, Chengdu, in 2006.

Since 2007, he has been pursuing the Ph.D. degree at the School of Electronic Engineering of the University of Electronic Technology and Science of China. 

From March 2010 to February 2012, he was a visiting student in the Melbourne Systems Laboratory, University of Melbourne, Australia. His research interests include particle filtering and target tracking (particular emphasis on multiple target tracking and track-before-detect techniques).

Mr. Yi received the ``Best Student Paper Award'' at the 2012 IEEE Radar Conference, Atlanta, United States and the ``Best Student Paper Award'' at the 15th International Conference on Information Fusion, Singapore, 2012.
\end{IEEEbiography}
 \begin{IEEEbiography}
[{\includegraphics[width=0.8\columnwidth,draft=false]{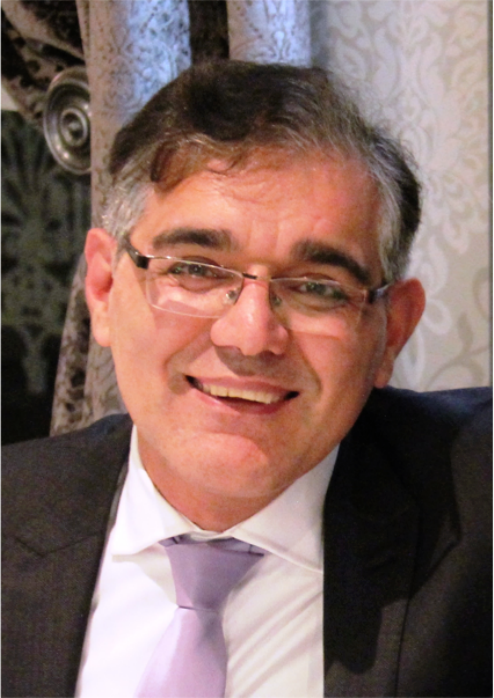}}]{Reza Hoseinnezhad}
received his B.Sc., M.Sc. and Ph.D. degrees in Electronic, Control and Electrical Engineering all from the University of Tehran, Iran, in 1994, 1996 and 2002, respectively. Since 2002, he has held various academic positions at the University of Tehran, Swinburne University of Technology, the University of Melbourne and RMIT University. He is currently a senior lecturer with the School of Aerospace, Mechanical and Manufacturing Engineering, RMIT University, Victoria, Australia. His research is currently focused on development of robust estimation and visual tracking methods in a point process framework.
\end{IEEEbiography}
\begin{IEEEbiography}
[{\includegraphics[width=0.8\columnwidth,draft=false]{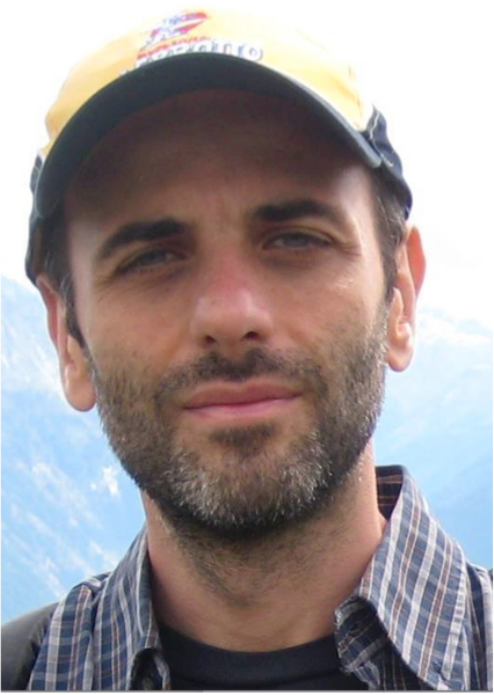}}]{Giorgio Battistelli}
Giorgio Battistelli received the Laurea degree in electronic engineering and the Ph.D. degree in robotics from the University of Genoa, Genoa, Italy, in 2000 and 2004, respectively.
From 2004 to 2006, he was a Research Associate with the Dipartimento di Informatica, Sistemistica e Telematica, University of Genoa. Since 2006, he has been with the University of Florence, Florence, Italy, where he is currently an Associate Professor of automatic control with the Dipartimento di Ingegneria dell’Informazione. His current research interests  include adaptive and learning systems, real-time control reconfiguration, linear and nonlinear estimation, hybrid systems, sensor networks, and data fusion. 
Dr. Battistelli was a member of the editorial boards of the IFAC Journal Engineering Applications of Artificial Intelligence and of the IEEE Transactions on Neural Networks and Learning Systems.
He is currently an Associate Editor of the IFAC Journal Nonlinear Analysis: Hybrid Systems, and a member of the conference editorial boards of IEEE Control Systems Society and the European Control Association.
\end{IEEEbiography}
\begin{IEEEbiography}
[{\includegraphics[width=0.8\columnwidth,draft=false]{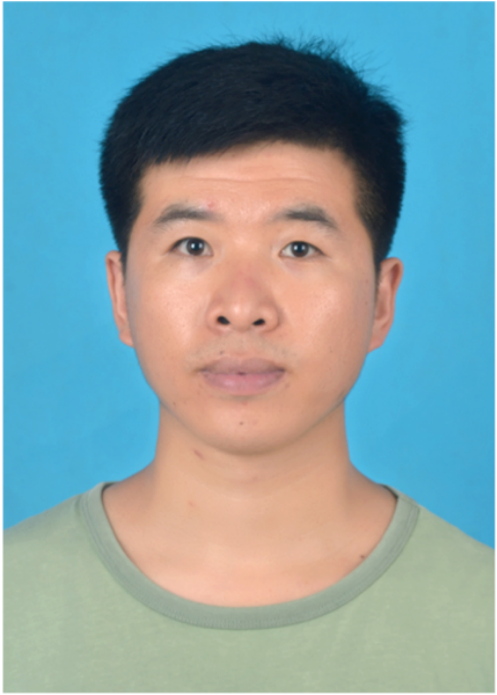}}]{Bailu Wang} received his B.S. degree from the University of Electronic Science and
Technology of China (UESTC) in 2011. He is now working toward his Ph.D. degree
on signal and information processing at UESTC.

From August  2016, he has been a visiting student at University
of Florence, Florence,  Italy. His current research interests include
radar and statistical signal processing, and multi-sensor multi-target fusion.
\end{IEEEbiography}
\begin{IEEEbiography}
[{\includegraphics[width=0.8\columnwidth,draft=false]
{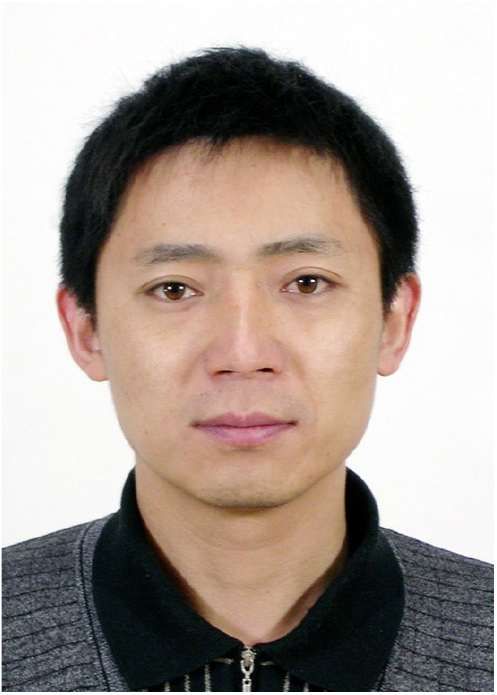}}]{Lingjiang Kong} was born in 1974. He received the B.S., M.S., and Ph.D. degrees from the
University of Electronic Science and Technology of China (UESTC) in 1997, 2000
and 2003, respectively.

From September 2009 to March 2010, he was a visiting researcher with the
University of Florida. 

He is currently a professor with the School of
Electronic Engineering, University of Electronic Science and Technology of
China (UESTC). His research interests include multiple-input multiple-output
(MIMO) radar, through the wall radar, and statistical signal processing.
\end{IEEEbiography}

\end{document}